\documentclass{article}

\usepackage[frenchb,english]{babel}

\usepackage{amsmath}
\usepackage{amssymb}
\usepackage{amsthm}
\usepackage{stmaryrd}
\usepackage{graphicx}
\usepackage{amsfonts}
\usepackage{comment}
\usepackage{dsfont}
\usepackage{xcolor}
\usepackage{mathtools}
\usepackage{bigints}
\usepackage{yhmath}
\usepackage{hyperref}
\usepackage{enumitem} 
\usepackage{subcaption} 
\usepackage[toc,page]{appendix} 
\usepackage{geometry} 
\usepackage{authblk} 
\usepackage{empheq} 
\usepackage{wrapfig}
\usepackage{authblk}

\newtheorem{lemma}{Lemma}
\newtheorem{propo}{Proposition}
\newtheorem{theor}{Theorem}

\newtheorem{defin}{Definition}
\newtheorem{remar}{Remark}

\newcommand{\vertii}[1]{{\left\vert\kern-0.3ex\left\vert #1 
    \right\vert\kern-0.3ex\right\vert}}

\newcommand{\vertiii}[1]{{\left\vert\kern-0.3ex\left\vert\kern-0.3ex\left\vert #1 
    \right\vert\kern-0.3ex\right\vert\kern-0.3ex\right\vert}}
    
\newcommand*\dd{\mathop{}\!\mathrm{d}}

\title{Boltzmann–Grad limit for the inelastic Lorentz gas:\\Part I. Existence, uniqueness, and rigorous derivation\\via weak convergence}

\author[1]{Th\'eophile Dolmaire\,\footnote{E-mail addresses: \texttt{theophile.dolmaire@univaq.it} (T.~Dolmaire), \texttt{alessia.nota@gssi.it} (A.~Nota).}}
\author[2]{Alessia Nota}
\affil[1]{Dipartimento di Ingegneria e Scienze dell'Informazione e Matematica (DISIM), Università degli Studi dell'Aquila, 67100  L'Aquila, Italy} 
\affil[2]{Gran Sasso Science Institute, Viale Francesco Crispi 7, 67100 L’Aquila, Italy}

\begin{document}

\maketitle

\begin{abstract}
\noindent
In this paper we provide a rigorous derivation of the inelastic linear Boltzmann equation, in the Boltzmann-Grad limit, from a dissipative, random, Lorentz gas in arbitrary dimensions $d \geq 2$. Specifically, we  consider a microscopic particle system where scatterers are randomly distributed according to a Poisson process, and a tagged light particle undergoes  inelastic collisions with the scatterers following a reflection law characterized by a fixed restitution coefficient. 
We establish the existence and uniqueness of weak solutions to the inelastic linear Boltzmann equation within the class of non-negative Radon measures, assuming that the initial data has a finite exponential moment. We first show that the forward dynamics of the dissipative particle system is globally defined almost surely and then prove the weak-$\ast$ convergence of the microscopic solution towards the weak solutions of the inelastic linear Boltzmann equation, providing an explicit rate of convergence. Furthermore, under the same initial data assumptions, we prove the existence of strong solutions to the inelastic linear Boltzmann equation, constructed via a series representation of the solutions.
\end{abstract}

\noindent
\textbf{Keywords:}  Boltzmann–Grad Limit; Kinetic Equations;
Rigorous Derivation;
Linear Boltzmann Equation;
Lorentz Gas; 
Dissipative Particle Systems;
Inelastic Collisions;
Granular Media

\tableofcontents

\numberwithin{equation}{section}

\section{Introduction}

Kinetic equations describe in terms of distribution functions the evolution of large particle systems such as gas, plasmas, or even animal or human populations in which spread of diseases or evolution of opinions can be studied. The paradigmatic example is the Boltzmann equation, where the quantity $f(t,x,v)$ represents the probability to find particles of a dilute gas lying at time $t$ at position $x$ and moving with velocity $v$. The evolution of the distribution function is obtained from the microscopic dynamics of such a particle system, and the rigorous justification of the validity of the equation, deduced directly from the microscopic dynamics, is a task of central importance to establish solid foundations of statistical mechanics. Such a task is often referred to as the \emph{derivation} of the kinetic equations. The rigorous derivation of kinetic equations ensures that the equations can be used in a reliable manner to study the behaviour of large particle systems, and in some cases the procedure can be refined to provide a quantitative estimation.\\
\noindent
In the present article, we will establish the rigorous derivation of the linear inelastic Boltzmann equation in dimension $d$, with $d \geq 2$ arbitrary. The equation describes the inelastic Lorentz gas, that is, the evolution of non-interacting, light particles that evolve among a background of infinitely heavy, randomly distributed 
scatterers, and that collide inelastically with the scatterers. Such a system describes for instance the transport of electrons in a metal. We will start from the deterministic evolution of the particles, assumed to be point-particles, among the scatterers of radius $\varepsilon > 0$. The scatterers will be distributed randomly in $\mathbb{R}^d$ 
according to a Poisson process of intensity $\mu_\varepsilon > 0$. We will show that in the Boltzmann-Grad limit (that is, when $\varepsilon \rightarrow 0$ in such a way that $\mu_\varepsilon \varepsilon^{d-1} = 1$, so that the mean free path remains constant, while the volume fraction, proportional to $\mu_\varepsilon \varepsilon^{d}$, vanishes), the distribution of a tagged particle of the microscopic system can be approximated by the associated solution to the linear inelastic Boltzmann equation, that writes:
\begin{align}
\label{EQUATInelaLineaBoltzFinal}
\partial_t f(t,x,v) + v\cdot\nabla_x f(t,x,v) = \int_{\mathbb{S}^{d-1}} \frac{\vert v\cdot\omega \vert}{r^2} f(\hspace{0.25mm}'\hspace{-0.5mm}v) \dd \omega - C_d \vert v \vert f(v)
\end{align}
where $'\hspace{-0.25mm}v = v - \big(1 + 1/r\big) \big(v\cdot\omega\big)\omega$ is the pre-collisional velocity of a particle that collides with a scatterer at a point such that the normal of the scatterer is $\omega \in \mathbb{S}^{d-1}$ and such that the velocity after such a collision is $v$, $r \in [0,1]$ is the \emph{restitution coefficient} measuring the inelasticity of the collision and $C_d = \int_{\mathbb{S}^{d-1}} \big\vert \frac{v}{\vert v \vert} \cdot \omega \big\vert \dd \omega$ 
is a positive constant that depends only on the dimension. In the case $r = 1$, we recover the elastic case.\\
\newline
To the best of our knowledge, even if in a linear framework, our result constitutes the first rigorous derivation of a dissipative collisional equation where energy is lost over time, and, specifically, of an inelastic version of the Boltzmann equation from a microscopic particle system that evolves according to deterministic dynamics. More precisely, we prove that the solution of the microscopic inelastic Lorentz model converges in a suitable weak sense towards the solution  of \eqref{EQUATInelaLineaBoltzFinal}. We also emphasize that the derivation we provide is quantitative, in the sense that the error made by approximating the solution of the linear inelastic Boltzmann equation by the evolution of the light particle is explicitly given, as a power of $\varepsilon$, where $\varepsilon$ is the size of the scatterers of the Lorentz gas.

\paragraph{State of the art on the Lorentz Gas Model.} The Lorentz gas, introduced by H. A. Lorentz in 1905 to model the motion of electrons in metals, stands as a simple but highly non-trivial model in this context. It is a rare source of exact results in kinetic theory, providing a concrete example where microscopic reversibility can be reconciled with macroscopic irreversibility. Indeed, for this system, one can prove, under suitable scaling limits, a rigorous validation of linear kinetic equations and, from this, of diffusion equations.\\
\noindent
The Lorentz gas consists of a particle moving through infinitely heavy, randomly distributed
scatterers. The interaction between the Lorentz particle and the scatterers is specified by a central potential of finite range. Hence, the motion of the Lorentz particle is defined through the solution of Newton’s equations of motion.   
In the elastic case, the original system is Hamiltonian, the only stochasticity being that of the positions of the
scatterers. This randomness is essential to obtain the correct kinetic description.\\
\noindent
The first scaling one could consider is the Boltzmann-Grad limit (or low-density limit), 
namely, when the number of collisions is small, and thus
the mean free path of the particle is macroscopic. 
The initial breakthrough in this direction was achieved by Gallavotti, who derived the linear Boltzmann equation for a particle moving through a random distribution of fixed, hard scatterers \cite{Galla1969}, \cite{Gallavotti}. This work was subsequently generalized in terms of convergence of path
measures and extended to more general scatterer distributions by Spohn \cite{Spohn1}, see also \cite{Spohn2}. Boldrighini, Bunimovich and Sinai 
proved instead that the limiting Boltzmann equation holds for almost every scatterer configuration
drawn from a Poisson distribution \cite{BBS}.  We remark that these results have
been provided in cases of compactly supported potentials and provide a qualitative validation of the
linear Boltzmann equation, with no explicit control of the error in the kinetic limit.  For a quantitative estimate of the error for a Lorentz Gas of hard-spheres, we
refer to \cite{BNPP}. For first contributions towards the open problem of the validation in the case of long-range potentials we refer to \cite{DP, Ayi_017} and \cite{NSV}.\\
\newline
In the nonlinear framework, the derivation of the nonlinear Boltzmann equation from a Newtonian system of hard spheres has been rigorously proven first by Lanford in \cite{Lanf975}, in the low-density regime and for short times. Lanford's result was later extended, among many others, in \cite{GSRT013} (quantitative derivation), \cite{PuSS014} (short range potential interactions), \cite{Dolm023}, \cite{LeBi022} (domains with boundary), \cite{BGSS023} (study of the the fluctuations), and finally \cite{DeHM024} (extension of Lanford's theorem to arbitrary time intervals). For more details on the derivation of the nonlinear Boltzmann equation, the reader may refer to \cite{BGSS023_3}, and the references therein.\\
\newline
\noindent
As previously emphasized, the randomness in the distribution of scatterers plays a crucial role in deriving the linear Boltzmann equation. Indeed, in the case of a periodic setting, where heavy particles are positioned at the vertices of a lattice in Euclidean space, we face the maximum amount of correlation between the heavy particles. This significantly alters the structure of the resulting kinetic equation. The linear Boltzmann equation fails to be the correct mesoscopic description of this model (see \cite{CG1}).  The first complete proof of the Boltzmann-Grad limit for the periodic Lorentz gas, valid for all lattices and space dimensions, has been obtained by Marklof and Strömbergsson \cite{MS} (see also \cite{MS_2}). The properties of the resulting generalized linear Boltzmann equation are discussed by Caglioti and Golse in  \cite{CG3}.\\
\noindent
In the weak-coupling regime, when there are very many but weak collisions, interpreted as a central limit effect, a  
linear Landau equation appears. The first result in this direction was obtained by Kesten and Papanicolaou for a particle in $\mathbb{R}^3$ subject to a weak mean zero random force field \cite{KP}. Later, D\"urr, Goldstein, and Lebowitz showed that in $\mathbb{R}^2$ the velocity process converges in distribution to the Brownian motion on a surface of constant speed for sufficiently smooth interaction potentials \cite{DGL}. The linear Landau equation appears also in an intermediate scale between the low density and
the weak-coupling regime, see \cite{DR}.\\ 
\noindent
The rigorous derivation of hydrodynamic equations, specifically the heat
equation, from the mechanical system given by the random Lorentz gas relies on the kinetic approximation of the microscopic dynamics, i.e. uses the kinetic equation a bridge. This approach has been used to obtain the heat equation in different contexts, see \cite{BNP, BNPP, BGSR016}, and also \cite{ErSY008}. We additionally refer to \cite{LuTo} for a different approach in this direction.\\   
\noindent
Most of the mathematical results on the linear Boltzmann equation assume
that there are no additional transport terms arising from external fields. However, the presence of such fields significantly impact both the derivation of the equation in the low-density regime and the properties of its solutions. Specifically, the motion of a Lorentz particle in $\mathbb{R}^2$ under a uniform, constant magnetic field formally leads to a generalized Boltzmann equation with memory effects (see \cite{BMHH,   BHPH} and  also \cite{KuSp} where the model has been studied numerically). A rigorous derivation of this equation has recently been achieved in \cite{NSS}. We refer also to \cite{MN} for the rigorous derivation of linear kinetic equations  that include magnetic transport effects.

\paragraph{The linear inelastic Lorentz model and the inelastic linear Boltzmann equation.} In the case of light particles interacting inelastically with the scatterers, the question of the rigorous derivation remained open. The derivation is also lacking for the nonlinear inelastic Boltzmann equation. This is because of the singularities developed by systems of inelastic particles, the so-called \emph{granular materials}, and by the kinetic equations describing these objects (see for instance \cite{Vill006}). The difficulties come on the one hand from the decay of the temperature, preventing the existence of steady states and implying the convergence towards self-similar solutions \cite{MiMo009}. On the other hand, the particles tend to create stable clusters (hence the name of \emph{granular} media), inducing an explosion of the gradient in the position variable or convergence towards Dirac masses at the level of the kinetic equation (\cite{BeCP997}, \cite{Tosc000}). In particular, the creation of clusters violates the separation of the micro- and mesoscopic scales, such a separation being a crucial ingredient to obtain the validity of kinetic equations from particle systems. At the level of the particle system, the dynamics among the clusters might degenerate so that infinitely many collisions can take place in finite time. This is the phenomenon of \emph{inelastic collapse} (see \cite{MNYo994}, \cite{CoGM995}, \cite{ZhKa996}, \cite{CDKK999}, \cite{BeCa999}, \cite{McNa012}, \cite{DoVe024}, \cite{DoVe025}, \cite{DoHu024}, \cite{Dolm025}), which represents in particular a serious obstruction to perform a rigorous derivation of the nonlinear inelastic Boltzmann equation. We will address carefully this issue in the linear setting.\\
In the case of the linear model, studied numerically in \cite{HDHa001} (see also \cite{MaPi999}, in the case when an external force re-injects energy in the system, and \cite{MaPi007}, \cite{DoMN025} in the specific case of the Maxwell collision kernel), the decay of the temperature can be estimated from below, so that it is possible to prove that the gas does not cool down in finite time. Besides, the question of the regularity of the solutions of the linear equation remained open to the best of our knowledge. We will discuss this question, providing a natural condition to ensure the convergence of the series representation of the solutions to the linear inelastic Boltzmann equation.\\
For a general introduction to granular gases and the inelastic Boltzmann equation, we refer to the classical references \cite{BrPo004}, \cite{Vill006} and \cite{CHMR021}.

\paragraph{Formal derivation of the inelastic linear Boltzmann equation.} For the sake of completeness, in this section, we will present the usual formal arguments that enable to recover the inelastic linear Boltzmann equation, describing at the kinetic scale the evolution of a tagged particle in a background of inelastic scatterers. The Boltzmann equation is obtained formally by considering the number $f(t,x,v)$ of particles at time $t$, position $x$ and moving with velocity $v$, and by evaluating the change rate of $f$ due to the different microscopic mechanisms that affect the dynamics of the particles. In the present case, the particles described by $f$ do not interact one with another, and can collide only with the background of fixed scatterers. The variation per unit of time of the number of particles, at time $t$ and at the position $x$, with a velocity $v \in A$ for $A$ measurable, is given by:
\begin{align}
\label{EQUATFormaDeriv}
\frac{\dd}{\dd t} \int_{\mathbb{R}^d_v} f(t,x,v) \mathds{1}_A(v) \dd v &= - \int_{\mathbb{R}^d_v} v\cdot\nabla_x f(t,x,v) \mathds{1}_A(v) \dd v \nonumber\\
&\hspace{10mm}+ \int_{\mathbb{S}^{d-1}_\omega}\int_{\mathbb{R}^d_v} b(\vert v \vert,\left\vert \frac{v}{\vert v \vert}\cdot \omega \right\vert) f(v) \Big[ \mathds{1}_A(v') - \mathds{1}_A(v) \Big] \dd v \dd \omega
\end{align}
where $v'$ is the post-collisonal velocity of a particle colliding with a scatterer with impact parameter $\omega \in \mathbb{S}^{d-1}$ and pre-collisional velocity $v$. Here, $v'$ is given by the scattering relation $v'= \kappa_\omega(v)$, already defined in \eqref{EQUATDefinFlot_TempsCroissLoi_Crois}, as:
\begin{align}
\label{EQUATScattering}
v' = \kappa_\omega(v) = v - (1+r) \big( v\cdot\omega\big) \omega.
\end{align}
If $v' \in A$, a collision $v \rightarrow v'$ increases the number of particles with velocity in $A$, while if $v \in A$ a collision $v \rightarrow v'$ decreases the number of particles in $A$ (assuming that $v$ and $v'$ do not belong together to $A$, which is very likely is $A$ has a small measure). All the possible collisions are considered thanks to the integration over all the possible pre-collisional velocities $v \in \mathbb{R}^d$ and angular parameters $\omega \in \mathbb{S}^{d-1}$.\\
Finally, the collision kernel $b$ describes the rate at which the collisions with particles with pre-collisional velocity $v$ and angular parameter $\omega$ take place. By Galilean invariance, such a collision kernel depends only on $\vert v \vert$ and $\big\vert \big(v/\vert v \vert\big) \cdot \omega \big\vert$.\\
Considering an approximation of any measurable function $\varphi$ by scale functions, we obtain the following weak form for the linear Boltzmann equation, in the case of a general collision kernel $b$:
\begin{align}
\label{EQUATLineaBoltzGenerFormeFaibl}
\frac{\dd}{\dd t} \int_{\mathbb{R}^d_v} f(t,x,v) \varphi(v) \dd v + \int_{\mathbb{R}^d_v} v\cdot\nabla_x f(t,x,v) \varphi(v) \dd v = \int_{\mathbb{S}^{d-1}_\omega}\int_{\mathbb{R}^d_v} b(\vert v \vert,\left\vert \frac{v}{\vert v \vert}\cdot \omega \right\vert) f(v) \Big[ \varphi(v') - \varphi(v) \Big] \dd v \dd \omega.
\end{align}
Denoting by $J\kappa_\omega$ the Jacobian determinant of the scattering $\kappa_\omega:v\mapsto v'$ (for $\omega$ fixed), we have $\big\vert J\kappa_\omega \big\vert = r$. Denoting by $\hspace{0.25mm}'\hspace{-0.25mm}v$ the velocity such that $\kappa_\omega(\hspace{0.25mm}'\hspace{-0.25mm}v) = v$ or equivalently $\hspace{0.25mm}'\hspace{-0.25mm}v = \kappa^{-1}_\omega(v) = v - (1+1/r)(v\cdot\omega)\omega$,  we deduce the following strong form for the linear Boltzmann equation, for a general collision kernel $b$:
\begin{align}
\label{EQUATLineaBoltzGenerFormeForte}
\partial_t f(t,x,v) + v\cdot\nabla_x f(t,x,v) &= \int_{\mathbb{S}^{d-1}_\omega} \frac{1}{\vert J\kappa_\omega ('v) \vert} b\left( \vert \hspace{0.25mm}'\hspace{-0.5mm}v \vert, \left\vert \frac{\hspace{0.25mm}'\hspace{-0.25mm}v}{\vert \hspace{0.25mm}'\hspace{-0.25mm}v \vert}\cdot\omega \right\vert \right) f(\hspace{0.25mm}'\hspace{-0.25mm}v) \dd \omega - f(v) \int_{\mathbb{S}^{d-1}_\omega} b\left( \vert v \vert, \left\vert \frac{v}{\vert v \vert}\cdot\omega \right\vert \right) \dd \omega \nonumber\\
&= \int_{\mathbb{S}^{d-1}_\omega}\frac{1}{r} b\left( \vert \hspace{0.25mm}'\hspace{-0.25mm}v \vert, \left\vert \frac{\hspace{0.25mm}'\hspace{-0.25mm}v}{\vert \hspace{0.25mm}'\hspace{-0.25mm}v \vert}\cdot\omega \right\vert \right)  f('\hspace{-0.25mm}v) \dd \omega - f(v) \int_{\mathbb{S}^{d-1}_\omega}b\left( \vert v \vert, \left\vert \frac{v}{\vert v \vert}\cdot\omega \right\vert \right) \dd \omega.
\end{align}

\noindent
In the specific case of \emph{hard sphere} collisions, that we will consider in this paper, the explicit expression of the collision kernel is $b(\vert V \vert, \left\vert \frac{V}{\vert V \vert}\cdot\omega \right\vert) = \vert V\cdot\omega \vert$. Therefore, using that $\vert \hspace{0.25mm}'\hspace{-0.5mm}v\cdot \omega \vert = \frac{1}{r} \vert v \cdot \omega \vert$, the strong form of the linear Boltzmann equation for a Lorentz gas, in the case of inelastic collisions with a fixed restitution coefficient, and for the hard sphere collision kernel, is:
\begin{align}
\label{EQUATLineaBoltzSphDuFormeForte}
\partial_t f(t,x,v) + v\cdot\nabla_x f(t,x,v) = \int_{\mathbb{S}^{d-1}_\omega} \frac{\vert v\cdot\omega \vert}{r^2} f('v) \dd \omega - f(v) \int_{\mathbb{S}^{d-1}_\omega} \vert v\cdot \omega \vert \dd\omega.
\end{align}
It is possible to simplify the expression of the loss term of the collision operator, since we have in general:
\begin{align}
\label{EQUATDefin_C_d__bis_}
\int_{\mathbb{S}^{d-1}_\omega} \vert \frac{v}{\vert v \vert} \cdot\omega \vert \dd \omega = \int_{\mathbb{S}^{d-1}_\omega} \vert R(e_1) \cdot R(\sigma) \vert \cdot \vert R'(\sigma) \vert \dd \sigma = \int_{\mathbb{S}^{d-1}_\omega} \vert e_1\cdot\sigma \vert \dd \sigma,
\end{align}
where $R$ is any vectorial rotation that sends the first vector of the canonical basis  $e_1$ on $v/\vert v \vert$ (which is a fixed vector, since the integration variable is $\omega$). We will denote by $C_d$ the previous integral, which depends only on the dimension $d$:
\begin{align}
\label{EQUATDefin_C_d_}
C_d =  \int_{\mathbb{S}^{d-1}_\omega} \vert e_1\cdot\sigma \vert \dd \sigma.
\end{align}
This allows to obtain, for a general dimension $d \geq 2$, the final expression \eqref{EQUATInelaLineaBoltzFinal} for the strong form of the inelastic linear Boltzmann equation. Notice that in dimension $2$, we have $
C_2 = 
\int_0^{2\pi} \vert \cos\theta \vert \dd\theta = 4,
$
whereas in dimension $3$, choosing the parametrization $\sigma = (\cos\varphi\cos\theta,\sin\varphi\cos\theta,\sin\theta)$ of the sphere, with $\theta \in [-\pi/2,\pi/2]$ and $\varphi \in [0,2\pi]$, we find
$ C_3 =
\int_0^{2\pi}\int_{-\pi/2}^{\pi/2} \vert \cos\varphi \cos\theta \vert \cos\theta \dd \theta \dd \varphi = 2\pi.$

\paragraph{The difficulties concerning the derivation of the inelastic linear Boltzmann equation.}
The rigorous derivation of the linear Boltzmann equation, in the elastic case, as performed originally by Gallavotti in the seminal article \cite{Gallavotti}, consists in considering the quantity:
\begin{align}
\label{EQUATQuantGallavotti}
f_\varepsilon(t,x,v) = \mathbb{E}_{\mu_\varepsilon} \big[ f_0\left(T^{-t}_{c,\varepsilon}(x,v)\right)\big],
\end{align}
and showing that $f_\varepsilon$ converges, in some strong sense (pointwise, in the $L^1$ norm, or more generally in any $L^p$ norm), in the Boltzmann-Grad limit $\varepsilon \rightarrow 0$, $\mu_\varepsilon \varepsilon^{d-1} = 1$ towards the solution of the linear Boltzmann equation with initial datum $f_0$. 
 Observe that $f_\varepsilon(t,x,v)$ is defined by \eqref{EQUATQuantGallavotti} as the mean value of $f_0\left(T^{-t}_{c,\varepsilon}(x,v)\right)$, averaged over all the distributions $c$ of the scatterers, distributed according to a Poisson process in $\mathbb{R}^d$ with intensity $\mu_\varepsilon$. $T^{-t}_{c,\varepsilon}(x,v)$ is the preimage of the point $(x,v)$ of the phase space, by the dynamical flow of the tagged particle. In other words, \eqref{EQUATQuantGallavotti} describes the probability to find the tagged particle in the configuration $(x,v)$ of the phase space at time $t$, which is equivalent to consider the probability to find the tagged particle at the initial time $t=0$ in the configuration $T^{-t}_{c,\varepsilon}(x,v)$.\\
In the case of a tagged particle colliding inelastically with the scatterers, one has to take into account the fact that the dynamical flow of the tagged particle, at $c$ fixed, is contracting the measure in the phase space. Therefore, in the present case, we have to consider:
\begin{align}
\label{EQUATQuantNotre_Cas_}
f_\varepsilon(t,x,v) = \mathbb{E}_{\mu_\varepsilon} \big[ \vert J\big(T^{-t}_{c,\varepsilon}(x,v)\big) \vert \cdot f_0 \big( T^{-t}_{c,\varepsilon}(x,v) \big) \big],
\end{align}
where $J\big(T^{-t}_{c,\varepsilon}(x,v)\big)$ denotes the Jacobian determinant of $T^{-t}_{c,\varepsilon}$, with respect to the $(x,v)$  variables. This Jacobian determinant can be computed explicitly, for example relying on the Transport-Collision-Transport formula introduced in \cite{DoVeNot}.\\
We further observe that, in the inelastic case, we are facing an additional major difficulty. More precisely, if the intensity $\mu_\varepsilon$ of the Poisson process satisfies the Boltzmann-Grad limit, that is, is scaled as $\mu_\varepsilon \varepsilon^{d-1} = 1$, the mean free path of the tagged particle is a constant independent of $\varepsilon$. But in the case of inelastic collisions, the velocity of the tagged particle is increasing at each collision when the dynamics is considered backward in time. More precisely, the norm of the velocity will grow geometrically with a positive probability, leading to a sequence of infinite collisions in finite time, which prevents us from considering directly the inverse $T_{c,\varepsilon}^{-t}$ of the dynamical flow.\\
To overcome this difficulty, we will rely on an approach based on the weak formulation of the linear inelastic Boltzmann equation. This approach is the analog of the one used in \cite{NoVe017}. We aim to discuss the direct derivation from the expression \eqref{EQUATQuantNotre_Cas_}, that is, the derivation of \eqref{EQUATInelaLineaBoltzFinal} in the strong sense, in a future work.\\
\newline
The proof of the derivation relies on the series representation of the solutions of the linear Boltzmann equation, following the pioneering work of Gallavotti \cite{Gallavotti}. It is worth to remark that in the case considered in this paper, due to the inelastic nature of the  collisions,  we face the following serious difficulty. Considering the series expansion of the solution to the inelastic linear Boltzmann equation, in the integrand of the $k$-th term appears the product $\prod_{l=1}^k \vert v^{-(l-1)}\cdot \omega_l \vert/r^2$, where $\omega_l \in \mathbb{S}^{d-1}$ is the angular parameter of the $l$-th collision, and $v^{-(l)}$ is the $l$-th pre-collisonal velocity (defined as the $l$-th iteration of the inverse of the scattering $\kappa_{\omega_l}$ \eqref{EQUATScattering}). Due to the dissipative nature of the collisions, $\vert v^{-(l)} \vert$ grows exponentially fast in $l$, so that the $k$-th term of the series representation grows as $(1/r)^{2k+k^2}$, yielding severe issues to prove the convergence of the series.\\
We address this issue by considering solutions with a bounded exponential moment. To the best of our knowledge, the series representation of the solutions of the inelastic linear Boltzmann equation, and the proof of its convergence, is a novelty.

\paragraph{Plan of the paper.} The content is organized as follows. In Section 2, we introduce the particle system we are considering at the microscopic level, as well as the Boltzmann-Grad scaling that allows to reach the mesoscopic description. We then present the main results, and the strategy we adopt for the derivation. Section 3 is devoted to the rigorous results we will need concerning the inelastic linear Boltzmann equation \eqref{EQUATInelaLineaBoltzFinal}. We discuss in particular the series representation of the solutions of \eqref{EQUATInelaLineaBoltzFinal}. In Section 4 we perform the rigorous derivation of \eqref{EQUATInelaLineaBoltzFinal} from the inelastic Lorentz gas. Finally, in Section 5 we prove that the dynamics of the light particle inelastically colliding with the scatterers, is well-posed, globally in time, for a set of distributions of scatterers of probability $1$.

\paragraph{Notations}
We will denote by $\mathcal{P}\big(\mathbb{R}^d\times \mathbb{R}^d\big)$ the set of the probability measures in $\mathbb{R}^d\times \mathbb{R}^d$, $\mathcal{M}\big(\mathbb{R}^d\times \mathbb{R}^d\big)$ will denote the set of finite signed Radon measures on $\mathbb{R}^d\times \mathbb{R}^d$, and $\mathcal{M}_{+}\big(\mathbb{R}^d\times \mathbb{R}^d\big)$ its non-negative cone, that is, the subset of non-negative measures in $\mathcal{M}\big(\mathbb{R}^d\times \mathbb{R}^d\big)$ (see \cite{Foll999}, Section 7.1). In addition, we will denote by $\mathcal{M}_{+,1}\big(\mathbb{R}^d\times \mathbb{R}^d\big)$ the subset of $\mathcal{M}\big(\mathbb{R}^d\times \mathbb{R}^d\big)$ of non-negative Radon measures with a finite order moment, that is, such that:
\begin{align}
\label{EQUATDefin_M_1+}
\mathcal{M}_{+,1}\big(\mathbb{R}^d\times \mathbb{R}^d\big) = \big\{ f \in \mathcal{M}_{+}\big(\mathbb{R}^d\times \mathbb{R}^d\big)\ /\ \int_{\mathbb{R}^d\times\mathbb{R}^d} \hspace{-5mm} \vert v \vert \, f(\dd x,\dd v) < +\infty \big\}.
\end{align}
$\mathcal{M}\big(\mathbb{R}^d\times \mathbb{R}^d\big)$, $\mathcal{M}_+\big(\mathbb{R}^d\times \mathbb{R}^d\big)$ and $\mathcal{M}_{+,1}\big(\mathbb{R}^d\times \mathbb{R}^d\big)$ will be endowed with the total variation norm.\\
We will denote by $\mathcal{C}_0\big(\mathbb{R}^d\times\mathbb{R}^d\big)$ the set of continuous functions on $\mathbb{R}^d \times \mathbb{R}^d$ that are vanishing at infinity, and by $\mathcal{C}^\infty_c([0,+\infty[\times\mathbb{R}^d\times\mathbb{R}^d)$ the set of infinitely differentiable functions, compactly supported in $[0,+\infty[\times\mathbb{R}^d\times\mathbb{R}^d$.\\
We recall that since $\mathbb{R}^d\times\mathbb{R}^d$ is locally compact, the dual of $\mathcal{C}_0\big(\mathbb{R}^d\times\mathbb{R}^d\big)$ is $\mathcal{M}\big(\mathbb{R}^d\times \mathbb{R}^d\big)$ (see for instance \cite{Rudi987}).
\newline
When $I$ is a finite set, we will denote by $\# I$ its cardinal. For a subset $B$ of $\mathbb{R}^d$ which is Lebesgue-measurable, we will denote by $\vert B \vert$ its Lebesgue measure.\\
For two subsets $A$, $B$ of $\mathbb{R}^d$, we also introduce their \emph{sum}, denoted by $A+B$, and defined as:
\begin{align}
A+B = \{x+y\ /\ x\in A,\, y\in B\}.
\end{align}
We will often consider the sum of segments with balls.\\
Finally, for any subset $A$ of $\mathbb{R}^d$, we will denote by $\overline{A}$ its closure, that is, $\overline{A}$ is the smallest closed subset of $\mathbb{R}^d$ that contains $A$.

\section{Model, main results and strategy}

\subsection{The model}

\subsubsection{The microscopic dynamics}

\paragraph{Distributions of scatterers.} We call \emph{distribution of scatterers} (or \emph{distribution} in short), denoted by $c$, a finite or countable set of points $c_i \in \mathbb{R}^d$, which is locally finite. In other words, $c = \big(c_i\big)_{i \in I} \in \big(\mathbb{R}^d\big)^I$, with $\#I \in \mathbb{N}$ or $I = \mathbb{N}$, and $\# ( c \cap K) \in \mathbb{N}$ for any compact set $K \subset \mathbb{R}^d$. The whole collection of distributions of scatterers will be denoted by $C$. In order to define the Poisson process of intensity $\mu > 0$, we equip $C$ with the measure $\mathbb{P}_\mu$ defined as follows. The probability of finding exactly $N$ points $c_{i_1},\dots, c_{i_N}$ of $c$ in a given Lebesgue-measurable subset $B\subseteq \mathbb{R}^d$, of finite measure, is equal to:
\begin{align}
\label{EQUATPoissonPro}
\mathbb{P}_{\mu,N}(B) = e^{-\mu \vert B \vert}\frac{\mu^N}{N!} \vert B \vert^N,
\end{align}
where $\vert B \vert$ is the Lebesgue measure of $B$. In other words, the probability of finding $N$ points $c_{i_1},\dots,c_{i_N}$ of $c$ lying respectively in the infinitesimal volumes $dc_1,\dots, dc_N$ centered at $x_1,\dots, x_N$ is 
\begin{align}
\mathbb{P}_{\mu,N} (x_1, \dots, x_N ) \dd x_1 \dots \dd x_N = e^{-\mu \vert B\vert} \frac{ \mu^N }{N!} \dd x_1 \dots \dd x_N,
\end{align}
which is the Janossy measure of order N of the Poisson point process restricted to the Borel subset $B\subseteq \mathbb{R}^d$. The reader may refer to Definition 4.6, Example 4.8 and formula (4.21) in \cite{LaPe}, and also to the reference \cite{Gols022}).\\
For $\varepsilon > 0$ and given a distribution $c$, we will also consider the family of balls $\big(B(c_i,\varepsilon)\big)_{i\in I}$ centered on the points $c_i$ of $c$, and of radius $\varepsilon > 0$. We will call this family of balls the \emph{distribution of scatterers $c$ of radius $\varepsilon$}.

\paragraph{The dynamics of the tagged particle.} We define now the dynamics of the tagged particle among the scatterers, in terms of a singular differential equation.

\begin{defin}[Forward flow of the tagged particle]
\label{DEFINFlot_TempsCrois}
Let $r \in\ ]0,1]$, $\varepsilon > 0$, $t \in \mathbb{R}_+^*$, and let $c \in C$ be a distribution of scatterers in $\mathbb{R}^d$. Let $(x,v) \in \mathbb{R}^d \times \mathbb{R}^d$. We define the \emph{forward flow of the tagged particle}, as the piecewise affine mapping:
\begin{align}
T_{c,\varepsilon} : \left\{
\begin{array}{ccc}
[0,t] \times \mathbb{R}^d \times \mathbb{R}^d & \rightarrow & \mathbb{R}^d \times \mathbb{R}^d,\\
(s,x,v) & \mapsto & T_{c,\varepsilon}^s(x,v) = \big( x_{c,\varepsilon}(s),v_{c,\varepsilon}(s) \big)
\end{array}
\right.
\end{align}
that is right continuous in $t$ at $(x,v) \in \mathbb{R}^d \times \mathbb{R}^d$ fixed, with a limit from the left at all points, and that satisfies, using the notation $\big( x_{c,\varepsilon}(s),v_{c,\varepsilon}(s) \big) = \big(x(s),v(s)\big)$:
\begin{align}
\label{EQUATDefinFlot_Cond1}
T_{c,\varepsilon}^0(x,v) = (x,v),
\end{align}
\begin{align}
\label{EQUATDefinFlot_Cond2}
\frac{\dd}{\dd s} T_{c,\varepsilon}^s(x,v) = \big(v(s),0\big) \hspace{5mm} \text{if} \hspace{5mm} d \big( x(s), c \big) > \varepsilon,
\end{align}
and
\begin{align}
\label{EQUATDefinFlot_TempsCroissLoi_Crois}
v(s) = v(s^-) - (1+r) \big( v(s^-) \cdot \omega(s) \big) \omega(s) \hspace{5mm} \text{if} \hspace{5mm} d \big( x(s), c \big) = \varepsilon \hspace{2mm} \text{and} \hspace{3mm} \# \big( c \cap \big( x(s) + \overline{B(0,\varepsilon)} \big) \big) = 1,
\end{align}
where:
\begin{align}
v(s^-) = \lim_{\substack{\tau \rightarrow s\\ \tau < s}} v(\tau),
\end{align}
and
\begin{align}
\label{EQUATDefinOmega}
\omega(s) = \lim_{\substack{\tau \rightarrow s\\ \tau < s}} \frac{x(\tau) - c(s)}{\big\vert x(\tau) - c(s) \big\vert} \hspace{5mm} \text{with} \hspace{5mm} c(s) \in \mathbb{R}^d \hspace{3mm} \text{being the center of the collided scatterer at time }s,\\
\text{ that is, such that} \hspace{1mm} \{ c(s) \} =  c \cap \big( x(s) + \overline{B(0,\varepsilon)} \big).\nonumber
\end{align}
\end{defin}

\begin{figure}[h!]
\centering
    \includegraphics[trim = 0cm 0cm 0cm 1.5cm, width=0.4\linewidth]{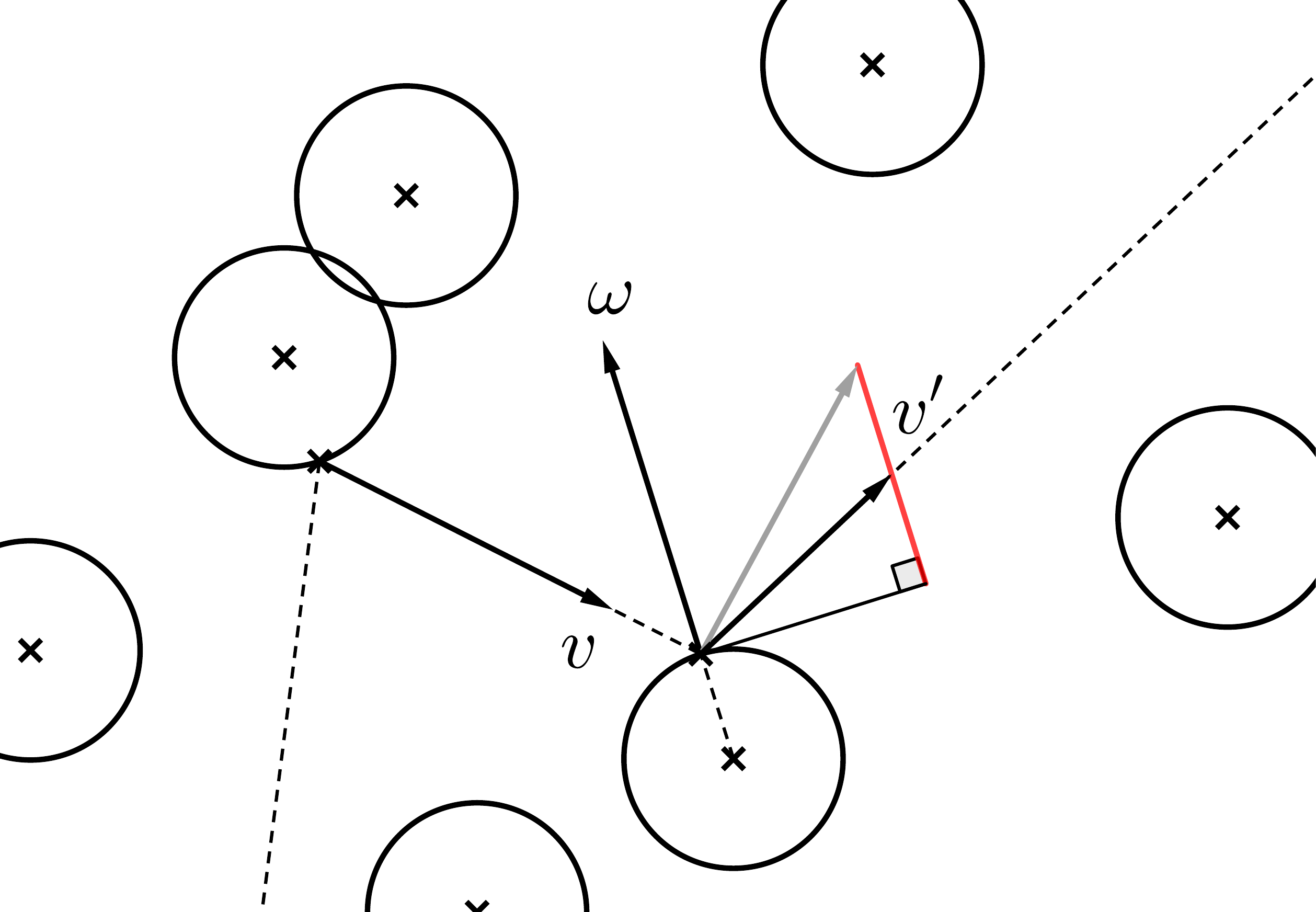}
\caption{Representation of a collision obtained by the scattering mapping: $v' = \kappa_\omega(v)$. In comparison, an elastic collision ($r = 1$) is represented in light grey.}
\label{FIGURScattering}
\end{figure}

\noindent
Depending on the configuration of scatterers $c$, the forward flow of the tagged particle introduced in Definition \ref{DEFINFlot_TempsCrois} might not be globally defined. We introduce therefore a notation that will be useful to state the main result of the present article.

\begin{defin}[Generalized forward flow] \label{def:genFlowF}
\label{DEFINGener_Flow}
For any positive time $t > 0$ and any configuration $c$ of scatterers, we define the \emph{generalized forward hard sphere flow} as the mapping:
\begin{align}
\widetilde{T}_{c,\varepsilon}:\left\{
\begin{array}{cccc}
[0,t] \times \mathbb{R}^d \times \mathbb{R}^d & \hspace{-2mm} \rightarrow& \mathbb{R}^d \times \mathbb{R}^d,&\\
(s,x,v) & \hspace{-2mm} \mapsto& \widetilde{T}_{c,\varepsilon}^s(x,v) &
\hspace{-2mm} = \left\{
\begin{array}{lc}
\, T_{c,\varepsilon}^s(x,v) &\hspace{-1mm}\text{if the flow } s \mapsto T_{c,\varepsilon}^s(x,v) \text{ is globally defined}\\
&\hspace{-5mm}\text{on }[0,t]\text{ for the distribution of scatterers } c\\
& \text{and the initial configuration } (x,v),\\\\
\, (x,v) &\text{ otherwise.}
\end{array}
\right.
\end{array}
\right.
\end{align}
\end{defin}

\begin{remar}
\label{REMARBackw_Flow}
In Definitions \ref{DEFINFlot_TempsCrois} and \ref{DEFINGener_Flow}, we introduced only the flow \emph{forward in time}. We will see that it will be possible to establish a  well-posedness result for such a flow (see Proposition \ref{PROPODefinGlobaDynam}). Nevertheless, it will be conceptually important to consider a \emph{backward in time} version of such a flow, that is, to consider the past of a trajectory which lies at $(x,v)$ in the phase space at time $0$. Such a backward flow will be denoted by $T_{c,\varepsilon}^{t}(x,v)$, with $t < 0$. 
However, and contrary to the elastic case, the backward flow turns out to be more singular than the forward flow. Indeed, it is easy to consider scatterer configurations that lead to an infinite number of collisions in finite time when considering the backward in time dynamics of a tagged particle colliding with inelastic scatterers, as the sequence of consecutive time intervals of mean free flight is a geometric series of ratio $> 1$. This phenomenon is an instance of the \emph{inelastic collapse}, well documented in the case of interacting inelastic hard spheres (see for instance \cite{BrPo004}, \cite{MNYo994}, \cite{McNa012}, \cite{DoVe024}, \cite{Dolm025}, and the references therein).
\end{remar}

\subsubsection{The Boltzmann-Grad limit}

In order to perform the derivation of the linear inelastic Boltzmann equation from the dynamics of the tagged particle, we will consider the so-called Boltzmann-Grad scaling, also known as the low density limit.

\begin{defin}[Boltzmann-Grad limit]
Let $\varepsilon > 0$ and $\mu_\varepsilon > 0$. We consider a distribution of scatterers $\big(B(c_i,\varepsilon)\big)_{i\in I}$ of radius $\varepsilon > 0$, distributed according to the Poisson process of intensity $\mu_\varepsilon > 0$ in $\mathbb{R}^d$. We say that the distribution of scatterers satisfies the Boltzmann-Grad scaling condition if:
\begin{align}
\label{EQUATBoltzmann-Grad_Limit}
\mu_\varepsilon \varepsilon^{d-1} = 1.
\end{align}
By definition, the \emph{Boltzmann-Grad limit} is characterized by $\varepsilon \rightarrow 0$ while \eqref{EQUATBoltzmann-Grad_Limit} holds.
\end{defin}

\noindent
By assumption, in the Boltzmann-Grad limit the tagged particle suffers, in average, a fixed number of collisions per unit of time with the scatterers $B(c_i,\varepsilon)$. We remark that the Boltzmann-Grad limit is also referred to as the low density limit, because the volume fraction tends to zero. Indeed, in any fixed ball $B(0,R)$, the average volume occupied by the scatterers of volume $C(d)\varepsilon^d$ is given by:
\begin{align}
\sum_{k=0}^{+\infty} k\,  e^{-\mu_\varepsilon \vert B(0,R) \vert} \frac{\mu_\varepsilon^k}{k!} \big\vert B(0,R) \big\vert^k C(d) \varepsilon^d = \mu_\varepsilon \big\vert B(0,R) \big\vert C(d) \varepsilon^d = C(d) \big\vert B(0,R) \big\vert \varepsilon \rightarrow 0.
\end{align}
\paragraph{Discussion on the inelastic collapse in the Boltzmann-Grad limit.} Taking into account the dissipation of the kinetic energy due to the scattering \eqref{EQUATFormaDeriv}, we have that  the norm of the velocity of the tagged particle is, in average,  contracted by a constant at each collision. Since the mean free path is constant in the Boltzmann-Grad scaling, the series of the consecutive time intervals of mean free flight is diverging.\\
As a consequence, we can expect that the gas of tagged particles does not cool down in finite time when \eqref{EQUATBoltzmann-Grad_Limit} holds. In other words, it is expected that the inelastic collapse (see Remark \ref{REMARBackw_Flow}) does not take place within the inelastic Lorentz gas, at least for almost every distribution of scatterers. This comment applies only to the case of the forward in time dynamics. Concerning the backwards dynamics, the effect is the opposite: the norm of the velocities grows exponentially fast with the number of collisions, and therefore in the low density limit the series of the consecutive time intervals of mean free flight times is converging. This means that we expect the inelastic collapse to take place for the backwards dynamics, and it is actually not hard to imagine configurations of scatterers leading to the collapse.

\subsection{Main result}

\noindent
We first provide the notion of solution of the microscopic inelastic Lorentz model as follows.
\begin{defin}[Microscopic solution of the inelastic Lorentz model]
\label{def:f_ep}
Let $\varepsilon > 0$ and $\mu_\varepsilon> 0$ be two real numbers satisfying the condition \eqref{EQUATBoltzmann-Grad_Limit}. Let $f_0\in \mathcal{P}\big(\mathbb{R}^d\times \mathbb{R}^d\big)$, for any Borel set $A$ of $\mathbb{R}^d\times \mathbb{R}^d$ we define the 
\emph{microscopic solution of the inelastic Lorentz model}, 
that we denote by $f_\varepsilon \in L^{\infty}\big( [0,T); \mathcal{M}_{+}(\mathbb{R}^d\times \mathbb{R}^d)\big)$, as the measure-valued function defined by the following duality relation:
\begin{equation}\label{defeq:f_ep}
\int_{\mathbb{R}^d\times\mathbb{R}^d} \mathds{1}_A(x,v) f_\varepsilon(t,\dd x,\dd v) =\int_{\mathbb{R}^d\times \mathbb{R}^d} \mathbb{P}_{\mu_\varepsilon}(\{c \in C \ /\ \widetilde{T}_{c,\varepsilon}^{t}(x_0,v_0) \in A\})f_0(\dd x_0,\dd v_0),
\end{equation}
where $C$ is the Poisson process on $\mathbb{R}^d$ of intensity $\mu_\varepsilon$, $\mathbb{P}_{\mu_\varepsilon}$ the associated measure, and $\widetilde{T}_{c,\varepsilon}^t(x,v)$ is the generalized forward flow of the tagged particle introduced in Definition \ref{DEFINGener_Flow}.
\end{defin}
\noindent
Our goal is to prove that in some weak sense $ f_{\varepsilon}(x,v,t)\to f(x,v,t)$ as $\varepsilon\to 0$ in the Boltzmann-Grad limit 
$\mu_\varepsilon \cdot \varepsilon^{d-1} = 1$, where $f$ is a solution of the kinetic equation \eqref{EQUATLineaBoltzSphDuFormeForte} 
with initial datum $f_0$. We now introduce the notion of solutions to  \eqref{EQUATLineaBoltzSphDuFormeForte} in the sense of measures. 

\begin{defin}[Weak solution of the linear inelastic Boltzmann equation \eqref{EQUATLineaBoltzSphDuFormeForte}]
\label{DEFINSolutFaibleEquatBoltzLineaInela}
Let $f_0 \in \mathcal{P}\big(\mathbb{R}^d\times \mathbb{R}^d\big)$ be a probability measure on $\mathbb{R}^d \times \mathbb{R}^d$ and $f \in \mathcal{C}\big([0,+\infty[,\mathcal{M}_{+,1}(\mathbb{R}^d\times\mathbb{R}^d)\big)$ be a continuous function taking values in the set of measure $\mathcal{M}_{+,1}(\mathbb{R}^d\times\mathbb{R}^d)$, defined in \eqref{EQUATDefin_M_1+}.\\
We say that $f$ is a \emph{weak solution of the linear inelastic Boltzmann equation \eqref{EQUATLineaBoltzSphDuFormeForte} with initial datum $f_0$} if $f(0,\cdot,\cdot) = f_0$ and if
\begin{align}
&\hspace{-2mm}- \int_{\mathbb{R}^d_x} \hspace{-1mm} \int_{\mathbb{R}^d_v} \hspace{-1.5mm} \widetilde{\varphi}(0,x,v) f_0(\dd x,\dd v) - \int_0^{+\infty} \hspace{-2.5mm} \int_{\mathbb{R}^d_x} \hspace{-1mm} \int_{\mathbb{R}^d_v} \hspace{-1.5mm} \partial_t \widetilde{\varphi}(t,x,v) f(t,\dd x,\dd v) \dd t - \int_0^{+\infty} \hspace{-2.5mm} \int_{\mathbb{R}^d_x} \hspace{-1mm} \int_{\mathbb{R}^d_v} \hspace{-1.5mm} v\cdot\nabla_x \widetilde{\varphi}(t,x,v) f(t,\dd x,\dd v) \dd t \nonumber\\ 
&\hspace{50mm} = \int_0^{+\infty} \hspace{-2.5mm} \int_{\mathbb{R}^d_x} \hspace{-1mm} \int_{\mathbb{R}^d_v} \int_{\mathbb{S}^{d-1}_\omega} \vert v \cdot \omega \vert \big[ \widetilde{\varphi}(t,x,v') - \widetilde{\varphi}(t,x,v) \big] \dd \omega f(t,\dd x,\dd v) \dd t
\end{align}
for any test function $\widetilde{\varphi} \in \mathcal{C}_c^\infty([0,+\infty[\times\mathbb{R}^d\times\mathbb{R}^d)$, with $v' = v - (1+r)(v\cdot\omega)\omega$ defined in \eqref{EQUATScattering}.
\end{defin}

\noindent
We are now in position to state the main result of the present article.

\begin{theor}[Derivation of the weak form of the linear inelastic Boltzmann equation]
\label{THEORDerivationBoltzmann_InelaLinea}
Let $f_0 \in \mathcal{P}\big(\mathbb{R}^d\times \mathbb{R}^d\big) \cap \mathcal{M}_{+,1}\big(\mathbb{R}^d\times\mathbb{R}^d\big)$. Let us assume in addition that there exists a constant $p > 1$ such that:
\begin{align}
\label{EQUATMomntExponDonneInitiBoltzInelaLinea}
\int_{\mathbb{R}^d_x} \hspace{-1mm} \int_{\mathbb{R}^d_v} e^{\vert v \vert^p} f_0(\dd x,\dd v) < +\infty,
\end{align}
and let $f \in \mathcal{C} \big([0,+\infty[,\mathcal{M}_+(\mathbb{R}^d\times\mathbb{R}^d)\big)$ be the unique weak solution of the linear inelastic Boltzmann equation \eqref{EQUATLineaBoltzSphDuFormeForte} with initial datum $f_0$. Then, $f_\varepsilon(t,\cdot,\cdot)$ converges weakly$-*$ towards $f(t,\cdot,\cdot)$ in the Boltzmann-Grad limit $\varepsilon \rightarrow 0$ with \eqref{EQUATBoltzmann-Grad_Limit} holding true, that is, for any $t \geq 0 $ and for any test function $\varphi \in \mathcal{C}_0(\mathbb{R}^d\times\mathbb{R}^d)$, we have:\\
\begin{align}
\label{EQUATConvergencFinalTheor}
\int_{\mathbb{R}^d_x} \hspace{-1mm} \int_{\mathbb{R}^d_v} \varphi(x,v) f_\varepsilon(t,\dd x, \dd v) \underset{\varepsilon \rightarrow 0}{\longrightarrow} \int_{\mathbb{R}^d_x} \hspace{-1mm} \int_{\mathbb{R}^d_v} \varphi(x,v) f(t,\dd x,\dd v),
\end{align}
where $f_\varepsilon$ is the microscopic solution of the inelastic Lorentz model  
starting from $f_0$ defined in \eqref{defeq:f_ep}, Definition \ref{def:f_ep}.\\
The rate of convergence in \eqref{EQUATConvergencFinalTheor} is explicit. More precisely, there exists a universal constant $\varepsilon_0 > 0$ such that, for any $0 < \varepsilon \leq \varepsilon_0$, and for any $t_0 > 0$, $\varphi \in \mathcal{C}_0\big(\mathbb{R}^d\times\mathbb{R}^d)$, we have:
\begin{align}
\Big\vert \int_{\mathbb{R}^d_x} \hspace{-1mm} \int_{\mathbb{R}^d_v} \varphi(x,v) f_\varepsilon(t,\dd x, \dd v) - \int_{\mathbb{R}^d_x} \hspace{-1mm} \int_{\mathbb{R}^d_v} \varphi(x,v) f(t,\dd x,\dd v) \Big\vert \leq C_\text{final}\, \varepsilon^{1/4},
\end{align}
where $C_\text{final} = C_\text{final}(d,r,p,t_0,f_0,\varphi)$ is a constant that depends only on the dimension $d$, the restitution coefficient $r$, the exponential weight $p$, the initial datum $f_0$ and the test function $\varphi$. The constant $C_\text{final}$ depends on $f_0$ and $\varphi$ only via the exponential moment \eqref{EQUATMomntExponDonneInitiBoltzInelaLinea} of $f_0$, and the supremum norm $\vertii{\varphi}_\infty$ of $\varphi$.
\end{theor}

\noindent
The existence and uniqueness of the weak solution to the inelastic linear Boltzmann equation \eqref{EQUATInelaLineaBoltzFinal} will be established in Section \ref{SSECTExistUniquSolut}, respectively in Proposition \ref{PROPOExistWeak_Solut} and \ref{PROPOUniquWeak_Solut} .

\begin{remar}
Observe that the weak$-*$ convergence stated in Theorem \ref{THEORDerivationBoltzmann_InelaLinea} is equivalent to:
\begin{align}
\int_{\mathbb{R}^d_x} \hspace{-1mm} \int_{\mathbb{R}^d_v} \mathds{1}_A(x,v) f_\varepsilon(\dd x, \dd v) \underset{\varepsilon \rightarrow 0}{\longrightarrow} \int_{\mathbb{R}^d_x} \hspace{-1mm} \int_{\mathbb{R}^d_v} \mathds{1}_A(x,v) f(t,\dd x,\dd v)
\end{align}
for any measurable set $A$ of $\mathbb{R}^d \times \mathbb{R}^d$, which is also equivalent to state that the sequence of measures $\big(f_\varepsilon\big)_\varepsilon$ converges towards the measure $f$, in the weak sense in terms of probability theory (see for instance \cite{Bill999}).
\end{remar}

\subsection{Strategy of the proof}
We summarize here the main steps needed to prove the main result of this paper, i.e. Theorem \ref{THEORDerivationBoltzmann_InelaLinea}. In a first part, we motivate the formal definition of the objects we will consider, and in a second part, we will describe the detailed plan of the proof of Theorem \ref{THEORDerivationBoltzmann_InelaLinea}, decomposed into several intermediate results.

\subsubsection{The main ideas and objects behind the proof}
\label{SSSeMain_Ideas}

\paragraph{The microscopic distribution function.} The arguments in the present Section \ref{SSSeMain_Ideas} have to be understood at a formal level. We will make all the steps rigorous in the sequel. 
Our aim is to perform the rigorous derivation of the inelastic linear Boltzmann equation \eqref{EQUATLineaBoltzSphDuFormeForte} from the microscopic inelastic Lorentz model. More precisely, we will prove that a solution $f$ of \eqref{EQUATLineaBoltzSphDuFormeForte} can be approximated by the distribution function $f_\varepsilon$ of a microscopic tagged particle, that collides inelastically with a set of scatterers of size $\varepsilon$, when $\varepsilon \rightarrow 0$, assuming that the distribution of the scatterers obeys a Poisson process of intensity $\mu_\varepsilon$ with $\mu_\varepsilon \varepsilon^{d-1} = 1$.\\
The natural object to be considered as distribution  function $f_\varepsilon$ 
is:
\begin{align}
\label{EQUATDefinDistrLorenMicro}
f_\varepsilon(x,v) = \mathbb{E}_{\mu_\varepsilon}\Big[ \big\vert J\big( T_{c,\varepsilon}^{-t}(x,v)\big) \big\vert \cdot f_0\big( T_{c,\varepsilon}^{-t}(x,v)\big) \Big],
\end{align}
where $f_0$ is the distribution function of the tagged particle in the phase space at the initial time, and $T_{c,\varepsilon}^{-t}(x,v)$ is the flow of the tagged particle, introduced in Definition \ref{DEFINFlot_TempsCrois}. As already mentioned, \eqref{EQUATDefinDistrLorenMicro} is the generalization of the classical distribution function introduced by Gallavotti in 
\cite{Gallavotti} to the inelastic case. Observe, nevertheless,  that \eqref{EQUATDefinDistrLorenMicro} is defined relying on the backward flow, which is not clearly defined a priori (see Remark \ref{REMARBackw_Flow}).\\
The quantity $f_\varepsilon$ defined in \eqref{EQUATDefinDistrLorenMicro} corresponds to the density, with respect to the Lebesgue measure, of the measure describing the probability to find at time $t$ the tagged particle in the elementary volume $x + \dd x$, with a velocity lying in $v + \dd v$. The Jacobian of the dynamical flow has to be considered, because the measure in the phase space is not conserved along the dynamics of the tagged particle, due to the dissipative nature of the collisions.

\paragraph{The approach to perform the derivation of the weak form.} In this paper, we focus on the derivation of the weak solutions $f$ of \eqref{EQUATLineaBoltzSphDuFormeForte}. That is, for any time $t \geq 0$ and any test function $\varphi=\varphi(x,v)$ that is regular enough, we consider the quantity $I_\varepsilon(\varphi,t)$ defined as:
\begin{align}
\overline{I}_\varepsilon(\varphi,t) = \int_{\mathbb{R}^d_x} \hspace{-1mm} \int_{\mathbb{R}^d_v} f_\varepsilon(t,x,v) \varphi(x,v) \dd v \dd x = \int_{\mathbb{R}^d_x} \hspace{-1mm} \int_{\mathbb{R}^d_v} \mathbb{E}_{\mu_\varepsilon} \bigg[ \vert J\big( T^{-t}_{c,\varepsilon}(x,v) \big) \vert \cdot f_0\big( T^{-t}_{c,\varepsilon}(x,v) \big) \bigg] \varphi(x,v) \dd v \dd x,
\end{align}
assuming that we can indeed consider the quantity \eqref{EQUATDefinDistrLorenMicro}.\\
Using Fubini's theorem, and the change of variables $(x,v) \rightarrow T^t_{c,\varepsilon}(x,v)$ in $\mathbb{R}^d_x \times \mathbb{R}^d_v$ for distributions of scatterers $c$ fixed, we find:
\begin{align}
\label{EQUATCalcuFormlObjet_Test}
\overline{I}_\varepsilon(\varphi,t) &= \mathbb{E}_{\mu_\varepsilon} \Bigg[ \int_{\mathbb{R}^d_x} \hspace{-1mm} \int_{\mathbb{R}^d_v} \vert J\big( T^{-t}_{c,\varepsilon}(x,v) \big) \vert \cdot f_0\big( T^{-t}_{c,\varepsilon}(x,v) \big) \varphi(x,v) \dd v \dd x \Bigg] \nonumber\\
&= \mathbb{E}_{\mu_\varepsilon} \Bigg[ \int_{\mathbb{R}^d_x} \hspace{-1mm} \int_{\mathbb{R}^d_v} f_0(x,v) \varphi\big(T^t_{c,\varepsilon}(x,v)\big) \dd v \dd x \Bigg] = \int_{\mathbb{R}^d_x} \hspace{-1mm} \int_{\mathbb{R}^d_v} f_0(x,v) \mathbb{E}_{\mu_\varepsilon} \big[ \varphi\big(T^t_{c,\varepsilon}(x,v)\big) \big] \dd v \dd x.
\end{align}
Therefore, the object we will consider at the level of the microscopic particle system to prove the weak convergence will be:
\begin{align}
\label{EQUATQuantNotre_Cas_Faibl}
\varphi_\varepsilon(t,x,v) = \mathbb{E}_{\mu_\varepsilon} \big[ \varphi \big(T^t_{c,\varepsilon}(x,v) \big) \big],
\end{align}
so that the weak convergence of $f_\varepsilon$ towards $f$, by testing against any test function $\varphi$, is obtained by comparing the two integrals:
\begin{align}
\label{EQUAT_Integrals_Weak_Conv}
I_0(\varphi,t) = \int_{\mathbb{R}^d_x}\int_{\mathbb{R}^d_v} f(t,x,v) \varphi(x,v) \dd v \dd x \hspace{5mm} \text{and} \hspace{5mm} I_\varepsilon(\varphi,t) = \int_{\mathbb{R}^d_x}\int_{\mathbb{R}^d_v} f_0(x,v) \varphi_\varepsilon(t,x,v) \dd x \dd v,
\end{align}
with $\varphi_\varepsilon$ defined by \eqref{EQUATQuantNotre_Cas_Faibl}. Let us also observe that, in the particular case when the test function $\varphi$ is chosen to be the indicator function $\mathds{1}_A$ of a measurable set $A$ of $\mathbb{R}^d\times\mathbb{R}^d$, we have the elementary identity concerning $\varphi_\varepsilon$:
\begin{align}
\label{EQUATQuant_Cas_Indic_Part}
\varphi_\varepsilon(t,x,v) = \mathbb{E}_{\mu_\varepsilon} \big[ \mathds{1}_A \big(T^t_{c,\varepsilon}(x,v) \big) \big] = \int_C \mathds{1}_A \big(T^t_{c,\varepsilon}(x,v) \big) \dd \mathbb{P}_{\mu_\varepsilon}(c) = \mathbb{P}_{\mu_\varepsilon}\big( \{c \in C\ /\ T^t_{c,\varepsilon}(x,v) \in A \} \big).
\end{align}
In turn, observe that introducing the expressions \eqref{EQUATQuantNotre_Cas_Faibl}, \eqref{EQUATQuant_Cas_Indic_Part} of $\varphi_\varepsilon$, together with $I_\varepsilon$, allows to define a time-dependent family of measures $f_\varepsilon(t)$ via its action on the test functions $\varphi$ relying on the duality relation:
\begin{align}
\langle f_\varepsilon(t),\varphi \rangle = I_\varepsilon(\varphi,t).
\end{align}
Considering the quantities \eqref{EQUATQuantNotre_Cas_Faibl} or \eqref{EQUATQuant_Cas_Indic_Part} has the great advantage to rely on the forward flow of the tagged particle, allowing therefore to consider objects that are more regular than \eqref{EQUATDefinDistrLorenMicro}. This discussion motivates the introduction of the quantity $f_\varepsilon$ by the action described in \eqref{defeq:f_ep}, Definition \ref{def:f_ep}.

\paragraph{Relying on the adjoint equation of \eqref{EQUATLineaBoltzSphDuFormeForte}.} In the same way as we rewrote the integral $\overline{I}_\varepsilon(\varphi,t)$ in a more regular object by transferring the dynamical flow to the test function $\varphi$, we can also take advantage of the definition of the weak solutions of \eqref{EQUATLineaBoltzSphDuFormeForte}. In particular, we observe that if $f$ is a weak solution of the linear inelastic Boltzmann equation \eqref{EQUATLineaBoltzSphDuFormeForte} in the sense of Definition \ref{DEFINSolutFaibleEquatBoltzLineaInela}, and if the function $\widetilde{\varphi}$ solves the \emph{adjoint equation} of \eqref{EQUATLineaBoltzSphDuFormeForte}:
\begin{align}
\label{EQUATLineaBoltzSphDuFormeAdjoi}
\partial_t \widetilde{\varphi} - v \cdot \nabla_x \widetilde{\varphi} = \int_{\mathbb{S}^{d-1}_\omega} \vert v \cdot \omega \vert \big[ \widetilde{\varphi}(v') - \widetilde{\varphi}(v) \big] \dd \omega
\end{align}
with initial datum $\widetilde{\varphi}(0,\cdot,\cdot) = \varphi$ (where $v' = \kappa_\omega(v)$ is defined in \eqref{EQUATScattering}), then
\begin{align}
\label{EQUATRegulTransfert_t=0_T}
\int_{\mathbb{R}^d_x} \hspace{-1mm} \int_{\mathbb{R}^d_v} \varphi(x,v) f(t,\dd x,\dd v) = \int_{\mathbb{R}^d_x} \hspace{-1mm} \int_{\mathbb{R}^d_v} \widetilde{\varphi}(t,x,v) f_0(\dd x,\dd v).
\end{align}

\noindent
In the end, we will prove:
\begin{align}
\label{EQUATMain_Convergenc}
\int_{\mathbb{R}_x^d}\int_{\mathbb{R}^d_v} \varphi_\varepsilon(t,x,v) f_0(\dd x,\dd v) \underset{\varepsilon \rightarrow 0}{\longrightarrow} \int_{\mathbb{R}^d_x}\int_{\mathbb{R}^d_v} \widetilde{\varphi}(t,x,v) f_0(\dd x,\dd v),
\end{align}
with $\varphi_\varepsilon$ as in \eqref{EQUATQuantNotre_Cas_Faibl} and $\widetilde{\varphi}$ solving \eqref{EQUATLineaBoltzSphDuFormeAdjoi}, relying on explicit expressions of the dual quantities $\varphi_\varepsilon$ and $\widetilde{\varphi}$.

\paragraph{A brief comparison with the Kolmogorov equations.} The quantities \eqref{EQUATDefinDistrLorenMicro} and \eqref{EQUATQuantNotre_Cas_Faibl} have strong links with the theory of the backward and forward Kolmogorov equations. Indeed, for a continuous Markov process $X(\tau)$ on a certain probability space $\mathcal{S}$, with $0 \leq \tau \leq t$, if we define $u$ and $f$, respectively, by
\begin{align}
u(\tau,s) = \mathbb{P}(X(\tau) = s) \hspace{5mm} \text{and} \hspace{5mm} f(\tau,s) = \mathbb{E}[ g(X(t)) \vert X(\tau) = s]
\end{align}
where $g$ is an arbitrary function, it is well-known that $u$ and $f$ solve respectively the forward and the backward Kolmogorov equations, and that this two quantities are linked by the duality relation:
\begin{align}
\sum_{s \in \mathcal{S}} u(\tau,s)f(\tau,s) = \sum_{s \in \mathcal{S}} u(0,s)f(0,s) \hspace{3mm} \forall\, \tau \in [0,t].
\end{align}
In our case, the continuous process is given by the flow $T^\tau_{c,\varepsilon}$, and the probability space is replaced by the phase space $\mathbb{R}^d_x \times \mathbb{R}^d_v$. Considering as a test function $\varphi$ the indicator function $\mathds{1}_A$ of a subset $A$ of the phase space, we have:
\begin{align}
\varphi_\varepsilon(\tau,x,v) = \mathbb{E}_{\mu_\varepsilon} \big[ \mathds{1}_A\big( T^\tau_{c,\varepsilon}(x,v) \big) \big] = \mathbb{P}_{\mu_\varepsilon} \big( T^\tau_{c,\varepsilon}(x,v) \in A \big),
\end{align}
which corresponds clearly to the quantity $u$, solving the forward Kolmogorov equation. In the kinetic limit, $u$ corresponds to the solution $\varphi$ of the adjoint equation \eqref{EQUATLineaBoltzSphDuFormeAdjoi}, and the relation \eqref{EQUATRegulTransfert_t=0_T} that the solution $\varphi$ satisfies with a solution $f$ of the linear inelastic Boltzmann equation can be interpreted as the duality relation between two solutions of a pair of backward and forward Kolmogorov equations. At the kinetic level, the limiting trajectories of the tagged particle are indeed Markov processes.\\
At the level of the particles, the process described by the dynamical flow cannot be Markovian anymore, because the scatterers have a positive size, allowing recollisions to take place. Nevertheless, considering anyway the duality bracket between $f_\varepsilon(\tau)$ and $\varphi_\varepsilon(t-\tau,\cdot,\cdot)$, we have:
\begin{align}
\int_{\mathbb{R}^d_x} \hspace{-1mm} \int_{\mathbb{R}^d_v} f_\varepsilon(\tau,x,v) \varphi_\varepsilon(t-\tau,x,v) \dd x \dd v &= \int_{\mathbb{R}^d_x} \hspace{-1mm} \int_{\mathbb{R}^d_v} \mathbb{E}_{\mu_\varepsilon} \bigg[ \vert J\big(T^{-\tau}_{c,\varepsilon}(x,v)\big) \vert \cdot f_0 \big( T^{-\tau}_{c,\varepsilon}(x,v) \big) \bigg] \nonumber\\
&\hspace{40mm} \cdot \mathbb{E}_{\mu_\varepsilon} \big[ \mathds{1}_A\big( T^{t-\tau}_{c,\varepsilon}(x,v) \big) \big] \dd v \dd x.
\end{align}
First, we observe that if $\tau = 0$ or $\tau = t$, one of the two expected values simplifies (because either $T^{-\tau}_{c,\varepsilon}$ or $T^{t-\tau}_{c,\varepsilon}$ becomes the identity, so it becomes independent from the distributions $c$ of the scatterers), and the equality between the duality brackets at time $\tau = 0$ and at time $\tau = t$ corresponds to the formal computation \eqref{EQUATCalcuFormlObjet_Test}.\\
In the other cases, that is when $0 < \tau < t$, we observe that if the covariance between the random variables $\vert J\big(T^{-\tau}_{c,\varepsilon}(x,v)\big) \vert \cdot f_0 \big( T^{-\tau}_{c,\varepsilon}(x,v) \big)$ and $\mathds{1}_A\big( T^{t-\tau}_{c,\varepsilon}(x,v) \big)$ is zero, then we can write:
\begin{align}
\int_{\mathbb{R}^d_x} \hspace{-1mm} \int_{\mathbb{R}^d_v} f_\varepsilon(t,x,v) \varphi_\varepsilon(t-\tau,x,v) \dd x \dd v &= \int_{\mathbb{R}^d_x} \hspace{-1mm} \int_{\mathbb{R}^d_v} \mathbb{E}_{\mu_\varepsilon} \bigg[ \vert J\big(T^{-\tau}_{c,\varepsilon}(x,v)\big) \vert \cdot f_0 \big( T^{-\tau}_{c,\varepsilon}(x,v) \big) \cdot \mathds{1}_A\big( T^{t-\tau}_{c,\varepsilon}(x,v) \big) \bigg] \dd v \dd x \nonumber\\
&= \mathbb{E}_{\mu_\varepsilon} \Big[ \int_{\mathbb{R}^d_x} \hspace{-1mm} \int_{\mathbb{R}^d_v} \vert J\big(T^{-\tau}_{c,\varepsilon}(x,v)\big) \vert \cdot f_0 \big( T^{-\tau}_{c,\varepsilon}(x,v) \big) \cdot \mathds{1}_A\big( T^{t-\tau}_{c,\varepsilon}(x,v) \big) \dd v \dd x \Big] \nonumber\\
&= \mathbb{E}_{\mu_\varepsilon} \Big[ \int_{\mathbb{R}^d_x} \hspace{-1mm} \int_{\mathbb{R}^d_v} f_0 (x,v) \cdot \mathds{1}_A\big( T^t_{c,\varepsilon}(x,v) \big) \dd v \dd x \Big],
\end{align}
which is a quantity independent from $\tau$.\\
However, there is no reason in general for the covariance between the two random variables $\vert J\big(T^{-\tau}_{c,\varepsilon}(x,v)\big) \vert \cdot f_0 \big( T^{-\tau}_{c,\varepsilon}(x,v) \big)$ and $\mathds{1}_A\big( T^{t-\tau}_{c,\varepsilon}(x,v) \big)$ to be zero. Nevertheless, if we assume that we consider only distributions $c$ of scatterers that give rise only to trajectories without recollisions, then $\vert J\big(T^{-\tau}_{c,\varepsilon}(x,v)\big) \vert \cdot f_0 \big( T^{-\tau}_{c,\varepsilon}(x,v) \big)$ and $\mathds{1}_A\big( T^{t-\tau}_{c,\varepsilon}(x,v) \big)$ are indeed independent random variables, because the first one involves the backward flow, while the second involves the forward flow, both starting from the same point $(x,v)$ of the phase space.\\
We recover therefore that in the absence of recollisions, the particle system behaves as a Markov process. We observe finally that the probability of recollisions vanishes in the Boltzmann-Grad limit, at least in the elastic case, which is consistent with the limit process that we recover in the kinetic limit.

\subsubsection{The steps of the proof of Theorem \ref{THEORDerivationBoltzmann_InelaLinea}}

To establish the weak convergence stated in Theorem \ref{THEORDerivationBoltzmann_InelaLinea}, we will proceed as follows. We consider any test function $\varphi \in \mathcal{C}_0(\mathbb{R}^d\times\mathbb{R}^d)$ of the dual of the space of finite Radon measures $\mathcal{M}(\mathbb{R}^d \times \mathbb{R}^d)$, and we define the quantities $I_0(\varphi,t)$ and $I_\varepsilon(\varphi,t)$ as in \eqref{EQUAT_Integrals_Weak_Conv}, \eqref{EQUATQuantNotre_Cas_Faibl}. We recall that the main content of Theorem \ref{THEORDerivationBoltzmann_InelaLinea} is the convergence $I_\varepsilon(\varphi,t) {\longrightarrow} I_0(\varphi,t)$, that is, $f_\varepsilon {\overset{*}{\rightharpoonup}} f$, as $\varepsilon \rightarrow 0$.

\begin{enumerate}
\item We first prove that the \textbf{forward} dynamics of the particle system is well-posed, globally in time, a.e.. This step is necessary to give sense to $\varphi_\varepsilon$ as defined in \eqref{EQUATQuantNotre_Cas_Faibl}, and used in \eqref{EQUAT_Integrals_Weak_Conv} to define $I_\varepsilon(\varphi,t)$. This is the content of Section \ref{SECTIWell_PosedDynam} and, specifically, Proposition \ref{PROPODefinGlobaDynam}.
\item We rewrite $\varphi_\varepsilon$ in a series form, where the index of each term corresponds to the number of obstacles that are collided by the tagged particle during the time interval $[0,t]$ when evolving according to the dynamics of $T_{c,\varepsilon}^t(x,v)$, and where the integration variables are the positions $c_k$ of the different obstacles that are collided. A small remainder will appear, corresponding on the one hand to a pathological distribution of the scatterers (pathological in the sense that the dynamics is not well-defined), and corresponding on the other hand to distributions of scatterers such that the initial position of the tagged particle at time $t=0$ lies inside one of the obstacles. This is the content of Proposition \ref{PROPOReecrPhiEpIntegObsta}.
\item In order to compare with the series expression of $\varphi_\varepsilon$, we write the solution $\widetilde{\varphi}$ of the adjoint equation \eqref{EQUATLineaBoltzSphDuFormeAdjoi} as a series as well. This is Proposition \ref{PROPORepreSerieEquatAdjoi}.\\
The series representation of $\widetilde{\varphi}$ involves only the initial datum $\varphi = \widetilde{\varphi}(0,\cdot,\cdot)$, and makes sense provided that this initial datum is continuous. Considering then the integral $I_0(\varphi,s) = \int_{\mathbb{R}^d\times\mathbb{R}^d} \widetilde{\varphi}(s,x,v)f_0(\dd x,\dd v)$ for any $s \in [0,t]$, we introduce this way a time-dependent family of measures $g(s) \in \mathcal{M}_+\big(\mathbb{R}^d \times \mathbb{R}^d\big)$ (for any $s \in [0,s]$) defined by duality as $g(s): \varphi \mapsto I_0(\varphi,s)$.
\item Restarting from the series representation of $\varphi_\varepsilon$, we isolate now the distributions of scatterers that lead to a recollision. This provides an additional remainder, defined and estimated in Proposition \ref{PROPOReecrPhiEpElimiRecol}. 
In addition, we rewrite the terms of the new series representation in $(t,\omega)$-coordinates, which almost completes the comparison with the series representation of the solution $\widetilde{\varphi}$ of the adjoint equation.\\
It is in this step that is done the careful (geometric) estimate on the domains of integration removed to prevent recollisions.
\item Finally, we carefully study the remaining differences between the expressions of the solution of the adjoint equation, and the last version of $\varphi_\varepsilon$ written in $(t,\omega)$-coordinates. Specifically, here it remains to estimate on the one hand the difference between the two dynamical tubes that arise. On the other hand, we estimate the error caused by the cut-offs in the integration domains that were introduced in the previous step.
\item We can then conclude the proof of Theorem \ref{THEORDerivationBoltzmann_InelaLinea}. Indeed, the previous steps provide that $f_\varepsilon(t)$ converges weakly towards the measure $g(t)$ defined in the third step. In addition, observing that $g$ is also a weak solution to the linear inelastic Boltzmann equation \eqref{EQUATLineaBoltzSphDuFormeForte} (Proposition \ref{PROPOExistWeak_Solut}), and that such a weak solution is unique once the initial datum is given (Proposition \ref{PROPOUniquWeak_Solut}), the proof of Theorem \ref{THEORDerivationBoltzmann_InelaLinea} follows.
\end{enumerate}

\begin{remar}
Our constructive approach involves expanding $f_\varepsilon$ into a series and comparing it with the series representation of the solution to the linear Boltzmann equation, as provided by the iterated Duhamel formula. This methodology aligns with the one that has been originally proposed in the pioneering work of Gallavotti \cite{Gallavotti}. On the other hand, the idea of employing weak solutions, tested against solutions to the adjoint equation, has been explored and discussed in \cite{NoVe017}.  
In that context, the coalescing process prevents to determine uniquely the past events that lead to the formation of a cluster of particles. In the present case, the singularity comes from the possibility that the inelastic collapse can take place in the past of a trajectory of the tagged particle. Considering weak solutions enables the analysis of the forward dynamics of the particle system, which is well-defined under the conditions considered in \cite{NoVe017}, as well as in our case. 
\end{remar}

\section{Properties of the inelastic linear Boltzmann equation}

In this section, we present the key results concerning the inelastic linear Boltzmann equation \eqref{EQUATInelaLineaBoltzFinal} that will be used throughout the rest of the article. We discuss in particular the existence and uniqueness of solutions (strong or weak), as well as the decay of kinetic energy.

\subsection{Duhamel formula for the solutions of \eqref{EQUATLineaBoltzSphDuFormeForte}}
\label{SSECTDerivFormuSerie}

Let us now provide a representation of the solutions of the linear Boltzmann equation \eqref{EQUATLineaBoltzSphDuFormeForte}. In the present Section \ref{SSECTDerivFormuSerie}, the computations are formal. We will investigate the convergence of the series representation in Section \ref{SSECTDiscuConveSerieRepre}. 
First, we define the quantity:
\begin{align}
g(t,x,v) = f(t,x+tv,v),
\end{align}
where $f$ is assumed to be a solution of \eqref{EQUATLineaBoltzSphDuFormeForte}. We have:
\begin{align}
\label{EQUATDiffeFonct__g__}
\partial_t g(t,x,v) + C_d \vert v \vert g(t,x,v) = \int_\omega \frac{\vert v\cdot\omega \vert}{r^2} g(t,x+t(v-\hspace{0.25mm}'\hspace{-0.75mm}v),'\hspace{-1mm}v) \dd \omega, 
\end{align}
with $C_d$ defined by \eqref{EQUATDefin_C_d_}. Integrating in time \eqref{EQUATDiffeFonct__g__}, we find:
\begin{align}
g(t,x,v) = e^{-C_d \vert v \vert t} g(0,x,v) + \int_0^t e^{-C_1 \vert v \vert (t-s)} \int_\omega \frac{\vert v\cdot\omega \vert}{r^2} g(s,x + s(v-\hspace{0.25mm}'\hspace{-0.75mm}v),\hspace{0.25mm}'\hspace{-0.5mm}v)\dd\omega \dd s.
\end{align}
We obtain after infinitely many iterations:
\begin{align}
\label{EQUATRepreSerieSolutTransBoltzLineaInela}
g(t,x,v) &= e^{-C_d \vert v \vert t} g(0,x,v) \nonumber\\
&\hspace{3mm}+ e^{-C_d\vert v \vert t}\sum_{k=1}^{+\infty} \int_{t_1=0}^t\int_{\omega_1 \in \mathbb{S}^{d-1}} \dots \int_{t_k=0}^{t_{k-1}}\int_{\omega_k \in \mathbb{S}^{d-1}} e^{C_d \sum_{j=1}^k t_j\left[ \vert v^{-(j-1)} \vert - \vert v^{-j} \vert \right]} \nonumber\\
&\hspace{16mm} \times \left( \prod_{l=1}^k \frac{\vert v^{-(l-1)}\cdot\omega_l \vert}{r^2} \right) g(0,x + \sum_{m=1}^k \left[ t_m(v^{-(m-1)} - v^{-m})\right],v^{-k}) \dd \omega_k \dd t_k \dots \dd \omega_1 \dd t_1.
\end{align}
Here, we denoted by $v^{-0}$ the velocity $v$, and $\hspace{0.25mm}'\hspace{-0.25mm}v$ is denoted by $v^{-1}$, i.e. $\kappa_{\omega_1}(v^{-1}) = v$ and, more generally,
\begin{align}
\kappa_{\omega_k}(v^{-k}) = v^{-(k-1)}.
\end{align}
\noindent
Back to the solution $f$ of \eqref{EQUATLineaBoltzSphDuFormeForte}, we obtain the following series representation:

\begin{align}
\label{EQUATRepreSerieSolution__BoltzLineaInela}
f(t,x,v) &= e^{-C_d \vert v \vert t} f_0(x-tv,v) \nonumber\\
&\hspace{5mm}+ e^{-C_d\vert v \vert t}\sum_{k=1}^{+\infty} \int_{t_1=0}^t\int_{\omega_1 \in \mathbb{S}^{d-1}} \dots \int_{t_k=0}^{t_{k-1}}\int_{\omega_k \in \mathbb{S}^{d-1}} e^{C_d \sum_{j=1}^k t_j\left[ \vert v^{-(j-1)} \vert - \vert v^{-j} \vert \right]} \left( \prod_{l=1}^k \frac{\vert v^{-(l-1)}\cdot\omega_l \vert}{r^2} \right) \nonumber\\
&\hspace{40mm} \times f_0(x - tv + \sum_{m=1}^k \left[ t_m(v^{-(m-1)} - v^{-m})\right],v^{-k}) \dd \omega_k \dd t_k \dots \dd \omega_1 \dd t_1.
\end{align}

\subsection{On the convergence of the series representation}
\label{SSECTDiscuConveSerieRepre}

\noindent
In the present section we will discuss the convergence of the series \eqref{EQUATRepreSerieSolution__BoltzLineaInela}. One important difference between the inelastic and elastic cases is the presence of the product $\prod_{l=1}^k \frac{\vert v^{-(l-1)}\cdot\omega_l \vert}{r^2}$ in the integrand, which prevents to obtain the convergence when assuming only $f_0 \in L^\infty$. Note that in absence of $\vert v^{-(l-1)}\cdot\omega_l \vert$ in the numerator, the term $1/r^{2k} = \left( 1/r^2 \right)^k$ could be absorbed as in the elastic case, in an exponential. However, here, the norm of $v^{-(l-1)}$ is causing an additional divergence, which grows fast. Indeed, depending on the angular parameters $\omega_l$, one can have:
\begin{align}
\vert \hspace{0.25mm}'\hspace{-0.5mm}v \vert = \frac{1}{r} \vert v \vert,\hspace{3mm} \vert v^{-2} \vert = \frac{1}{r} \vert \hspace{0.25mm}'\hspace{-0.5mm}v \vert = \frac{1}{r^2} \vert v \vert, \dots
\end{align}
so that in the end, the product $\prod_{l=1}^k \frac{\vert v^{-(l-1)}\cdot\omega_l \vert}{r^2}$ grows like $(1/r)^{k^2}$. On the other hand, for $a,b > 0$, the series $
\sum_{n \geq 0} \frac{a^n b^{n^2}}{n!}$ is converging only if $b \leq 1$. In this case, $b = 1/r^2$, so the series \eqref{EQUATRepreSerieSolution__BoltzLineaInela}  may not converge without additional assumptions, except in the elastic case.
\\
We will show that the series in \eqref{EQUATRepreSerieSolution__BoltzLineaInela} is converging provided that $f_0$ decays sufficiently fast at infinity. Specifically, we will assume that $f_0$ presents an exponential decay in velocity at infinity. Observe that a Maxwellian decay is not necessary. Actually, any super-polynomial decay in velocity would be enough to conclude.

\begin{theor}[Convergence of the series representation in the space of exponential weights in velocity]
\label{PROPOConveRepreSerieSolutBoltzLineaInela}
Let $f_0: \mathbb{R}^{2d} \rightarrow \mathbb{R}$ be a non-negative, measurable function. Let us assume that there exist $p$ and $\alpha$ two strictly positive real numbers such that:
\begin{align}
\label{EQUATDecaySuperExpon_f_0_}
\text{supess}_{(x,v) \in \mathbb{R}^{2d}} \left\vert f_0(x,v) e^{\alpha \vert v \vert^p} \right\vert < +\infty.
\end{align}
Then, the series given by the formula \eqref{EQUATRepreSerieSolution__BoltzLineaInela} is converging.\\
Therefore, if $f_0 \in \mathcal{C}^1\big(\mathbb{R}^d\times\mathbb{R}^d\big)$ and if \eqref{EQUATDecaySuperExpon_f_0_} holds, there exists a strong solution to the linear inelastic Boltzmann equation \eqref{EQUATLineaBoltzSphDuFormeForte} with initial datum $f_0$, given by the formula \eqref{EQUATRepreSerieSolution__BoltzLineaInela}.
\end{theor}

\begin{proof}
We consider \eqref{EQUATRepreSerieSolution__BoltzLineaInela}. We start with estimating precisely the difference between the norms of $v^{-(j-1)}$ and $v^{-j}$. We have:
\begin{align}
v^{-j} = v^{-(j-1)} - \left(1+\frac{1}{r}\right) v^{-(j-1)}\cdot \omega_j \omega_j,
\end{align}
so that:
\begin{align}
\vert v^{-j} \vert = \vert v^{-(j-1)} \vert \sqrt{1 + \left(\frac{1}{r^2}-1\right) \frac{\left( v^{-(j-1)}\cdot\omega_j\right)^2}{\vert v^{-(j-1)} \vert^2}} \cdotp
\end{align}
The convergence of the series is based on the analysis of two different cases: $\vert \frac{v^{-(j-1)}}{\vert v^{-(j-1)} \vert} \cdot\omega_j \vert \leq \beta$ or $\vert \frac{v^{-(j-1)}}{\vert v^{-(j-1)} \vert} \cdot\omega_j \vert > \beta$, where $0 \leq \beta = \beta(k,l) \leq 1$ will be chosen later. In the second case, one can take advantage of the exponential decay of $f_0$.\\
\newline
Let us consider now the $k$-th term ($k \geq 1$) of the series \eqref{EQUATRepreSerieSolution__BoltzLineaInela}. Each of the $k$ integrals in $\omega_l$ ($1 \leq l \leq k$) is decomposed between the domains $\vert \frac{v^{-(l-1)}}{\vert v^{-(l-1)} \vert}\cdot\omega_l \vert \leq \beta$ and $\vert \frac{v^{-(l-1)}}{\vert v^{-(l-1)} \vert}\cdot\omega_l \vert > \beta$. The $k$-th term of the series is then decomposed into $2^k$ sub-terms, encoded as follows. To each of such sub-terms, we associate the vector $(e_1,\dots,e_k) \in \{0,1\}^k$, where
\begin{align}
\left\{
\begin{array}{ccc}
e_l = 0 &\text{if the integral over }\omega_l\text{ runs over the domain:}& I(e_l) = \big\{ \omega_l \in \mathbb{S}^{d-1}\ /\ \vert \frac{v^{-(l-1)}}{\vert v^{-(l-1)} \vert}\cdot\omega_l \vert \leq \beta \big\},\\
e_l = 1 &\text{if the integral over }\omega_l\text{ runs over the domain:}& I(e_l) = \big\{ \omega_l \in \mathbb{S}^{d-1}\ /\ \vert \frac{v^{-(l-1)}}{\vert v^{-(l-1)} \vert}\cdot\omega_l \vert > \beta \big\}.
\end{array}
\right.
\end{align}
To each of the vectors $(e_1,\dots,e_k)$, we associate its \emph{size} $s$, defined as:
\begin{align}
s\left( \left(e_1,\dots,e_k\right) \right) = \sum_{l=1}^k e_l.
\end{align}
In other words, the size of $(e_1,\dots,e_k)$ counts the number of entries that are equal to $1$. If $e_l = 0$, we have:
\begin{align}
\vert v^{-(l-1)} \vert \leq \vert v^{-l} \vert \leq \sqrt{1 + \left(\frac{1}{r^2}-1\right)\beta^2} \cdot \vert v^{-(l-1)} \vert,
\end{align}
and if $e_l = 1$, we have:
\begin{align}
\sqrt{1 + \left(\frac{1}{r^2}-1\right)\beta^2} \cdot \vert v^{-(l-1)} \vert < \vert v^{-l} \vert \leq \frac{1}{r} \vert v^{-(l-1)} \vert.
\end{align}
For $0 \leq \beta \leq 1$, we define the following quantity:
\begin{align}
q_\beta = \sqrt{1 + \left(\frac{1}{r^2}-1\right)\beta^2} \geq 1,
\end{align}
which is equal to $1$ if $\beta = 0$, strictly larger than $1$ if $\beta \neq 0$, and which is equal to $1/r$ if $\beta = 1$. Finally, let us observe that this quantity is an increasing function of $\beta \in [0,1]$.\\
For any vector $(e_1,\dots,e_k) \in \{0,1\}^k$, we have:
\begin{align}
\label{EQUATTermeSerieGenerOrdre__k__}
\Bigg\vert &\int_{t_1=0}^t\int_{\omega_1 \in I(e_1)} \dots \int_{t_k=0}^{t_{k-1}}\int_{\omega_k \in I(e_k)} e^{C_d \sum_{j=1}^k t_j\left[ \vert v^{-(j-1)} \vert - \vert v^{-j} \vert \right]} \left( \prod_{l=1}^k \frac{\vert v^{-(l-1)}\cdot\omega_l \vert}{r^2} \right) \nonumber\\
&\hspace{60mm} \times f_0(x - tv + \sum_{m=1}^k \left[ t_m(v^{-(m-1)} - v^{-m})\right],v^{-k}) \dd \omega_k \dd t_k \dots \dd \omega_1 \dd t_1 \Bigg\vert \nonumber\\
&\leq \text{supess}_{(x,v) \in \mathbb{R}^{2d}} \left\vert f_0(x,v) e^{\alpha \vert v \vert^p} \right\vert \int_{t_1=0}^t\int_{\omega_1 \in I(e_1)} \dots \int_{t_k=0}^{t_{k-1}}\int_{\omega_k \in I(e_k)} \left( \prod_{l=1}^k \vert v^{-(l-1)} \vert \right) \left( \prod_{l=1}^k \frac{\left\vert \frac{v^{-(l-1)}}{\vert v^{-(l-1)} \vert}\cdot\omega_l \right\vert}{r^2} \right) \nonumber\\
&\hspace{115mm} \times e^{-\alpha \vert v^{-k} \vert^p} \dd \omega_k \dd t_k \dots \dd \omega_1 \dd t_1.
\end{align}

\noindent
Let us consider a vector $(e_1,\dots,e_k) \in \{0,1\}^k$ of length $s$. In this case, since we have:
\begin{align}
\label{EQUATNorme_v^-(l-1)_Produ}
\vert v^{-(l-1)} \vert = \left[ \prod_{m=1}^{l-1} \sqrt{1 + \left(\frac{1}{r^2}-1\right) \left( \frac{v^{-(m-1)}}{\vert v^{-(m-1)} \vert}\cdot \omega_m\right)^2} \right] \vert v \vert.
\end{align}
In the product:
\begin{align}
\label{EQUATProduNormeVites}
\prod_{l=1}^k \vert v^{-(l-1)} \vert
\end{align}
the scalar product $\frac{v}{\vert v \vert}\cdot \omega_1$ appears $(k-1)$ times, the scalar product $\frac{v^{-1}}{\vert v^{-1} \vert}\cdot \omega_2$ appears $(k-2)$ times, and more generally, the scalar product $\frac{v^{-(m-1)}}{\vert v^{-(m-1)} \vert}\cdot \omega_m$ appears $(k-m)$ times. Therefore, the product \eqref{EQUATProduNormeVites} is maximal when the terms that appear the most are the largest. In other words, the product \eqref{EQUATProduNormeVites} is maximal when $e_1 = \dots = e_s = 1$, and $e_{s+1} = \dots = e_k = 0$.\\
As a consequence, we deduce on the one hand (using the convention $\prod_{m=1}^0 u_m = 1$):
\begin{align}
\label{EQUATMajorProduNormeVites_v^-l}
\prod_{l=1}^k \vert v^{-(l-1)} \vert &= \prod_{l=1}^k \left[ \prod_{m=1}^{l-1} \sqrt{1 + \left(\frac{1}{r^2}-1\right) \left( \frac{v^{-(m-1)}}{\vert v^{-(m-1)} \vert}\cdot \omega_m\right)^2} \right] \vert v \vert \nonumber\\
&\leq \left( \prod_{l=1}^{s+1} \left[ \prod_{m=1}^{l-1} \frac{1}{r} \right] \vert v \vert \right) \left( \prod_{l=s+2}^k \left[ \left(\prod_{m=1}^{s} \frac{1}{r} \right) \left( \prod_{m=s+1}^{l-1} q_{\beta_k} \right) \right] \vert v \vert \right) \nonumber\\
&\leq \left( \prod_{l=1}^{s+1} \frac{1}{r^{l-1}} \vert v \vert \right) \left( \prod_{l=s+2}^k \frac{1}{r^s}q_{\beta_k}^{l-s-1} \vert v \vert \right) = \frac{1}{r^{\sum_{l=1}^{s+1}(l-1) + \sum_{l=s+2}^k s}} q_{\beta}^{\sum_{l=s+2}^k (l-s-1)} \vert v \vert^k \nonumber\\
&\leq \frac{1}{r^{\left[\frac{s(s+1)}{2} + (k-s-1)s\right]}} q_{\beta}^\frac{(k-s-1)(k-s)}{2} \vert v \vert^k \leq \left( \frac{1}{r} \right)^{ks} q_{\beta}^\frac{(k-s)^2}{2} \vert v \vert^k.
\end{align}
On the other hand, from \eqref{EQUATNorme_v^-(l-1)_Produ} we deduce also:
\begin{align}
\label{EQUATMinorProduNormeVites_v^-l}
\vert v^{-k} \vert = \left[ \prod_{m=1}^{k} \sqrt{ 1 + \left( \frac{1}{r^2} - 1 \right) \left( \frac{v^{-(m-1)}}{\vert v^{-(m-1)} \vert}\cdot\omega_m \right)^2 } \right] \vert v \vert \geq q_{\beta}^s \vert v \vert.
\end{align}
Finally, we have:
\begin{align}
\label{EQUATMajorProduProduScala}
\prod_{l=1}^k \left\vert \frac{v^{-(l-1)}}{\vert v^{-(l-1)} \vert} \cdot \omega_l \right\vert \leq \beta^{k-s}.
\end{align}
Gathering \eqref{EQUATMajorProduNormeVites_v^-l}, \eqref{EQUATMinorProduNormeVites_v^-l} and \eqref{EQUATMajorProduProduScala}, we obtain the following upper bound on the elementary term \eqref{EQUATTermeSerieGenerOrdre__k__}, which is the part associated to the vector $(e_1,\dots,e_k)$ (of length $s$) of the $k$-th term of the series representation \eqref{EQUATRepreSerieSolution__BoltzLineaInela}:
\begin{align}
\label{EQUATTermeSerieGenerOrdre__k_2}
\Bigg\vert &\int_{t_1=0}^t\int_{\omega_1 \in I(e_1)} \dots \int_{t_k=0}^{t_{k-1}}\int_{\omega_k \in I(e_k)} e^{C_d \sum_{j=1}^k t_j\left[ \vert v^{-(j-1)} \vert - \vert v^{-j} \vert \right]} \left( \prod_{l=1}^k \frac{\vert v^{-(l-1)}\cdot\omega_l \vert}{r^2} \right) \nonumber\\
&\hspace{60mm} \times f_0(x - tv + \sum_{m=1}^k \left[ t_m(v^{-(m-1)} - v^{-m})\right],v^{-k}) \dd \omega_k \dd t_k \dots \dd \omega_1 \dd t_1 \Bigg\vert \nonumber\\
&\leq \left(\text{supess}_{(x,v) \in \mathbb{R}^{2d}} \left\vert f_0(x,v) e^{\alpha \vert v \vert^p} \right\vert \right) \int_{t_1=0}^t\int_{\omega_1 \in I(e_1)} \dots \int_{t_k=0}^{t_{k-1}}\int_{\omega_k \in I(e_k)}
\left(\frac{1}{r}\right)^{ks} q_{\beta}^{\frac{(k-s)^2}{2}} \vert v \vert^k \nonumber\\
&\hspace{80mm} \times \left(\frac{1}{r}\right)^{2k} \beta^{k-s} e^{-\alpha (q_{\beta})^{sp} \vert v \vert^p}
\dd \omega_k \dd t_k \dots \dd \omega_1 \dd t_1 \nonumber\\
&\leq \left(\text{supess}_{(x,v) \in \mathbb{R}^{2d}} \left\vert f_0(x,v) e^{\alpha \vert v \vert^p} \right\vert\right) \frac{t^k}{k!} \left( \frac{\vert \mathbb{S}^{d-1} \vert \cdot \vert v \vert}{r^2} \right)^k \left( \frac{1}{r}\right)^{ks} q_{\beta}^{\frac{(k-s)^2}{2}} \beta^{k-s} e^{-\alpha (q_{\beta})^{sp} \vert v \vert^p}.
\end{align}
Now, each of the $k$-th terms of the series \eqref{EQUATRepreSerieSolution__BoltzLineaInela} is decomposed into $2^k$ sub-terms, labelled by the vectors $(e_1,\dots,e_k) \in \{0,1\}^k$. In addition, for $k$ fixed, and for a length $s$ fixed, there are ${k \choose s}$ such vectors. Therefore, the infinite series in \eqref{EQUATRepreSerieSolution__BoltzLineaInela} can be bounded from above by:
\begin{align}
\label{EQUATSerieBorneSuper}
&\sum_{k=1}^{+\infty} \sum_{s=0}^k {k \choose s} \left(\text{supess}_{(x,v) \in \mathbb{R}^{2d}} \left\vert f_0(x,v) e^{\alpha \vert v \vert^p} \right\vert\right) \frac{t^k}{k!} \left( \frac{\vert \mathbb{S}^{d-1} \vert \cdot \vert v \vert}{r^2} \right)^k \left( \frac{1}{r}\right)^{ks} q_{\beta}^{\frac{(k-s)^2}{2}} \beta^{k-s} e^{-\alpha (q_{\beta})^{sp} \vert v \vert^p} \nonumber\\
&= \left(\text{supess}_{(x,v) \in \mathbb{R}^{2d}} \left\vert f_0(x,v) e^{\alpha \vert v \vert^p} \right\vert\right) \sum_{k=1}^{+\infty} C^k \sum_{s=0}^k \frac{1}{s!(k-s)!}\left( \frac{1}{r}\right)^{ks} q_{\beta}^{\frac{(k-s)^2}{2}} \beta^{k-s} e^{-\alpha (q_{\beta})^{sp} \vert v \vert^p},
\end{align}
where
\begin{align}
C = \frac{\vert \mathbb{S}^{d-1} \vert \cdot \vert v \vert}{r^2}t.
\end{align}
We can now conclude the convergence of the series, relying on the study of the series:
\begin{align}
\mathcal{S} = \sum_{k=1}^{+\infty} C^k \sum_{s=0}^k \frac{1}{s!(k-s)!}\left( \frac{1}{r}\right)^{ks} q_{\beta}^{\frac{(k-s)^2}{2}} \beta^{k-s} e^{-\alpha (q_{\beta})^{sp} \vert v \vert^p}.
\end{align}
In order to do this, we need to choose in an appropriate manner the cut-off parameter $\beta$. Let us observe that the estimate \eqref{EQUATTermeSerieGenerOrdre__k_2} is obtained uniformly in terms of the vectors $(e_1,\dots,e_k)$, provided that $k$ is fixed, as well as the length $s$ of such vectors. Therefore, $\beta$ cannot be chosen depending on the vectors $(e_1,\dots,e_k)$, nevertheless, we can choose $\beta$ depending  on both $k$ and $s$. More precisely, we will consider:
\begin{align}
\label{EQUATDefinBeta__k_s_}
\beta = \beta_{k,s} = \left\{
\begin{array}{cc}
\displaystyle{\frac{1}{(k-s) e^{\ln(1/r) s(s+1)}}} \hspace{2mm} &\text{if } s\neq 0 \text{ and } s\neq k,\\
\frac{1}{k} \hspace{2mm} &\text{if } s = 0,\\\beta_0 > 0 \hspace{2mm} &\text{if } s = k.
\end{array}
\right.
\end{align}
To estimate  the series $\mathcal{S}$, we start with separating the sum over $s$ in three parts, treating separately the two extreme cases $s = 0$ and $s = k$. We set:
\begin{align}
\mathcal{S} &= \underbrace{\sum_{k=1}^{+\infty} C^k \frac{1}{k!} q_{\beta}^{\frac{k^2}{2}} \beta^k e^{-\alpha \vert v \vert^p}}_{=\,\mathcal{S}_1}
+ \underbrace{\sum_{k=1}^{+\infty} C^k \sum_{s=1}^{k-1} \frac{1}{s!(k-s)!}\left( \frac{1}{r}\right)^{ks} q_{\beta}^{\frac{(k-s)^2}{2}} \beta^{k-s} e^{-\alpha (q_{\beta})^{sp} \vert v \vert^p}}_{=\,\mathcal{S}_2} \nonumber\\
&\hspace{5mm}+ \underbrace{\sum_{k=1}^{+\infty} C^k \frac{1}{k!}\left( \frac{1}{r}\right)^{k^2} \hspace{-3mm} e^{-\alpha (q_{\beta})^{kp} \vert v \vert^p}}_{=\,\mathcal{S}_3}
\end{align}
An important point to keep in mind is that $q_\beta$ depends on $\beta$, and that $q_\beta$ converges to $1$ when $\beta$ converges to $0$, so that one cannot easily make use of the exponential weight in such a case.\\
We will rely on the two following formulas.
\begin{align}
\label{EQUATFormuGener__q__beta1}
q_\beta^{\frac{k^2}{2}}\beta^k = e^{\frac{k^2}{4} \ln\left( 1 + \left(\frac{1}{r^2}-1\right) \beta^2 \right)} e^{k\ln\beta} \leq e^{\frac{((1/r^2)-1)}{4}k^2\beta^2 + k\ln\beta}
\end{align}
and
\begin{align}
\label{EQUATFormuGener__q__beta2}
\left( \frac{1}{r} \right)^{k^2} e^{-\alpha (q_\beta)^{kp} \vert v \vert^p} = e^{k^2 \ln(1/r) - \alpha (q_\beta)^{kp} \vert v \vert^p}.
\end{align}
Concerning the first series $\mathcal{S}_1$, using \eqref{EQUATFormuGener__q__beta1}, the choice $\beta_{k,0} = 1/k$ provides:
\begin{align}
\label{EQUATBorneSomme_S_1_}
\mathcal{S}_1 &\leq \sum_{k=1}^{+\infty} \frac{C^k}{k!} e^{\frac{((1/r^2)-1)}{4}k^2\beta^2 + k\ln\beta} \leq \sum_{k=1}^{+\infty} \frac{C^k}{k!} e^{\frac{((1/r^2)-1)}{4} - k\ln k} \leq e^{C + \frac{((1/r^2)-1)}{4}}.
\end{align}
Concerning the third series $\mathcal{S}_3$, \eqref{EQUATFormuGener__q__beta2} provides, with the choice $\beta_{k,k} = \beta_0 > 0$ (so that $q_{\beta_0}^p > 1$ is independent from $k$), that:
\begin{align}
\label{EQUATBorneSomme_S_3_}
\mathcal{S}_3 \leq \left( \sup_{k \geq 1} e^{k^2 \ln(1/r) - \alpha (q_{\beta_0})^{kp} \vert v \vert^p} \right) \sum_{k=1}^{+\infty} \frac{C^k}{k!} \leq e^C \left( \sup_{k \geq 1} e^{k^2 \ln(1/r) - \alpha (q_{\beta_0})^{kp} \vert v \vert^p} \right).
\end{align}
We turn now to the central series $\mathcal{S}_2$. Exchanging the sums over $k$ and $s$ we find, and performing the change of variables $j = k-s$, we find:
\begin{align}
\mathcal{S}_2 &= \sum_{k=1}^{+\infty} C^k \sum_{s=1}^{k-1} \frac{1}{s!(k-s)!}\left( \frac{1}{r}\right)^{ks} q_{\beta}^{\frac{(k-s)^2}{2}} \beta^{k-s} e^{-\alpha (q_{\beta})^{sp} \vert v \vert^p} \nonumber\\
&= \sum_{s=1}^{+\infty} \sum_{j=1}^{+\infty} \frac{C^{j+s}}{s!j!} \left(\frac{1}{r}\right)^{(j+s)s} q_\beta^{\frac{j^2}{2}} \beta^j e^{-\alpha q_\beta^{sp} \vert v \vert^p}.
\end{align}
The choice \eqref{EQUATDefinBeta__k_s_} of $\beta_{k,s}$ provides:
\begin{align}
C^j \left(\frac{1}{r}\right)^{js} q_{\beta_{j+s,s}}^{\frac{j^2}{2}} \beta_{j+s,s}^j \leq \exp\left[ j(\ln C - \ln j) + \frac{(1/r^2-1)}{4} - js^2 \ln(1/r) \right].
\end{align}
Therefore, since $j(\ln C - \ln j)$ tends to $-\infty$ when $j \rightarrow +\infty$, $\exp\left(j(\ln C - \ln j) + \frac{(1/r^2-1)}{4} \right)$ is bounded from above, uniformly in $s$, so that:
\begin{align}
\label{EQUATBorneSomme_S_2_}
\mathcal{S}_2 &\leq \sum_{s=1}^{+\infty} \frac{C^s}{s!} \left(\frac{1}{r}\right)^{s^2}e^{-s^2 \ln(1/r)} \sum_{k=1}^{+\infty} \frac{\left[\sup_{l\geq 1} \exp\left( j(\ln C - \ln j) + \frac{(1/r^2-1)}{4}\right)\right]}{j!} \nonumber\\
&\leq \left[\sup_{l\geq 1} \exp\left( j(\ln C - \ln j) + \frac{(1/r^2-1)}{4}\right)\right] e^{C+1}.
\end{align}
Gathering \eqref{EQUATBorneSomme_S_1_}, \eqref{EQUATBorneSomme_S_3_} and \eqref{EQUATBorneSomme_S_2_}, we deduce that the series $\mathcal{S}$ (given by the expression \eqref{EQUATSerieBorneSuper}) is converging.\\
Therefore, the series \eqref{EQUATRepreSerieSolution__BoltzLineaInela} is converging, and the proof of Theorem \ref{PROPOConveRepreSerieSolutBoltzLineaInela} is complete.
\end{proof}

\subsection{Series representation of the solutions of the adjoint equation}

We turn now the series representation of the solutions of the adjoint of the inelastic linear Boltzmann equation.

\begin{propo}[Series representation for the solutions of the adjoint equation \eqref{EQUATLineaBoltzSphDuFormeAdjoi}]
\label{PROPORepreSerieEquatAdjoi}
Let $\varphi: \mathbb{R}^d \times \mathbb{R}^d \rightarrow \mathbb{R}$ be a $\mathcal{C}^1$ function, such that $\varphi$ and its gradient $\nabla_{x,v} \varphi$ are vanishing at infinity.\\
We define the series:
\begin{align}
\label{EQUATRepreSerieEquatAdjoi}
\psi(t,x,v) &= e^{-C_d \vert v \vert t}\varphi(x+tv,v) \nonumber\\
&\hspace{5mm}+ \sum_{k=1}^{+\infty} \int_{t_1=0}^t \int_{\mathbb{S}^{d-1}_{\omega_1}} \hspace{-3mm}\dots \int_{t_k=0}^{t_{k-1}}\int_{\mathbb{S}^{d-1}_{\omega_k}} e^{\big[ \sum_{j=1}^k C_d \vert v^{(j-1)} \vert (t_j-t_{j-1}) - C_d \vert v^{(k)} \vert t_k \big]} \prod_{l=1}^k \big\vert v^{(l-1)}\cdot\omega_l \big\vert \nonumber\\
&\hspace{38mm} \times \varphi(x + \sum_{m=1}^k \big(t_{m-1}-t_m\big) v^{(m-1)} + t_k v^{(k)},v^{(k)}) \dd \omega_k \dd t_k \dots \dd \omega_1 \dd t_1,
\end{align}
with $C_d$ defined in \eqref{EQUATDefin_C_d_}, $t^{(0)} = t$ by convention and:
\begin{align}
v^{(k)} = v^{(k-1)} - (1+r) \big( v^{(k-1)} \cdot \omega_k \big) \omega_k \hspace{2mm} \forall k \geq 1,\hspace{5mm} v^{(0)} = v.
\end{align}
Then the series $\psi(t,x,v)$ is converging for any $t \in \mathbb{R}$, $x,v \in \mathbb{R}^d$, $\psi$ is a $\mathcal{C}^1$ function, and solves (in the strong sense) the adjoint equation \eqref{EQUATLineaBoltzSphDuFormeAdjoi} of the linear inelastic Boltzmann equation \eqref{EQUATLineaBoltzSphDuFormeForte}, with initial datum $\varphi$. In other words, we have:
\begin{align}
\left\{
\begin{array}{rclc}
\partial_t \psi(t,x,v) - v \cdot \nabla_x \psi(t,x,v) &=& \displaystyle{\int_{\mathbb{S}^{d-1}_\omega}} \vert v \cdot \omega \vert \big[ \psi(t,x,v') - \psi(t,x,v) \big] \dd \omega &\forall\, t\in \mathbb{R},\, x,v\in \mathbb{R}^d,\\
\psi(0,x,v) &=& \varphi(x,v) &\forall\, x,v\in \mathbb{R}^d,
\end{array}
\right.
\end{align}
with $v' = v - (1+r) \big(v\cdot\omega)\omega$.\\
Conversely, any $\mathcal{C}^1$ function $\widetilde{\varphi}:\mathbb{R}\times\mathbb{R}^d\times\mathbb{R}^d \rightarrow \mathbb{R}$ that is a strong solution to the adjoint equation \eqref{EQUATLineaBoltzSphDuFormeAdjoi} of the linear inelastic Boltzmann equation \eqref{EQUATLineaBoltzSphDuFormeForte}, with initial datum $\varphi$, is equal to the series $\psi$. In other words, $\widetilde{\varphi} = \psi$, the solution of the adjoint equation \eqref{EQUATLineaBoltzSphDuFormeAdjoi} is unique in the class of the $\mathcal{C}^1$ functions, and is given by the expression \eqref{EQUATRepreSerieEquatAdjoi} of $\psi$.\\
In addition, assuming only that $\varphi \in \mathcal{C}_0(\mathbb{R}^d\times\mathbb{R}^d)$, for any $(t,x,v) \in [0,+\infty[ \times \mathbb{R}^d \times \mathbb{R}^d$, the expression \eqref{EQUATRepreSerieEquatAdjoi} defining the function $\psi$ is still well-defined, it satisfies:
\begin{align}
\label{EQUATBoundGrowtWeak_Solut}
\big\vert \psi(t,x,v) \big\vert \leq \vertii{ \varphi }_\infty e^{C_d \vert v \vert t}
\end{align}
and if the initial datum $\varphi = \psi(0,\cdot,\cdot)$ is compactly supported, then the support $\text{supp} \,(\psi)$ of $\psi$ is such that $\text{supp} \, (\psi) \cap \big( [0,t_0] \times \mathbb{R}^d \times \mathbb{R}^d \big)$ is a compact set of $\mathbb{R}_+ \times \mathbb{R}^d \times \mathbb{R}^d$ for any $t_0 > 0$.
\end{propo}

\noindent
Observe that the integration variables $t_k < t_{k-1} < \dots <t_1$ (with $0 \leq t_k$ and $t_1 \leq t$) have to be interpreted as the consecutive collision times of the tagged particle with the scatterers. We emphasize that $t_k$ corresponds to the time of the first collision, $t_{k-1}$ the second, and so on, so that the labeling of the collision times is inverted with respect to the labeling of the collisions.\\
The proof of Proposition \ref{PROPORepreSerieEquatAdjoi} relies on classical arguments. We do not present the proof here, but it is postponed to Appendix \ref{SECTIProofSerieRepreSolutAdjoi}, for the sake of completeness.
\begin{remar}
It is important to observe here that the regularity of the initial datum $\varphi$ does not propagate to the solution $\widetilde{\varphi} = \psi$ of the adjoint equation \eqref{EQUATLineaBoltzSphDuFormeAdjoi}, in the following sense. If we assume that $\varphi$ is vanishing at infinity, it is not clear a priori that such a property holds true for $\psi$. This phenomenon is specific to the inelastic case, and might be interpreted as the consequence of the fact that, contrary to the elastic case for which $
e^{\big[ \sum_{j=1}^k C_d \vert v^{(j-1)} \vert (t_j-t_{j-1}) - C_d \vert v^{(k)} \vert t_k \big]} = e^{-C_d \vert v \vert t}$, here the exponential term in \eqref{EQUATRepreSerieEquatAdjoi} is more singular, and does not compensate entirely the growth coming from the product $\prod_{l=1}^k \big\vert v^{(l-1)}\cdot\omega_l \big\vert$.
\end{remar}

\subsection{Existence and uniqueness of the weak solution to the inelastic linear Boltzmann equation}
\label{SSECTExistUniquSolut}

We start with establishing the existence of weak solutions to the inelastic linear Boltzmann equation. In the same spirit as in \cite{NoVe017}, we construct explicitly a weak solution, relying on solutions of the adjoint equation \eqref{EQUATLineaBoltzSphDuFormeAdjoi}.

\begin{propo}[Existence of weak solutions to \eqref{EQUATLineaBoltzSphDuFormeForte}]
\label{PROPOExistWeak_Solut}
Let $f_0 \in \mathcal{P}\big( \mathbb{R}^d \times \mathbb{R}^d \big)$ be a probability measure, which is also a non-negative Radon measure, and let $p > 1$ be constant such that:
\begin{align}
\label{EQUATExpo_Bound_Init_DataTHEORExist}
\int_{\mathbb{R}^d_x} \hspace{-1mm} \int_{\mathbb{R}^d_v} e^{\vert v \vert^p} f_0(\dd x,\dd v) < + \infty.
\end{align}
Then, there exists a weak solution to \eqref{EQUATLineaBoltzSphDuFormeForte} in the sense of Definition \ref{DEFINSolutFaibleEquatBoltzLineaInela}, with initial datum $f_0$, that we denote by $g$, and that is defined by duality for all $t \geq 0$ as:
\begin{align}
\int_{\mathbb{R}^d}\hspace{-1mm}\int_{\mathbb{R}^d} \varphi(x,v) g(t,\dd x,\dd v) = \int_{\mathbb{R}^d}\hspace{-1mm}\int_{\mathbb{R}^d} \psi(t,x,v) f_0(\dd x,\dd v) \hspace{5mm} \forall \varphi \in \mathcal{C}_0(\mathbb{R}^d\times\mathbb{R}^d),
\end{align}
where $\psi$ is the series given by the expression \eqref{EQUATRepreSerieEquatAdjoi}, with initial datum $\varphi$.
\end{propo}

\begin{proof}
First, according to the bound \eqref{EQUATBoundGrowtWeak_Solut} for the solution $\psi$ of the adjoint equation, we have that $(x,v) \mapsto \psi(t,x,v) e^{- \vert v \vert^p}$ is a bounded function. This ensures, together with the assumption \eqref{EQUATExpo_Bound_Init_DataTHEORExist}, that $g(t)$ is a finite Radon measure.\\
To prove that $g$ is a weak solution to \eqref{EQUATLineaBoltzSphDuFormeForte}, we consider a general test function $\widetilde{\eta} \in \mathcal{C}^\infty_c([0,+\infty[\times\mathbb{R}^d\times\mathbb{R}^d)$, and we introduce the quantity:
\begin{align}
\Delta(\widetilde{\eta}) &= \int_{\mathbb{R}^d_x} \hspace{-1mm} \int_{\mathbb{R}^d_v} \hspace{-1.5mm} \widetilde{\eta}(0,x,v) f_0(\dd x,\dd v) + \int_0^{+\infty} \hspace{-2.5mm} \int_{\mathbb{R}^d_x} \hspace{-1mm} \int_{\mathbb{R}^d_v} \hspace{-1.5mm} \partial_t \widetilde{\eta}(t,x,v) g(t,\dd x,\dd v) \dd t \nonumber\\
&\hspace{5mm}+ \int_0^{+\infty} \hspace{-2.5mm} \int_{\mathbb{R}^d_x} \hspace{-1mm} \int_{\mathbb{R}^d_v} \hspace{-1.5mm} v\cdot\nabla_x \widetilde{\eta}(t,x,v) g(t,\dd x,\dd v) \dd t + \int_0^{+\infty} \hspace{-2.5mm} \int_{\mathbb{R}^d_x} \hspace{-1mm} \int_{\mathbb{R}^d_v} \mathcal{L}^*\big[\widetilde{\eta}\big](t,x,v) g(t,\dd x,\dd v) \dd t
\end{align}
with $\displaystyle{\mathcal{L}^*\big[\widetilde{\eta}\big](t,x,v) = \int_{\mathbb{S}^{d-1}_\omega} \vert v \cdot \omega \vert \big[ \widetilde{\eta}(t,x,v') - \widetilde{\eta}(t,x,v) \big] \dd \omega}$. We denote by $\xi$ the quantity:
\begin{align}
\label{EQUAT_PDE_Test_Funct}
\partial_t \widetilde{\eta}(t,x,v) + v\cdot\nabla_x \widetilde{\eta}(t,x,v) + \mathcal{L}^*\big[\widetilde{\eta}\big](t,x,v) = \xi(t,x,v).
\end{align}
We prove now an integral representation formula for $\widetilde{\eta}$ in terms of solutions of the equation $\partial_t \phi + v\cdot\nabla_x \phi + \mathcal{L}^*\big[\phi\big] = 0$. By assumption, $\widetilde{\eta}$ is compactly supported, so let $t_0 > 0$ be a real number such that $\text{supp}\, (\widetilde{\eta}) \subset [0,t_0] \times \mathbb{R}^d \times \mathbb{R}^d$. We consider the function $\widetilde{\theta}$ defined as $\widetilde{\theta}(t,x,v) = -\int_t^{t_0} \phi(s,t) \dd s$, where, for any $0 \leq s \leq t_0$, $\phi(s,\cdot)$ solves the Cauchy problem:
\begin{align}
\left\{
\begin{array}{ccc}
\partial_t \phi(s,t) + v\cdot\nabla_x \phi(s,t) + \mathcal{L}^*\big[\phi(s,t)\big] &=& 0,\\
\phi(s,s) &=& \xi(s).
\end{array}
\right.
\end{align}
By definition, we have $\widetilde{\theta}(t_0,\cdot,\cdot) = 0$ for any $(x,v) \in \mathbb{R}^d \times \mathbb{R}^d$. In addition, by construction we have also:
\begin{align}
\partial_t \widetilde{\theta} = \phi(t,t) - \int_t^{t_0} \partial_t \phi(s,t) \dd s &= \xi(t) + \int_t^{t_0} v \cdot \nabla_x \phi(s,t) \dd s + \int_t^{t_0} \mathcal{L}^* \big[ \phi(s,t) \big] \dd s \nonumber\\
&= \xi(t) - v \cdot \nabla_x \widetilde{\theta}(t) - \mathcal{L}^*\big[ \widetilde{\theta}(t) \big],
\end{align}
so that $\widetilde{\theta}$ solves also the equation \eqref{EQUAT_PDE_Test_Funct}. We deduce then that $\widetilde{\eta} = \widetilde{\theta}$. We can now rewrite $\Delta(\widetilde{\eta})$ as follows:
\begin{align}
\Delta(\widetilde{\eta}) &= \int_{\mathbb{R}^d_x} \hspace{-1mm} \int_{\mathbb{R}^d_v} \hspace{-1.5mm} \widetilde{\eta}(0,x,v) f_0(\dd x,\dd v) + \int_0^{+\infty} \hspace{-2.5mm} \int_{\mathbb{R}^d_x} \hspace{-1mm} \int_{\mathbb{R}^d_v} \hspace{-1.5mm} \xi(t,x,v) g(t,\dd x,\dd v) \dd t \nonumber\\
&= \int_{\mathbb{R}^d_x} \hspace{-1mm} \int_{\mathbb{R}^d_v} \hspace{-1.5mm} \widetilde{\eta}(0,x,v) f_0(\dd x,\dd v) + \int_0^{t_0} \hspace{-2.5mm} \int_{\mathbb{R}^d_x} \hspace{-1mm} \int_{\mathbb{R}^d_v} \hspace{-1.5mm} \phi(t,t,x,v) g(t,\dd x,\dd v) \dd t \nonumber\\
&= \int_{\mathbb{R}^d_x} \hspace{-1mm} \int_{\mathbb{R}^d_v} \hspace{-1.5mm} \widetilde{\eta}(0,x,v) f_0(\dd x,\dd v) + \int_0^{t_0} \hspace{-2.5mm} \int_{\mathbb{R}^d_x} \hspace{-1mm} \int_{\mathbb{R}^d_v} \hspace{-1.5mm} \phi(t,0,x,v) f_0(\dd x,\dd v) \dd t,
\end{align}
where in the last line we used that $\tau \mapsto \phi(s,t-\tau)$ solves the adjoint equation \eqref{EQUATLineaBoltzSphDuFormeAdjoi} for any $s$, and we applied the definition of the function $g$. We find in the end:
\begin{align}
\Delta(\widetilde{\eta}) &= \int_{\mathbb{R}^d_x} \hspace{-1mm} \int_{\mathbb{R}^d_v} \hspace{-1.5mm} \widetilde{\eta}(0,x,v) f_0(\dd x,\dd v) - \int_{\mathbb{R}^d_x} \hspace{-1mm} \int_{\mathbb{R}^d_v} \hspace{-1.5mm} \Big( - \int_0^{t_0} \phi(t,0) \dd t \Big) f_0(\dd x,\dd v) = 0,
\end{align}
so that $g$ is indeed a weak solution in the sense of Definition \ref{DEFINSolutFaibleEquatBoltzLineaInela}, which concludes the proof.
\end{proof}
\noindent
We complete the results concerning the weak solutions of \eqref{EQUATLineaBoltzSphDuFormeForte} by proving the uniqueness of such weak solutions.

\begin{propo}[Uniqueness of weak solutions to \eqref{EQUATLineaBoltzSphDuFormeForte}]
\label{PROPOUniquWeak_Solut}
Let $f_0 \in \mathcal{P}\big( \mathbb{R}^d \times \mathbb{R}^d \big)$ be a probability measure, which is also a non-negative Radon measure. 
Then, there exists at most one weak solution to \eqref{EQUATLineaBoltzSphDuFormeForte} in the sense of Definition \ref{DEFINSolutFaibleEquatBoltzLineaInela}, with initial datum $f_0$.
\end{propo}

\begin{proof}
By linearity, it is enough to prove that only the zero function is a weak solution with zero initial datum. We consider then a weak solution $f$ to \eqref{EQUATLineaBoltzSphDuFormeForte}, with zero initial datum.\\
Let $t_0 > 0$ be any positive number. We consider a sequence $\big(\chi_n\big)_n$ of smooth functions $\chi_n:\mathbb{R}_+ \rightarrow [0,1]$ that converges pointwise towards $\mathds{1}_{[0,t_0]}$, supported on $[0,t_0]$, and such that $\chi_n(t) = 1$ for any $t \leq t_0 - 1/n$. By assumption, for any compactly supported test function $\widetilde{\eta}$, we have:
\begin{align}
&\int_0^{t_0} \hspace{-2.5mm} \int_{\mathbb{R}^d_x} \hspace{-1mm} \int_{\mathbb{R}^d_v} \hspace{-1.5mm} \partial_t \big( \chi_n(t) \widetilde{\eta}(t,x,v) \big) f(t,\dd x,\dd v) \dd t \nonumber\\
&\hspace{5mm}= - \int_0^{t_0} \hspace{-2.5mm} \chi_n(t) \hspace{-1mm} \int_{\mathbb{R}^d_x} \hspace{-1mm} \int_{\mathbb{R}^d_v} \hspace{-1.5mm} v\cdot\nabla_x \widetilde{\eta}(t,x,v) f(t,\dd x,\dd v) \dd t - \int_0^{t_0} \hspace{-2.5mm} \chi_n(t) \hspace{-1mm} \int_{\mathbb{R}^d_x} \hspace{-1mm} \int_{\mathbb{R}^d_v} \mathcal{L}^*\big[\widetilde{\eta}\big](t,x,v) f(t,\dd x,\dd v) \dd t.
\end{align}
We have:
\begin{align}
\int_0^{t_0} \hspace{-2.5mm} \int_{\mathbb{R}^d_x} \hspace{-1mm} \int_{\mathbb{R}^d_v} \hspace{-1.5mm} \partial_t \chi_n(t) \widetilde{\eta}(t,x,v) f(t,\dd x,\dd v) \dd t &= \Big( \int_{t_0-1/n}^{t_0} \hspace{-6mm} \partial_t \chi_n(t) \dd t \Big) \Big( \int_{\mathbb{R}^d_x} \hspace{-1mm} \int_{\mathbb{R}^d_v} \hspace{-1mm}  \widetilde{\eta}(t_0,x,v) f(t_0,\dd x,\dd v)  \Big) + o(1) \nonumber\\
&= - \int_{\mathbb{R}^d_x} \hspace{-1mm} \int_{\mathbb{R}^d_v} \hspace{-1mm}  \widetilde{\eta}(t_0,x,v) f(t_0,\dd x,\dd v) + o(1)
\end{align}
as $n \rightarrow +\infty$, using in particular that $\widetilde{\eta}$ and $f$ are both continuous. Therefore, by dominated convergence theorem, we have:
\begin{align}
\label{EQUATRepreFormuWeak_Solut}
&- \int_{\mathbb{R}^d_x} \hspace{-1mm} \int_{\mathbb{R}^d_v} \hspace{-1mm}  \widetilde{\eta}(t_0,x,v) f(t_0,\dd x,\dd v) + \int_0^{t_0} \hspace{-2.5mm} \int_{\mathbb{R}^d_x} \hspace{-1mm} \int_{\mathbb{R}^d_v} \hspace{-1.5mm} \partial_t  \widetilde{\eta}(t,x,v) f(t,\dd x,\dd v) \dd t \nonumber\\
&\hspace{5mm}= - \int_0^{t_0} \hspace{-2.5mm} \int_{\mathbb{R}^d_x} \hspace{-1mm} \int_{\mathbb{R}^d_v} \hspace{-1.5mm} v\cdot\nabla_x \widetilde{\eta}(t,x,v) f(t,\dd x,\dd v) \dd t - \int_0^{t_0} \hspace{-2.5mm} \int_{\mathbb{R}^d_x} \hspace{-1mm} \int_{\mathbb{R}^d_v} \mathcal{L}^*\big[\widetilde{\eta}\big](t,x,v) f(t,\dd x,\dd v) \dd t.
\end{align}
For any Borel set $A \subset \mathbb{R}^d\times\mathbb{R}^d$, we consider now a sequence $\big(\widetilde{\eta}_m\big)_m$ of smooth functions  smooth such that $\widetilde{\eta}_m(t_0,\cdot,\cdot)$ approximates $\mathds{1}_A$ and such that $\widetilde{\eta}_m(-t)$ solves the adjoint equation \eqref{EQUATLineaBoltzSphDuFormeAdjoi}. According to Proposition \ref{PROPORepreSerieEquatAdjoi}, $\widetilde{\eta}_m\chi_n$ is indeed a smooth, compactly supported function, so that \eqref{EQUATRepreFormuWeak_Solut} holds for $\widetilde{\eta}_m$. In the limit $m \rightarrow +\infty$, we find:
\begin{align}
f(t_0,A) = \int_{\mathbb{R}^d_x} \hspace{-1mm} \int_{\mathbb{R}^d_v} \mathds{1}_A(x,v) f(t_0,\dd x, \dd v) = 0,
\end{align}
which shows that $f(t_0)$ is the zero measure, concluding the proof of Proposition \ref{PROPOUniquWeak_Solut}.
\end{proof}

\subsection{Some considerations on the kinetic energy}

To complete this section on the inelastic linear Boltzmann equation, we conclude with a discussion on the cooling of the tagged particle as it evolves and collides with the scatterers. More precisely, we provide here an upper bound on the decay of the kinetic energy, following classical arguments in inelastic kinetic theory.

\begin{propo}
\label{PROPOBorneSuperEnergCinet}
Let $f$ be a non-negative and regular solution of \eqref{EQUATLineaBoltzSphDuFormeForte} with initial datum $f_0$, such that its mass is equal to $1$, and such that its second and third moments are always finite. Then we have:
\begin{align}
\int_x\int_v \vert v \vert^2 f(t,x,v) \dd v \dd x \leq \left( \left(\int_x\int_v \vert v \vert^2 f_0(t,x,v) \dd v \dd x\right)^{-1/2} + \frac{(1-r^2)}{2} \left[ \int_\omega \vert e_1\cdot \omega \vert^3 \dd \omega \right] t \right)^{-2}.
\end{align}
\end{propo}

\begin{proof}
The result is obtained as a consequence of a Gr\"onwall-type estimate. We have:
\begin{align}
\label{EQUATDerivEnergCinet}
\frac{\dd}{\dd t} \int_x\int_v \vert v \vert^2 f(t,x,v) \dd v \dd x &= \left( \int_\omega \vert e_1\cdot\omega \vert^3 \dd \omega \right) (r^2-1) \int_x\int_v \vert v \vert^3 f(t,x,v) \dd v \dd x.
\end{align}
Applying now the H\"older inequality, and keeping in mind that the mass of $f$ is constantly equal to $1$:
\begin{align}
\int_x\int_v \vert v \vert^2 f(v) \dd v\dd x \leq \left( \int_x\int_v \vert v \vert^{2p} f(v) \dd v\right)^{1/p} \left( \int_x\int_v f(v) \dd v \dd x \right)^{1/q} = \left( \int_x\int_v \vert v \vert^{2p} f(v) \dd v\right)^{1/p},
\end{align}
for any $p,q \in [1,+\infty]$ that are H\"older conjugates. In particular, taking $p = 3/2$, 
the result follows by direct integration of the inequality.
\end{proof}

\begin{remar}
Let us observe that \eqref{EQUATDerivEnergCinet} enables to recover the conservation of the kinetic energy in the elastic case, that is, when $r = 1$. In the general case, the upper bound on the decay of the kinetic energy agrees with the celebrated Haff's law for granular gases of hard spheres with fixed restitution coefficient \cite{BrPo004}.
\end{remar}

\section{Rigorous derivation of the linear inelastic Boltzmann equation, via weak convergence}

\subsection{Technical preliminaries}

To perform the derivation of the weak form, we will rely on the following technical results. The purpose of the following lemma is to deduce a condition on the angular parameter $\omega$, providing that the pre-collisional velocity $v$ is given, and that the direction of the post-collisional velocity $v' = v-(1+r)\big(v\cdot\omega)\omega$ is given, up to a small error.

\begin{lemma}[Almost colinearity after scattering]
\label{LEMMEColinearitScatt}
Let $r \in \ ]0,1[$. Then, at $v,p \in \mathbb{R}^d$ fixed, the post-collisional velocity $v' = \kappa_\omega(v)$ defined in \eqref{EQUATScattering} is almost colinear to $p$, that is:
\begin{align}
\label{EQUATLemmeColinHypot}
\left[ 1 - \left\vert \frac{v'}{\vert v' \vert} \cdot p \right\vert \right] \leq \delta
\end{align}
only if the angular parameter $\omega$ belongs to a measurable subset $\mathcal{P}_{\text{colin.}}(v,p,\delta) \subset \mathbb{S}^{d-1}$ of the unit sphere that has a Lebesgue measure smaller than:
\begin{align}
\label{EQUATLemmeColinConclMesur}
\big\vert \mathcal{P}_{\text{colin.}}(v,p,\delta) \big\vert \leq C(d,r) \delta^{1/2},
\end{align}
where $C(d,r)$ is a constant that depends only on the dimension $d$ and the restitution coefficient $r$. In particular, the estimate \eqref{EQUATLemmeColinConclMesur} holds uniformly in $v,p \in \mathbb{R}^d$.\\
More precisely, there exist two positive constants $\delta_0 \in \mathbb{R}_+^*$ and $C(d,r) \in \mathbb{R}_+^*$ such that, for any vectors $v,\omega, p \in \mathbb{R}^d$ such that $v \neq 0$ and $\omega, p \in \mathbb{S}^{d-1}$, and any positive number $0 < \delta \leq \delta_0$, such that if \eqref{EQUATLemmeColinHypot}
holds true, then, $\omega$ belongs to a subset subset $\mathcal{P}_{\text{colin.}}(v,p,\delta) \subset \mathbb{S}^{d-1}$ which satisfies the condition \eqref{EQUATLemmeColinConclMesur}.
\end{lemma}

\noindent
In the core of the proof of Theorem \ref{THEORDerivationBoltzmann_InelaLinea}, we will also make use of the following result. More precisely, the following result will be used to estimate the measure of the dynamical tube, that we will introduce in due time.

\begin{lemma}[Estimate on the measure of twisted tubes]
\label{LEMMEEstimMesurTube_Dynam}
Let $\varepsilon > 0$ be a positive number, and let $x_1$, $x_2$ and $x_3$ be three vectors of $\mathbb{R}^d$. We consider the following set:
\begin{align}
\mathcal{T}_\varepsilon &= \big( [x_1,x_2 ] \cup [x_2,x_3] \big) + \overline{B(0,\varepsilon)} \nonumber\\
&= \big\{ y \in \mathbb{R}^d\ /\ \exists\, \lambda \in [0,1],\, z \in \overline{B(0,\varepsilon)} \ \text{such that} \ y = \lambda x_1 + (1-\lambda)x_2 + z \nonumber\\
&\hspace{80mm} \text{or}\hspace{2mm} y = \lambda x_2 + (1-\lambda)x_3 + z \big\}.
\end{align}
Then, the Lebesgue measure of $\mathcal{T}_\varepsilon$ is maximal when $x_1$, $x_2$ and $x_3$ are aligned, and in this order, that is, when:
\begin{align}
\frac{x_2-x_1}{\vert x_2 - x_1 \vert} \cdot \frac{x_3 - x_2}{\vert x_3 - x_2 \vert} = 1.
\end{align}
\end{lemma}

\noindent
The proofs of Lemmas \ref{LEMMEColinearitScatt} and \ref{LEMMEEstimMesurTube_Dynam} are postponed to the appendix.\\
Now that the technical tools are in place, we can now turn to the main part of the proof of Theorem \ref{THEORDerivationBoltzmann_InelaLinea}.

\begin{figure}[h!]
\centering
    \includegraphics[trim = 0cm 1.5cm 0cm 0cm, width=0.6\linewidth]{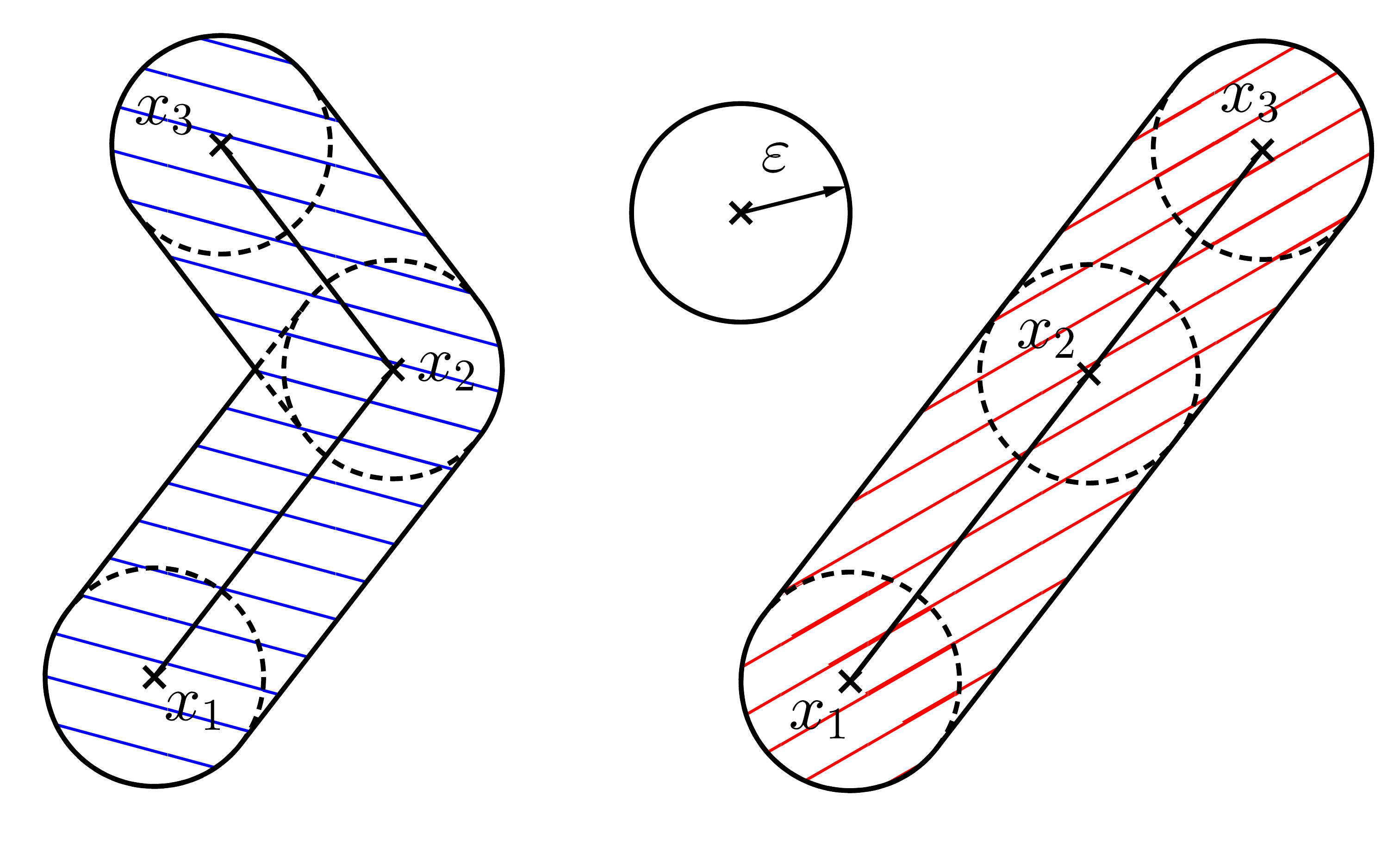}
\caption{Representation of the result of Lemma \ref{LEMMEEstimMesurTube_Dynam}. The surface of the twisted tube (in blue) is smaller than the surface of the straight tube (in red), even if the distances $\vert x_1 - x_2 \vert$ and $\vert x_2 - x_3 \vert$ are the same for the left and the right tubes. The phenomenon is independent from the dimension.}
\label{FIGURTubes}
\end{figure}

\subsection{Rewriting $\varphi_\varepsilon$ as a series}

The first step to rewrite $\varphi_\varepsilon = \mathbb{E}_{\mu_\varepsilon} \big[ \varphi\big( T^t_{c,\varepsilon}(x,v) \big) \big]$ is to ensure that the flow $T^t_{c,\varepsilon}(x,v)$ is well-defined. In order to define globally the flow $T_{c,\varepsilon}^t$ for a given distribution of scatterers $c \in C$, we will assume that $c \in X_0 \, \cap\, \mathcal{P}^c_{\text{patho.}}$, where $X_0$ is the set of the distributions $c = (c_i)_{i \in I}$ such that none of the $c_i$ is centered in the ball $\overline{B(x,\varepsilon)}$ (so that the initial position $x$ of the tagged particle is not inside any of the scatterers), and $\mathcal{P}_{\text{patho.}}$ is defined along Proposition \ref{PROPODefinGlobaDynam} in \eqref{EQUATDefinEnsmbPatho}, and contains all the distributions $c$ for which the flow $T_{c,\varepsilon}^t$ may not be globally defined because of the inelastic collapse.\\
If $c \in X_0^c \cup \mathcal{P}_{\text{patho.}}$, we extend the definition of the flow $T_{c,\varepsilon}$ by introducing $\widetilde{T}_{c,\varepsilon}$ as described in Definition \ref{DEFINGener_Flow}. We can then consider the object:
\begin{align}
\label{EQUATDefinPhiEp_Avec_FlotDefin}
\varphi_\varepsilon = \mathbb{E}_{\mu_\varepsilon} \big[ \varphi\big( \widetilde{T}^t_{c,\varepsilon}(x,v) \big) \big],
\end{align}
which is properly defined and which is a function of $t$, $x$ and $v$. For shortness, we omit the dependency on these variables. Then, we rewrite the quantity $\varphi_\varepsilon = \mathbb{E}_{\mu_\varepsilon} \big[ \varphi\big( \widetilde{T}^t_{c,\varepsilon}(x,v) \big) \big]$ in a way that allows to identify the limiting object towards which $\varphi_\varepsilon$ is converging in the Boltzmann-Grad limit.

\begin{propo}
\label{PROPOReecrPhiEpIntegObsta}
Let $\varepsilon > 0$. Let $\varphi$ be a $\mathcal{C}_0(\mathbb{R}^d\times\mathbb{R}^d)$ function. Then, the quantity $\varphi_\varepsilon$, defined as \eqref{EQUATDefinPhiEp_Avec_FlotDefin}, satisfies:
\begin{align}
\varphi_\varepsilon = \varphi_\varepsilon^{(1)} + R_1,
\end{align}
with
\vspace{-7mm}
\begin{align}
\label{EQUATDefinPhi_1Propo}
\varphi_\varepsilon^{(1)} &= e^{- \mu_\varepsilon \vert [x,x+tv] + \overline{B(0,\varepsilon)} \vert} \varphi(x+tv,v) \nonumber\\
&\hspace{5mm}+ \sum_{k=1}^{+\infty} \frac{\mu_\varepsilon^k}{k!} \int_{c_1} \dots \int_{c_k} e^{- \mu_\varepsilon \vert \overline{B(x,\varepsilon)} \, \cup \, \mathcal{T}^t(c^k) \vert} \varphi\big( T_{c,\varepsilon}^t(x,v) \big) \\
&\hspace{40mm} \times \mathds{1}_{X_0}(c^k) \mathds{1}_{\mathcal{P}_\text{patho.}^c}(c^k) \mathds{1}_{\{ c_i \in B_\varepsilon \forall 1 \leq i \leq k \}} \mathds{1}_{\substack{\text{the obstacles}\\ c_j,\, 1\leq j \leq k,\\ \text{are internal}}} \dd c_k \dots \dd c_1, \nonumber
\end{align}
where $\mathcal{T}^t(c^k)$ is the dynamical tube, defined in \eqref{EQUATDynam_Tube} below, and $R_1$ being a remainder term such that:
\begin{align}
\label{EQUATDefinReste_R_1_}
\big\vert R_1 \big\vert \leq \vertii{ \varphi }_\infty \mu_\varepsilon \vert \overline{B(0,\varepsilon)} \vert.
\end{align}
\end{propo}

\noindent
Observe that, in principle, we have $R_1 = R_1(t,x,v)$, even though the estimate \eqref{EQUATDefinReste_R_1_} provides an upper bound that does not depend on the variables $t$, $x$ and $v$. In the sequel, we will consider additional remainder terms, that also depend on these variables, but for which we will obtain estimates that depend also on $t$, $x$ and $v$ (specifically, $R_2$ and $R_3$, see Propositions \ref{PROPOReecrPhiEpElimiRecol} and \ref{PROPOCompaFinal_Phi_PhiEp}).

\begin{proof}
We have:
\begin{align}
\varphi_\varepsilon = \mathbb{E}_{\mu_\varepsilon} \big[ \varphi\big( \widetilde{T}^t_{c,\varepsilon}(x,v) \big) \big] &= \sum_{k=0}^{+\infty} \mathbb{E}_{\mu_\varepsilon} \big[ \varphi\big( \widetilde{T}^t_{c,\varepsilon}(x,v) \big) \mathds{1}_{\# (c\, \cap \overline{B(x,\varepsilon)}) = k}\big] \nonumber\\
&= \mathbb{E}_{\mu_\varepsilon} \big[ \varphi\big( T^t_{c,\varepsilon}(x,v) \big) \mathds{1}_{X_0}(c) \mathds{1}_{\mathcal{P}_\text{patho.}^c}(c) \big] + \mathbb{E}_{\mu_\varepsilon} \big[ \varphi\big( \widetilde{T}^t_{c,\varepsilon}(x,v) \big) \mathds{1}_{X_0}(c) \mathds{1}_{\mathcal{P}_\text{patho.}}(c)\big] \nonumber\\
&\hspace{5mm}+ \sum_{k=1}^{+\infty} \mathbb{E}_{\mu_\varepsilon} \big[ \varphi\big( \widetilde{T}^t_{c,\varepsilon}(x,v) \big) \mathds{1}_{\# (c\, \cap \overline{B(x,\varepsilon)}) = k}\big].
\end{align}
Denoting by $R_1$ the sum of the two quantities:
\begin{align}
R_1 = \mathbb{E}_{\mu_\varepsilon} \big[ \varphi\big( \widetilde{T}^t_{c,\varepsilon}(x,v) \big) \mathds{1}_{X_0}(c) \mathds{1}_{\mathcal{P}_\text{patho.}}(c) \big]  + \sum_{k=1}^{+\infty} \mathbb{E}_{\mu_\varepsilon} \big[ \varphi\big( \widetilde{T}^t_{c,\varepsilon}(x,v) \big) \mathds{1}_{\# (c\, \cap \overline{B(x,\varepsilon)}) = k}\big],
\end{align}
we have $\big\vert \varphi\big( \widetilde{T}^t_{c,\varepsilon}(x,v) \big) \mathds{1}_{X_0}(c) \mathds{1}_{\mathcal{P}_\text{patho.}}(c) \big\vert \leq \vertii{\varphi}_{\infty} \mathds{1}_{\mathcal{P}_\text{patho.}}(c)$ so that
\begin{align}
\Big\vert \mathbb{E}_{\mu_\varepsilon} \big[ \varphi\big( \widetilde{T}^t_{c,\varepsilon}(x,v) \big) \mathds{1}_{X_0}(c) \mathds{1}_{\mathcal{P}_\text{patho.}}(c) \big]  \Big\vert &\leq \vertii{\varphi}_{\infty} \mathbb{E}_{\mu_\varepsilon} \big[ \mathds{1}_{\mathcal{P}_\text{patho.}} \big] = \vertii{\varphi}_{\infty} \mathbb{P}_{\mu_\varepsilon}\big( \mathcal{P}_\text{patho.} \big) = 0.
\end{align}
Besides, using the properties of the Poisson process, we have:
\begin{align}
\Big\vert \sum_{k=1}^{+\infty} \mathbb{E}_{\mu_\varepsilon} \big[ \varphi\big( \widetilde{T}^t_{c,\varepsilon}(x,v) \big) \mathds{1}_{\# (c\, \cap \overline{B(x,\varepsilon)}) = k}\big] \Big\vert &\leq \vertii{\varphi}_{\infty} \sum_{k=1}^{+\infty} e^{-\mu_\varepsilon \vert \overline{B(0,\varepsilon)} \vert} \frac{\mu_\varepsilon^k}{k!} \big\vert \overline{B(0,\varepsilon)} \big\vert^k \leq \vertii{\varphi}_{\infty} \mu_\varepsilon \vert \overline{B(0,\varepsilon)} \vert,
\end{align}
so that in the end we obtained:
\begin{align}
\big\vert R_1 \big\vert \leq  \vertii{\varphi}_{\infty} \mu_\varepsilon \vert \overline{B(0,\varepsilon)} \vert.
\end{align}
\noindent
The second step to rewrite $\varphi_\varepsilon$ consists in identifying a dynamical ball, and in decomposing the cases depending on the number of scatterers contained in the dynamical ball.\\
More precisely, considering the forward in time dynamics of the tagged particle, given by the flow $T^t_{c,\varepsilon}(x,v)$, since the norm of the velocity of the particle is decreasing at any collision, it is clear that the tagged particle cannot exit the ball $B(x,t \vert v \vert)$ during the time interval $[0,t]$. Therefore, only the obstacles in the ball $B_\varepsilon$, that we define as:
\begin{align}
B_\varepsilon = \overline{B(x,t\vert v \vert + \varepsilon)}
\end{align}
can contribute to the dynamics of the tagged particle, that is, $T^t_{c,\varepsilon}(x,v)$ depends only on the obstacles that are in $B_\varepsilon$. We will call the closed ball $B_\varepsilon$ the \emph{dynamical ball}. We separate then the cases according to the number $n$ of obstacles in $B_\varepsilon$, that is $
n = \# \big( c \cap B_\varepsilon \big)$, and introducing the quantity $\varphi_\varepsilon^{(1)}$ as follows, we obtain:
\begin{align}
\varphi_\varepsilon^{(1)} = \mathbb{E}_{\mu_\varepsilon} \big[ \varphi\big( T^t_{c,\varepsilon}(x,v) \big) \mathds{1}_{X_0}(c) \mathds{1}_{\mathcal{P}_\text{patho.}^c}(c) \big] &= \sum_{n = 0}^{+\infty} \mathbb{E}_{\mu_\varepsilon} \big[ \varphi\big( T^t_{c,\varepsilon}(x,v) \big) \mathds{1}_{X_0}(c) \mathds{1}_{\mathcal{P}_\text{patho.}^c}(c) \mathds{1}_{\# ( c\, \cap B_\varepsilon ) = n}\big].
\end{align}
In the case when $n = 0$, there is no obstacle in the dynamical ball $B_\varepsilon$, and therefore:
$T_{c,\varepsilon}^t(x,v) = \big( x+tv,v \big)$. The expected value of $\varphi\big( T^t_{c,\varepsilon}(x,v) \big)$ is independent from $c$, and by definition of a Poisson process of intensity $\mu_\varepsilon$ we have:
\begin{align}
\mathbb{E}_{\mu_\varepsilon} \big[ \varphi\big( T^t_{c,\varepsilon}(x,v) \big) \mathds{1}_{X_0}(c) \mathds{1}_{\mathcal{P}_\text{patho.}^c}(c) \mathds{1}_{\# ( c\, \cap B_\varepsilon ) = 0} \big] = e^{-\mu_\varepsilon \vert B_\varepsilon \vert} \varphi(x+tv,v).
\end{align}
In the general case, we find:
\begin{align*}
\varphi_\varepsilon^{(1)} &= e^{-\mu_\varepsilon \vert B_\varepsilon \vert} \varphi(x+tv,v) + \sum_{n = 1}^{+\infty} e^{-\mu_\varepsilon \vert B_\varepsilon \vert} \frac{\mu_\varepsilon^n}{n!} \int_{c_1} \hspace{-2mm} \dots \hspace{-0.5mm} \int_{c_n} \varphi\big( T_{c,\varepsilon}^t(x,v) \big) \mathds{1}_{X_0}(c) \mathds{1}_{\mathcal{P}_\text{patho.}^c}(c) \mathds{1}_{c_i \in B_\varepsilon \forall 1 \leq i \leq n} \dd c_n \dots \dd c_1.
\end{align*}
The next step in the rewriting consists in separating the obstacles, between the ``internal'' and ``external'' ones. We say that \emph{an obstacle $c_i$ ($1 \leq i \leq n$) is internal} if and only if:
\begin{align}
\label{EQUATPreuvReecrDefinObstaInter}
\inf_{s\in[0,t]} \big\vert x(s) - c_i \big\vert = \varepsilon,
\end{align}
where $x(s)$ correspond to the position variables of the inelastic hard sphere flow $T^s_{c,\varepsilon}(x,v)$ of the tagged particle. In other words, $c_i$ is internal if this obstacle is touched by the tagged particle during the time interval $[0,t]$.\\
An obstacle is said to be \emph{external} if it is not internal.\\
Observe that the infimum in \eqref{EQUATPreuvReecrDefinObstaInter} is a minimum: since $\{ T^s_{c,\varepsilon}(x,v)\ / \ s \in [0,t] \}$ is a compact set, the infimum of the distance between this set and the compact set $\{c_i\}$ (formed by the single point $c_i$) is reached for some $s_{c_i} \in [0,t]$. We then perform the following decomposition, separating the particular case $k = 0$:
\begin{align}
&\sum_{n = 1}^{+\infty} e^{-\mu_\varepsilon \vert B_\varepsilon \vert} \frac{\mu_\varepsilon^n}{n!} \int_{c_1} \dots \int_{c_n} \varphi\big( T_{c,\varepsilon}^t(x,v) \big) \mathds{1}_{X_0}(c) \mathds{1}_{\mathcal{P}_\text{patho.}^c}(c) \mathds{1}_{\{ c_i \in B_\varepsilon \forall 1 \leq i \leq n \}} \dd c_n \dots \dd c_1 \nonumber\\
&\hspace{-1mm}= \sum_{n=1}^{+\infty} e^{-\mu_\varepsilon \vert B_\varepsilon \vert} \frac{\mu_\varepsilon^n}{n!} \varphi (x+tv,v) \int_{c_1} \dots \int_{c_n} \mathds{1}_{X_0}(c) \mathds{1}_{\mathcal{P}_\text{patho.}^c}(c) \mathds{1}_{\{ c_i \in B_\varepsilon \forall 1 \leq i \leq n \}} \mathds{1}_{\substack{\, \text{no obstacle}\\ \text{is internal}}} \dd c_n \dots \dd c_1 \nonumber\\
&\hspace{5mm} +\sum_{n = 1}^{+\infty} \sum_{k = 1}^n e^{-\mu_\varepsilon \vert B_\varepsilon \vert} \frac{\mu_\varepsilon^n}{n!} \frac{n!}{k!(n-k)!} \int_{c_1} \dots \int_{c_n} \varphi\big( T_{c,\varepsilon}^t(x,v) \big) \mathds{1}_{X_0}(c) \mathds{1}_{\mathcal{P}_\text{patho.}^c}(c) \nonumber\\
&\hspace{70mm} \times \mathds{1}_{\{ c_i \in B_\varepsilon \forall 1 \leq i \leq n \}} \mathds{1}_{\substack{\text{only the }k\\ \text{first obstacles}\\ c_j,\, 1\leq j \leq k,\\ \text{are internal}}} \dd c_n \dots \dd c_1,
\end{align}
which allows in the first term to integrate with respect to obstacles, since all of them have to be outside the tube $\big\{ y \in \mathbb{R}^d\ /\ \exists s \in [0,t], z \in \overline{B(0,\varepsilon)}\ /\ y = x + sv + z \big\}$. We obtain:
\begin{align}
\label{EQUATReecrSerie_Phi1}
&\sum_{n = 1}^{+\infty} e^{-\mu_\varepsilon \vert B_\varepsilon \vert} \frac{\mu_\varepsilon^n}{n!} \int_{c_1} \dots \int_{c_n} \varphi\big( T_{c,\varepsilon}^t(x,v) \big) \mathds{1}_{X_0}(c) \mathds{1}_{\mathcal{P}_\text{patho.}^c}(c) \mathds{1}_{\{ c_i \in B_\varepsilon \forall 1 \leq i \leq n \}} \dd c_n \dots \dd c_1 \nonumber\\
&\hspace{-1mm} = e^{- \mu_\varepsilon \vert [x,x+tv] + \overline{B(0,\varepsilon)} \vert} \varphi(x+tv,v) - e^{-\mu_\varepsilon \vert B_\varepsilon \vert} \varphi(x+tv,v) \\
&\hspace{-1mm} +\sum_{n = 1}^{+\infty} \sum_{k = 1}^n e^{-\mu_\varepsilon \vert B_\varepsilon \vert} \frac{\mu_\varepsilon^n}{k!(n-k)!} \int_{c_1} \dots \int_{c_n} \varphi\big( T_{c,\varepsilon}^t(x,v) \big) \mathds{1}_{X_0}(c) \mathds{1}_{\mathcal{P}_\text{patho.}^c}(c) \mathds{1}_{\{ c_i \in B_\varepsilon \forall 1 \leq i \leq n \}} \mathds{1}_{\substack{\text{only the }k\\ \text{first obstacles}\\ c_j,\, 1\leq j \leq k,\\ \text{are internal}}} \dd c_n \dots \dd c_1 \nonumber
\end{align}
By definition, $T_{c,\varepsilon}^t(x,v)$ does not depend on the position of any of the external obstacles. We introduce the definition of the \emph{dynamical tube}, which corresponds to the set of positions that are at a distance smaller or equal than $\varepsilon$ from the trajectory of the tagged particle. In other words, we define:
\begin{align}
\label{EQUATDynam_Tube}
\mathcal{T}^t(x,v; c,\varepsilon)
= \big\{ y \in \mathbb{R}^d\ /\ \inf_{s\in[0,t]} \vert x_{c,\varepsilon}(s) - y \vert \leq \varepsilon \big\}, \hspace{3mm} \text{that we will also denote, in short, by} \hspace{3mm} \mathcal{T}^t(c)
\end{align}
where we recall that $x_{c,\varepsilon}(s) \in \mathbb{R}^d$ corresponds to the position of the particle, namely the first component of the flow, that is $T_{c,\varepsilon}^s(x,v) = \big(x_{c,\varepsilon}(s),v_{c,\varepsilon}(s)\big)$ (see Definition \ref{DEFINFlot_TempsCrois}). To light the notation we  simply denoted the dynamical tube $\mathcal{T}^t(c)=\mathcal{T}^t(x,v;c,\varepsilon)$ omitting the dependence on $\varepsilon, x,v$.\\  
\noindent
By definition, the internal obstacles belong to the (boundary of the) dynamical tube $\mathcal{T}^t(c)$, while the external obstacles are outside. Denoting by $c^k$ the first $k$ obstacles, that is $c^k = (c_1,\dots,c_k)$, by $c^{k+1}$ the last $n-k$ obstacles, that is $c^{k+1} = (c_{k+1},\dots,c_n)$, and observing that:
\begin{align}
\mathds{1}_{X_0}(c)& \mathds{1}_{\mathcal{P}^c_\text{patho.}}(c) \mathds{1}_{c_i\, \text{is internal } \forall\, 1 \leq i \leq k} \mathds{1}_{c_j\, \text{is external } \forall\, (k+1) \leq j \leq n} \nonumber\\
&= \mathds{1}_{X_0}(c^k) \mathds{1}_{X_0}(c^{k+1}) \mathds{1}_{\mathcal{P}^c_\text{patho.}}(c^k) \mathds{1}_{c_i\, \text{is internal } \forall\, 1 \leq i \leq k} \mathds{1}_{c_i \in \big( \mathcal{T}^t(c^k) \big)^c\ \forall\, (k+1) \leq i \leq n},
\end{align}
we can integrate with respect to the positions of the last scatterers, which provides:
\begin{align}
&\sum_{n = 1}^{+\infty} \sum_{k = 1}^n e^{-\mu_\varepsilon \vert B_\varepsilon \vert} \frac{\mu_\varepsilon^n}{k!(n-k)!} \int_{c_1} \dots \int_{c_n} \varphi\big( T_{c,\varepsilon}^t(x,v) \big) \mathds{1}_{X_0}(c) \mathds{1}_{\mathcal{P}_\text{patho.}^c}(c) \mathds{1}_{\{ c_i \in B_\varepsilon \forall 1 \leq i \leq n \}} \mathds{1}_{\substack{\text{only the }k\\ \text{first obstacles}\\ c_j,\, 1\leq j \leq k,\\ \text{are internal}}} \dd c_n \dots \dd c_1 \nonumber\\
&= \sum_{n = 1}^{+\infty} \sum_{k = 1}^n e^{-\mu_\varepsilon \vert B_\varepsilon \vert} \frac{\mu_\varepsilon^n}{k!(n-k)!} \int_{c_1} \dots \int_{c_k} \varphi\big( T_{c,\varepsilon}^t(x,v) \big) \mathds{1}_{X_0}(c^k) \mathds{1}_{\mathcal{P}_\text{patho.}^c}(c^k) \mathds{1}_{\{ c_i \in B_\varepsilon \forall 1 \leq i \leq k \}} \mathds{1}_{\substack{\text{the obstacles}\\ c_j,\, 1\leq j \leq k,\\ \text{are internal}}} \nonumber\\
&\hspace{10mm} \times \Bigg[ \int_{c_{k+1}} \hspace{-2mm}\dots\int_{c_n} \mathds{1}_{X_0}(c^{k+1}) \mathds{1}_{\{ c_i \in B_\varepsilon \forall (k+1) \leq i \leq n \}} \mathds{1}_{c_j \in \big( \mathcal{T}^t(c^k) \big)^c\ \forall\, (k+1) \leq j \leq n} \dd c_n \dots \dd c_{k+1} \Bigg] \dd c_k \dots \dd c_1 \nonumber\\
&= \sum_{n = 1}^{+\infty} \sum_{k = 1}^n e^{-\mu_\varepsilon \vert B_\varepsilon \vert} \frac{\mu_\varepsilon^n}{k!(n-k)!} \int_{c_1} \dots \int_{c_k} \varphi\big( T_{c,\varepsilon}^t(x,v) \big) \mathds{1}_{X_0}(c^k) \mathds{1}_{\mathcal{P}_\text{patho.}^c}(c^k) \mathds{1}_{\{ c_i \in B_\varepsilon \forall 1 \leq i \leq k \}} \mathds{1}_{\substack{\text{the obstacles}\\ c_j,\, 1\leq j \leq k,\\ \text{are internal}}} \nonumber\\
&\hspace{80mm} \times \Big\vert B_\varepsilon \backslash \big( \overline{B(x,\varepsilon)} \cup \mathcal{T}^t(c^k) \big) \Big\vert^{n-k} \dd c_k \dots \dd c_1.
\end{align}
Inverting the two sums, and performing the change of variable $m=n-k$ we find:
\begin{align}
&\sum_{n = 1}^{+\infty} \sum_{k = 1}^n e^{-\mu_\varepsilon \vert B_\varepsilon \vert} \frac{\mu_\varepsilon^n}{k!(n-k)!} \int_{c_1} \dots \int_{c_n} \varphi\big( T_{c,\varepsilon}^t(x,v) \big) \mathds{1}_{X_0}(c) \mathds{1}_{\mathcal{P}_\text{patho.}^c}(c) \mathds{1}_{\{ c_i \in B_\varepsilon \forall 1 \leq i \leq n \}} \mathds{1}_{\substack{\text{only the }k\\ \text{first obstacles}\\ c_j,\, 1\leq j \leq k,\\ \text{are internal}}} \dd c_n \dots \dd c_1 \nonumber\\
&\hspace{0mm} = \sum_{k=1}^{+\infty} \sum_{m=0}^{+\infty} e^{-\mu_\varepsilon \vert B_\varepsilon \vert}\frac{\mu_\varepsilon^k \mu_\varepsilon^m}{k! \, m!} \int_{c_1} \dots \int_{c_k} \varphi\big( T_{c,\varepsilon}^t(x,v) \big) \mathds{1}_{X_0}(c^k) \mathds{1}_{\mathcal{P}_\text{patho.}^c}(c^k) \mathds{1}_{\{ c_i \in B_\varepsilon \forall 1 \leq i \leq k \}} \mathds{1}_{\substack{\text{the obstacles}\\ c_j,\, 1\leq j \leq k,\\ \text{are internal}}} \nonumber\\
&\hspace{85mm} \times \Big\vert B_\varepsilon \backslash \big( \overline{B(x,\varepsilon)} \cup \mathcal{T}^t(c^k) \big) \Big\vert^m \dd c_k \dots \dd c_1.
\end{align}
Regrouping the terms that depend on the index $m$ in order to identify an exponential series, we obtain:
\begin{align}
&\sum_{n = 1}^{+\infty} \sum_{k = 1}^n e^{-\mu_\varepsilon \vert B_\varepsilon \vert} \frac{\mu_\varepsilon^n}{k!(n-k)!} \int_{c_1} \dots \int_{c_n} \varphi\big( T_{c,\varepsilon}^t(x,v) \big) \mathds{1}_{X_0}(c) \mathds{1}_{\mathcal{P}_\text{patho.}^c}(c) \mathds{1}_{\{ c_i \in B_\varepsilon \forall 1 \leq i \leq n \}} \mathds{1}_{\substack{\text{only the }k\\ \text{first obstacles}\\ c_j,\, 1\leq j \leq k,\\ \text{are internal}}} \dd c_n \dots \dd c_1 \nonumber\\
&\hspace{0mm} = \sum_{k=1}^{+\infty} \frac{\mu_\varepsilon^k}{k!} \int_{c_1} \dots \int_{c_k} e^{- \mu_\varepsilon \vert \overline{B(x,\varepsilon)} \, \cup \, \mathcal{T}^t(c^k) \vert} \varphi\big( T_{c,\varepsilon}^t(x,v) \big) \mathds{1}_{X_0}(c^k) \nonumber\\
&\hspace{70mm}\times \mathds{1}_{\mathcal{P}_\text{patho.}^c}(c^k) \mathds{1}_{\{ c_i \in B_\varepsilon \forall 1 \leq i \leq k \}} \mathds{1}_{\substack{\text{the obstacles}\\ c_j,\, 1\leq j \leq k,\\ \text{are internal}}} \dd c_k \dots \dd c_1.
\end{align}
$\varphi_\varepsilon^{(1)}$ being the sum of the term $e^{-\mu_\varepsilon \vert B_\varepsilon \vert} \varphi(x+tv,v)$ and of the series rewritten in \eqref{EQUATReecrSerie_Phi1}, this concludes the proof of Proposition \ref{PROPOReecrPhiEpIntegObsta}.
\end{proof}

\subsection{Comparing the series representation of $\varphi_\varepsilon$ with the solution of the adjoint equation}

In order to compare solutions to the adjoint equation with the expression \eqref{EQUATDefinPhi_1Propo} of $\varphi_\varepsilon = \mathbb{E}_{\mu_\varepsilon} \big[ \varphi\big( \widetilde{T}_{c,\varepsilon}^t(x,v) \big)\big]$ given by Proposition \ref{PROPOReecrPhiEpIntegObsta}, we rely on the series description of the solutions to the adjoint equation \eqref{EQUATLineaBoltzSphDuFormeAdjoi} provided by Proposition \ref{PROPORepreSerieEquatAdjoi}. 

To this aim, we consider \eqref{EQUATDefinPhi_1Propo} (see Proposition \ref{PROPOReecrPhiEpIntegObsta}) and \eqref{EQUATRepreSerieEquatAdjoi} (see Proposition \ref{PROPORepreSerieEquatAdjoi}) and estimate the difference between these two expressions.  We will now compare the two quantities:
\begin{align}
T_{c,\varepsilon}^t(x,v) \hspace{5mm} \text{and} \hspace{5mm} \big( x + \sum_{m=1}^k \big( t_{m-1}-t_m \big) v^{(m-1)} + t_k v^{(k)},v^{(k)} \big),
\end{align}
that are the respective variables of the initial datum $\varphi$ in \eqref{EQUATDefinPhi_1Propo} and \eqref{EQUATRepreSerieEquatAdjoi}.\\
There are two main differences between these two trajectories. On the one hand, the trajectory $T_{c,\varepsilon}^t(x,v)$ of the tagged particle, evolving among obstacles with a positive size, might experience \emph{recollisions}. In other words, some obstacles can be collided several times, inducing correlations. The microscopic dynamics is therefore non-Markovian.\\ 
On the other hand, we consider the trajectory $x + \sum_{m=1}^k \big( t_{m-1}-t_m \big) v^{(m-1)} + t_k v^{(k)}$ associated to the adjoint Boltzmann equation, which is the limit stochastic trajectory of the microscopic Lorentz process. We will call such a limit trajectory a \emph{pseudo-trajectory}. We  first observe that no recollision can take place in such pseudo-trajectories, in the sense that the parameters $t_m$ and $\omega_m$ are always chosen independently from each other. This exhibits the memoryless nature of the process which generates such trajectories. In addition, we observe that the position of an ``obstacle'', for  
a trajectory associated to the adjoint equation, can be exactly on a part of the trajectory that connects two other obstacles which have been collided before. In other words, we may have:
\begin{align}
x(n) \in [x(p),x(p+1)] \hspace{5mm} \text{for} \hspace{5mm} n > p+1,
\end{align}
where:
\begin{align}
x(n) = x + \sum_{m=1}^n \big( t_{m-1}-t_m \big) v^{(m-1)}.
\end{align}
Such a phenomenon is called an \emph{interference}. It is important to observe that such a phenomenon cannot take place at the level of the particle system, with obstacles of a positive size. Indeed, a third obstacle cannot lie on the segment between two consecutive obstacles that are collided: the tagged particle would have collided with this third obstacle, violating the definition of the two consecutively collided obstacles.\\
In summary, the parametrizations of the pseudo-trajectories and the inelastic hard sphere flow present significant difference. The parametrization of the pseudo-trajectories allows to consider dynamics that cannot be achieved by a particle evolving according to the hard sphere flow, for which no interference might take place. Conversely, the parametrization of the hard sphere flow encodes also trajectories that present recollisions, which is not the case for the parametrization of the pseudo-trajectories.\\
Nevertheless, both interferences and recollisions correspond to extremely rare situtations in the low density limit. The purpose of Proposition \ref{PROPOReecrPhiEpElimiRecol} that follows is to estimate precisely the subsets of the domains of the integrals in \eqref{EQUATDefinPhi_1Propo} leading either to an interference or to a recollision, so that $\varphi_\varepsilon = \mathbb{E}_{\mu_\varepsilon}\big[ \varphi \big( \widetilde{T}_{c,\varepsilon}^t(x,v) \big)\big]$ and $\widetilde{\varphi}$, solution of the adjoint equation, can be easily compared on the complement of such pathological sets.

\begin{propo}[Elimination of the recollisions]
\label{PROPOReecrPhiEpElimiRecol}
Let $\varepsilon > 0$. Let $c$ be a $d$-dimensional Poisson process of intensity $\mu_\varepsilon > 0$ such that the \emph{Boltzmann-Grad} scaling \eqref{EQUATBoltzmann-Grad_Limit} holds. Then, there exists a universal constant $\varepsilon_0 > 0$ and two positive constants $\widetilde{C}_1 = \widetilde{C}_1(d,r)$ and $\widetilde{C}_2 = \widetilde{C}_2(d,r)$ that depend only on the dimension $d$ and the restitution coefficient $r$ such that, for any $0 < \varepsilon \leq \varepsilon_0$, and any $\mathcal{C}_0(\mathbb{R}^d\times\mathbb{R}^d)$ function $\varphi$, the quantity $\varphi_\varepsilon$, defined as \eqref{EQUATDefinPhiEp_Avec_FlotDefin}, satisfies:
\begin{align}
\varphi_\varepsilon = \varphi_\varepsilon^{(2)} + R_1 + R_2 \hspace{5mm} \text{with} \hspace{5mm} \varphi_\varepsilon^{(2)} = \sum_{k=1}^{+\infty} \varphi_{\varepsilon,k}^{(2)}(k)
\end{align}
and with $\varphi_{\varepsilon,k}^{(2)}(k)$ defined as
\begin{align}
\varphi_{\varepsilon,k}^{(2)}(k) &= e^{- \mu_\varepsilon \vert [x,x+tv] + \overline{B(0,\varepsilon)} \vert} \varphi(x+tv,v) \nonumber\\
&\hspace{3mm}+ \int_{t_1 = 0}^t \int_{\mathbb{S}^{d-1}_{\omega_1}} \mathds{1}_{\widetilde{\mathcal{P}}_{c_1}^c} \int_{t_2=0}^{t_1} \int_{\mathbb{S}^{d-1}_{\omega_2}} \mathds{1}_{\widetilde{\mathcal{P}}_{c_2}^c} \dots \int_{t_k=0}^{t_{k-1}} \int_{\mathbb{S}_{\omega_k}^{d-1}} \mathds{1}_{\widetilde{\mathcal{P}}_{c_k}^c} e^{- \mu_\varepsilon \vert \overline{B(x,\varepsilon)} \, \cup \, \mathcal{T}^t(c^k) \vert} \prod_{l=1}^k \vert v^{(l-1)} \cdot \omega_l \vert \nonumber\\
&\hspace{33.75mm} \times \varphi\big( x + \sum_{m=1}^k \big( t_{m-1}-t_m \big) v^{(m-1)} + t_k v^{(k)},v^{(k)} \big) \dd \omega_k \dd t_k \dd \omega_2 \dd t_2 \dd \omega_1 \dd t_1,
\end{align}
$R_1$ as in \eqref{EQUATDefinReste_R_1_}, and $R_2$ being a remainder term such that:
\begin{align}
\label{EQUATDefinReste_R_2_PropoElimiRecol}
\big\vert R_2  \big\vert \leq \vertii{\varphi}_\infty \widetilde{C}_1 \big( 1 + \max(1,t) \vert v \vert \big) e^{\widetilde{C}_2 t \vert v \vert} \varepsilon^{1/4}.
\end{align}
\end{propo}

\begin{proof}
We restart from the main term $\varphi_\varepsilon^{(1)}$ of $\varphi_\varepsilon$ given by Proposition \ref{PROPOReecrPhiEpIntegObsta}. First, we observe that, up to relabel the obstacles, we can always order them, in the following sense:
\begin{itemize}
\item $c_1$ is the first obstacle that is collided by the tagged particle evolving according to the trajectory $T_{c,\varepsilon}^s(x,v)$ of the inelastic hard sphere flow, that is:
\begin{align}
\hspace{-10mm} x(t-t_1) \in \overline{B(c_1,\varepsilon)} \hspace{2mm} \text{for a certain} \hspace{2mm} 0 < t_1 < t \hspace{5mm} \text{and} \hspace{5mm} t-t_1 = \min\{ s \in [0,t] \ /\ T_{c,\varepsilon}^s(x,v) \in c + \overline{B(0,\varepsilon)} \},
\end{align}
\item $c_2$ is the next first obstacle to be collided by the tagged particle:
\begin{align}
\hspace{-8mm} x(t-t_2) \in \overline{B(c_2,\varepsilon)} \hspace{2mm} \text{for a certain} \hspace{2mm} 0 < t_2 < t_1 \hspace{5mm} \text{and} \hspace{5mm} t-t_2 = \min\{ s \in [0,t] \ /\ T_{c,\varepsilon}^s(x,v) \in \big( c \backslash \{ c_1 \} \big) + \overline{B(0,\varepsilon)} \}.
\end{align}
\item Similarly, we order the other obstacles:
\begin{align}
\hspace{-8mm} &x(t-t_j) \in \overline{B(c_j,\varepsilon)} \hspace{2mm} \text{for a certain} \hspace{2mm} 0 < t_j < t_{j-1} \nonumber\\
&\hspace{20mm}\text{and} \hspace{5mm} t-t_j = \min\{ s \in [0,t] \ /\ T_{c,\varepsilon}^s(x,v) \in \big( c \backslash \{ c_1,\dots,c_{j-1} \} \big) + \overline{B(0,\varepsilon)} \}.
\end{align}
\end{itemize}
Observe that, a priori, the obstacles $c_1,\dots,c_j$ might be collided more than once by the tagged particle before $t-t_{j+1}$. Nevertheless, it is impossible in the case of $c_1$ before $t-t_2$: indeed, after the first collision with $c_1$, the tagged particle is in a post-collisional configuration with $c_1$, and the tagged particle cannot collide by definition with any other scatterer before $t-t_2$.\\
\newline
Since in the expression \eqref{EQUATDefinPhi_1Propo} of $\varphi_\varepsilon^{(1)}$ the obstacles are not ordered, we obtain $k!$ similar terms after the relabelling in the case of $k$ internal obstacles. Therefore we have:
\begin{align}
\varphi_\varepsilon^{(1)} &= e^{-\mu_\varepsilon \vert [x,x+tv] + \overline{B(0,\varepsilon)} \vert}\varphi(x+tv,v) \nonumber\\
&\hspace{5mm}+ \sum_{k=1}^{+\infty} \mu_\varepsilon^k \int_{c_1} \dots \int_{c_k} e^{- \mu_\varepsilon \vert \overline{B(x,\varepsilon)} \, \cup \, \mathcal{T}^t(c^k) \vert} \varphi\big( T_{c,\varepsilon}^t(x,v) \big) \mathds{1}_{X_0}(c^k) \mathds{1}_{\mathcal{P}_\text{patho.}^c}(c^k) \mathds{1}_{\{ c_i \in B_\varepsilon \forall\, 1 \leq i \leq k \}} \nonumber\\
&\hspace{70mm} \times \mathds{1}_{\substack{\text{the obstacles}\\ c_j,\, 1\leq j \leq k,\\ \text{are internal}}}(c^k) \mathds{1}_{\substack{\text{the $l$-th collided} \\ \text{obstacle is }c_l \\ \forall\, 1 \leq l \leq k}}(c^k) \dd c_k \dots \dd c_1.
\end{align}
We will denote by $\varphi_\varepsilon^{(2)}(k)$ ($k \geq 1$) the different terms in the previous series, that is $\varphi_\varepsilon^{(1)} = e^{-\mu_\varepsilon \vert [x,x+tv] + \overline{B(0,\varepsilon)} \vert}$ $\times\varphi(x+tv,v) + \sum_{k=1}^{+\infty} \varphi_\varepsilon^{(2)} (k)$ with
\begin{align}
\varphi_\varepsilon^{(2)} (k) &= \mu_\varepsilon^k \int_{c_1} \dots \int_{c_k} e^{- \mu_\varepsilon \vert \overline{B(x,\varepsilon)} \, \cup \, \mathcal{T}^t(c^k) \vert} \varphi\big( T_{c,\varepsilon}^t(x,v) \big) \mathds{1}_{X_0}(c^k) \mathds{1}_{\mathcal{P}_\text{patho.}^c}(c^k) \mathds{1}_{\{ c_i \in B_\varepsilon \forall\, 1 \leq i \leq k \}} \nonumber\\
&\hspace{60mm} \times \mathds{1}_{\substack{\text{the obstacles}\\ c_j,\, 1\leq j \leq k,\\ \text{are internal}}}(c^k) \mathds{1}_{\substack{\text{the $l$-th collided} \\ \text{obstacle is }c_l \\ \forall\, 1 \leq l \leq k}}(c^k) \dd c_k \dots \dd c_1.
\end{align}
\paragraph{Reparametrization of the positions of the scatterers.} After the relabelling, $c_1$ is the center of the first obstacle collided by the tagged particle, at time $t-t_1$. Therefore, before this time the trajectory of the tagged particle is the free flow:
\begin{align}
T_{c,\varepsilon}^s(x,v) = \big( x+sv,v \big) \hspace{3mm} \forall\ s \in [0,t-t_1[.
\end{align}
In addition, the position $c_1$ can be parametrized as follows:
\begin{align}
\label{EQUATParamObsta_c_1_}
c_1 = x + (t-t_1) v - \varepsilon \omega_1,
\end{align}
for a certain $\omega_1 \in \mathbb{S}^{d-1}$ such that $\omega_1 \cdot v \leq 0$ (because if the scalar product is positive, there exists a smaller time $s$ such that $x(s) = x + sv$ belongs to $\overline{B(c_1,\varepsilon)}$, which is absurd). Provided that $\omega_1\cdot v \leq 0$, observe that any choice of the parameters $(t_1,\omega_1) \in\ ]0,t[ \times\, \mathbb{S}^{d-1}$ corresponds to an admissible choice for the position $c_1$ of the first scatterer. In particular, the condition $c_1 \notin \overline{B(x,\varepsilon)}$ systematically holds true.\\
Based on \eqref{EQUATParamObsta_c_1_}, we can perform the change of variables $c_1 \rightarrow c_1(t_1,\omega_1)$. The Jacobian determinant of this change of variables is:
\begin{align}
\label{EQUATJacobChangVariaParam__c__}
\text{J}c_1(t_1,\omega_1) = \varepsilon^{d-1} \big\vert v \cdot \omega_1 \big\vert.
\end{align}
We find therefore:
\begin{align}
\varphi_\varepsilon^{(2)}(k) &= \mu_\varepsilon^k \int_{t_1 = 0}^t \int_{\mathbb{S}^{d-1}_{\omega_1}} \mathds{1}_{\omega_1 \cdot v \leq 0} \int_{c_2} \dots \int_{c_k} e^{- \mu_\varepsilon \vert \overline{B(x,\varepsilon)} \, \cup \, \mathcal{T}^t(c^k) \vert} \varphi\big( T_{c,\varepsilon}^t(x,v) \big) \mathds{1}_{X_0}(c^k) \mathds{1}_{\mathcal{P}_\text{patho.}^c}(c^k) \mathds{1}_{\{ c_i \in B_\varepsilon \forall\, 2 \leq i \leq k \}} \nonumber\\
&\hspace{42mm} \times \mathds{1}_{\substack{\text{the obstacles}\\ c_j,\, 2\leq j \leq k,\\ \text{are internal}}}(c^k) \mathds{1}_{\substack{\text{the $l$-th collided} \\ \text{obstacle is }c_l \\ \forall\, 2 \leq l \leq k}}(c^k) \varepsilon^{d-1} \vert v \cdot \omega_1 \vert \dd c_k \dots \dd c_2 \dd \omega_1 \dd t_1.
\end{align}
\paragraph{The cut-off concerning the integration variable $\omega_1$.} In order to estimate more easily the pathological positions $c_2$ of the second scatterer, we define the following subset of the positions $c_1$ of the first scatterer:
\begin{align}
\label{EQUATDefinPathoDelt1}
\widetilde{\mathcal{P}}_{c_1} = \{ \omega_1 \in \mathbb{S}^{d-1} \ /\ \omega_1 \cdot v \leq 0 \ \text{and}\ 1 - \Big\vert \frac{v'}{\vert v' \vert}\cdot\frac{v}{\vert v \vert} \Big\vert \leq \varepsilon^{\delta_1} \},
\end{align}
for a certain constant $\delta_1 > 0$ that will be chosen later. We denote the complement of $\widetilde{\mathcal{P}}_{c_1}$ by $\widetilde{\mathcal{P}}_{c_1}^c$ that is $\{ \omega_1 \in \mathbb{S}^{d-1}\ /\ \omega_1\cdot v \leq 0 \} = \widetilde{\mathcal{P}}_{c_1} \cup \widetilde{\mathcal{P}}_{c_1}^c$, and decompose then:
\begin{align}
\varphi_\varepsilon^{(2)}(k) &= \mu_\varepsilon^k \int_{t_1 = 0}^t \int_{\mathbb{S}^{d-1}_{\omega_1}} \mathds{1}_{\widetilde{\mathcal{P}}_{c_1}^c} \int_{c_2} \dots \int_{c_k} e^{- \mu_\varepsilon \vert \overline{B(x,\varepsilon)} \, \cup \, \mathcal{T}^t(c^k) \vert} \varphi\big( T_{c,\varepsilon}^t(x,v) \big) \mathds{1}_{X_0}(c^k) \mathds{1}_{\mathcal{P}_\text{patho.}^c}(c^k) \mathds{1}_{\{ c_i \in B_\varepsilon \forall\, 2 \leq i \leq k \}} \nonumber\\
&\hspace{42mm} \times \mathds{1}_{\substack{\text{the obstacles}\\ c_j,\, 2\leq j \leq k,\\ \text{are internal}}}(c^k) \mathds{1}_{\substack{\text{the $l$-th collided} \\ \text{obstacle is }c_l \\ \forall\, 2 \leq l \leq k}}(c^k) \varepsilon^{d-1} \vert v \cdot \omega_1 \vert \dd c_k \dots \dd c_2 \dd \omega_1 \dd t_1 \nonumber\\
&\hspace{3mm}+ \mu_\varepsilon^k \int_{t_1 = 0}^t \int_{\mathbb{S}^{d-1}_{\omega_1}} \mathds{1}_{\widetilde{\mathcal{P}}_{c_1}} \int_{c_2} \dots \int_{c_k} e^{- \mu_\varepsilon \vert \overline{B(x,\varepsilon)} \, \cup \, \mathcal{T}^t(c^k) \vert} \varphi\big( T_{c,\varepsilon}^t(x,v) \big) \mathds{1}_{X_0}(c^k) \mathds{1}_{\mathcal{P}_\text{patho.}^c}(c^k) \mathds{1}_{\{ c_i \in B_\varepsilon \forall\, 2 \leq i \leq k \}} \nonumber\\
&\hspace{42mm} \times \mathds{1}_{\substack{\text{the obstacles}\\ c_j,\, 2\leq j \leq k,\\ \text{are internal}}}(c^k) \mathds{1}_{\substack{\text{the $l$-th collided} \\ \text{obstacle is }c_l \\ \forall\, 2 \leq l \leq k}}(c^k) \varepsilon^{d-1} \vert v \cdot \omega_1 \vert \dd c_k \dots \dd c_2 \dd \omega_1 \dd t_1,
\end{align}
the second term corresponding to a pathological choice of scatterer positions, that we will estimate later. This second term will be denoted by $\varphi_{\varepsilon,c_1}^{(2),\text{patho.}}(k)$, while the first term will be denoted by $\varphi_{\varepsilon,1}^{(2)}(k)$, so that $\varphi_\varepsilon^{(2)}(k) = \varphi_{\varepsilon,1}^{(2)}(k) + \varphi_{\varepsilon,c_1}^{(2),\text{patho.}}(k)$.\\
Concerning the position of the second obstacle, since  no other collision  can happen on $[0,t-t_2[$ except the only collision, at time $t-t_1$, between the tagged particle and $c_1$, we deduce that $c_2$ necessarily writes:
\begin{align}
c_2 = x + (t-t_1)v + (t_1-t_2)v' - \varepsilon\omega_2,
\end{align}
for certain $t_2 \in\ ]0,t_1[$ and $\omega_2 \in \mathbb{S}^{d-1}$ such that $\omega_2 \cdot v' \leq 0$. This time, not all the choices of the parameters $(t_2,\omega_2)$ lead to an admissible position $c_2$ for the second scatterer. Indeed, one needs to ensure that:
\begin{align}
\label{EQUATCondi_c_2_}
c_2 \notin \overline{B(x,\varepsilon)}, \hspace{5mm} d\big(c_2,T_{c,\varepsilon}^s(x,v)\big) > \varepsilon\ \forall s \in [0,t-t_1],
\end{align}
the first condition corresponding to the fact that the tagged particle is positionned initially outside the scatterer centered on $c_2$, while the second condition describes the absence of interference.

\paragraph{The cut-off concerning the integration variables $t_2$ and $\omega_2$.} Let us denote by $\mathcal{A}_{c_2} \subset\ ]0,t_1[ \times \mathbb{S}^{d-1}$ the set of parameters $(t_2,\omega_2)$ such that the conditions \eqref{EQUATCondi_c_2_} are satisfied, as well as the other conditions already described by the product of the indicator functions:
\begin{align*}
\mathds{1}_{X_0}(c^k) \mathds{1}_{\mathcal{P}_\text{patho.}^c}(c^k) \mathds{1}_{\{ c_i \in B_\varepsilon \forall\, 2 \leq i \leq k \}} \mathds{1}_{\substack{\text{the obstacles}\\ c_j,\, 2\leq j \leq k,\\ \text{are internal}}}(c^k) \mathds{1}_{\substack{\text{the $l$-th collided} \\ \text{obstacle is }c_l \\ \forall\, 2 \leq l \leq k}}(c^k).
\end{align*}
In particular, we have $c_k \in \mathcal{P}_\text{patho.}^c$. We perform now the change of variables $c_2 \rightarrow c_2(t_2,\omega_2)$, of Jacobian determinant $\vert v' \cdot \omega_2 \vert$, and the first term $\varphi_{\varepsilon,1}^{(2)}(k)$ of $\varphi_\varepsilon^{(2)}(k)$ can be written as:
\begin{align}
\varphi_{\varepsilon,1}^{(2)}(k) &= \mu_\varepsilon^k \int_{t_1 = 0}^t \int_{\mathbb{S}^{d-1}_{\omega_1}} \mathds{1}_{\widetilde{\mathcal{P}}_{c_1}^c} \int_{t_2=0}^{t_1} \int_{\mathbb{S}^{d-1}_{\omega_2}} \mathds{1}_{\mathcal{A}_{c_2}} \int_{c_3} \dots \int_{c_k} e^{- \mu_\varepsilon \vert \overline{B(x,\varepsilon)} \, \cup \, \mathcal{T}^t(c^k) \vert} \varphi\big( T_{c,\varepsilon}^t(x,v) \big) \mathds{1}_{X_0}(c^k) \nonumber\\
&\hspace{45mm} \times \mathds{1}_{\mathcal{P}_\text{patho.}^c}(c^k) \mathds{1}_{\{ c_i \in B_\varepsilon \forall\, 3 \leq i \leq k \}} \mathds{1}_{\substack{\text{the obstacles}\\ c_j,\, 3 \leq j \leq k,\\ \text{are internal}}}(c^k) \mathds{1}_{\substack{\text{the $l$-th collided} \\ \text{obstacle is }c_l \\ \forall\, 3 \leq l \leq k}}(c^k) \nonumber\\ &\hspace{60mm}\times \varepsilon^{2(d-1)} \vert v \cdot \omega_1 \vert \cdot \vert v' \cdot \omega_2 \vert \dd c_k \dots \dd c_3 \dd \omega_2 \dd t_2 \dd \omega_1 \dd t_1.
\end{align}
From now on, we will not undertake to describe precisely the pathological parameters leading to recollisions or interferences. We will only estimate them, by introducing larger sets, such that outside those sets, the parameters will lead to well-defined dynamics as a consequence of elementary and direct arguments.\\
In the case of $c_2$, we decompose the set of parameters $\mathcal{A}_{c_2}$ as follows. We introduce:
\begin{align}
\widetilde{\mathcal{P}}_{c_2}^1 = \{ (t_2,\omega_2) \in \mathcal{A}_{c_2} \ /\ t_1-t_2 < \frac{\varepsilon^{\delta_2}}{r \vert v \vert} \}.
\end{align}
$\widetilde{\mathcal{P}}_{c_2}^1$ is the set such that the obstacles $c_1$ and $c_2$ are close, with $\delta_2 > 0$ a constant to be chosen later. We introduce also:
\begin{align}
\widetilde{\mathcal{P}}_{c_2}^2 = \{ (t_2,\omega_2) \in \mathcal{A}_{c_2} \ /\ x+(t-t_1)v + (t_1-t_2)v' - \varepsilon \omega_2 \in \overline{B(x,\varepsilon)} \},
\end{align}
$\widetilde{\mathcal{P}}_{c_2}^2$ is the set such that the tagged particle would be initially situated inside the second scatterer. And finally:
\begin{align}
\widetilde{\mathcal{P}}_{c_2}^3 = \{ (t_2,\omega_2) \in \mathcal{A}_{c_2} \ /\ x+(t-t_1)v+(t_1-t_2)v' - \varepsilon \omega_2 \in [x,x+(t-t_1)v] + \overline{B(0,\varepsilon)} \}.
\end{align}
$\widetilde{\mathcal{P}}_{c_2}^3$ is the set such that the second obstacle lies in between the initial position of the first particle and its position when it collides with the first scatterer. In other words, this set corresponds to an interference.\\
Observe that no recollision can occur before the collision with the second scatterer $c_2$, because the tagged particle is in post-collisional configuration with the first scatterer on the time interval $[t-t_1,t-t_2]$. Nevertheless, recollisions might occur after the collision with the second scatterer. For instance, the tagged particle can collide once again with $c_1$ after the time $t-t_2$. We will remove such a possibility by prescribing a condition on the velocity $v^{(2)}$ of the tagged particle after its collision with $c_2$. Since the velocity $v^{(2)}$ of the tagged particle after its collision with $c_2$ depends on the angular parameter $\omega_2$, we will impose a condition on this parameter, by introducing the following set:
\begin{align}
\label{EQUATDefinPathoDelt3}
\widetilde{\mathcal{P}}_{c_2}^4 = \{ (t_2,\omega_2) \in \mathcal{A}_{c_2} \ /\ 1 - \Big\vert \frac{v''}{\vert v'' \vert} \cdot \frac{c_2-c_1}{\vert c_2-c_1 \vert} \Big\vert \leq \varepsilon^{\delta_3}\},
\end{align}
with $\delta_3 > 0$ a positive constant to be chosen later.\\
Finally, we perform a last cut-off to prevent the trajectory to present segments between two consecutive collisions that are almost parallel, which will ensure that the time intervals during which the interfence may take place are small.
\begin{align}
\label{EQUATDefinPathoDelt1_bis_}
\widetilde{\mathcal{P}}_{c_2}^5 = \{ (t_2,\omega_2) \in \mathcal{A}_{c_2} \ /\ 1 - \Big\vert \frac{v''}{\vert v'' \vert} \cdot \frac{v}{\vert v \vert} \Big\vert \leq \varepsilon^{\delta_1} \hspace{3mm} \text{and} \hspace{3mm} 1 - \Big\vert \frac{v''}{\vert v'' \vert} \cdot \frac{v'}{\vert v' \vert} \Big\vert \leq \varepsilon^{\delta_1} \},
\end{align}
In the end, we define:
\begin{align}
\widetilde{\mathcal{P}}_{c_2} = \widetilde{\mathcal{P}}_{c_2}^1 \cup \widetilde{\mathcal{P}}_{c_2}^2 \cup \widetilde{\mathcal{P}}_{c_2}^3 \cup \widetilde{\mathcal{P}}_{c_2}^4 \cup \widetilde{\mathcal{P}}_{c_2}^5, \hspace{5mm} \widetilde{\mathcal{P}}_{c_2}^c = \mathcal{A}_{c_2} \backslash \widetilde{\mathcal{P}}_{c_2},
\end{align}
and two integrals $\varphi_{\varepsilon,c_2}^{(2),\text{patho.}}(k)$ and $\varphi_{\varepsilon,2}^{(2)}(k)$, which differ from each other only by the indicator functions $\mathds{1}_{\widetilde{\mathcal{P}}_{c_2}}$ and $\mathds{1}_{\widetilde{\mathcal{P}}^c_{c_2}}$ concerning the integration variables $(t_2,\omega_2)$:
\begin{align}
\varphi_{\varepsilon,c_2}^{(2),\text{patho.}}(k) &= \mu_\varepsilon^k \int_{t_1 = 0}^t \int_{\mathbb{S}^{d-1}_{\omega_1}} \mathds{1}_{\widetilde{\mathcal{P}}_{c_1}^c} \int_{t_2=0}^{t_1} \int_{\mathbb{S}^{d-1}_{\omega_2}} \mathds{1}_{\widetilde{\mathcal{P}}_{c_2}} \int_{c_3} \dots \int_{c_k} e^{- \mu_\varepsilon \vert \overline{B(x,\varepsilon)} \, \cup \, \mathcal{T}^t(c^k) \vert} \varphi\big( T_{c,\varepsilon}^t(x,v) \big) \mathds{1}_{X_0}(c^k) \nonumber\\
&\hspace{40mm} \times \mathds{1}_{\mathcal{P}_\text{patho.}^c}(c^k) \mathds{1}_{\{ c_i \in B_\varepsilon \forall\, 3 \leq i \leq k \}} \mathds{1}_{\substack{\text{the obstacles}\\ c_j,\, 3 \leq j \leq k,\\ \text{are internal}}}(c^k) \mathds{1}_{\substack{\text{the $l$-th collided} \\ \text{obstacle is }c_l \\ \forall\, 3 \leq l \leq k}}(c^k) \nonumber\\
&\hspace{55mm} \times \varepsilon^{2(d-1)} \vert v \cdot \omega_1 \vert \cdot \vert v' \cdot \omega_2 \vert \dd c_k \dots \dd c_3 \dd \omega_2 \dd t_2 \dd \omega_1 \dd t_1
\end{align}
and:
\begin{align}
\varphi_{\varepsilon,2}^{(2)}(k) &= \mu_\varepsilon^k \int_{t_1 = 0}^t \int_{\mathbb{S}^{d-1}_{\omega_1}} \mathds{1}_{\widetilde{\mathcal{P}}_{c_1}^c} \int_{t_2=0}^{t_1} \int_{\mathbb{S}^{d-1}_{\omega_2}} \mathds{1}_{\widetilde{\mathcal{P}}_{c_2}^c} \int_{c_3} \dots \int_{c_k} e^{- \mu_\varepsilon \vert \overline{B(x,\varepsilon)} \, \cup \, \mathcal{T}^t(c^k) \vert} \varphi\big( T_{c,\varepsilon}^t(x,v) \big) \mathds{1}_{X_0}(c^k) \nonumber\\
&\hspace{45mm} \times \mathds{1}_{\mathcal{P}_\text{patho.}^c}(c^k) \mathds{1}_{\{ c_i \in B_\varepsilon \forall\, 3 \leq i \leq k \}} \mathds{1}_{\substack{\text{the obstacles}\\ c_j,\, 3 \leq j \leq k,\\ \text{are internal}}}(c^k) \mathds{1}_{\substack{\text{the $l$-th collided} \\ \text{obstacle is }c_l \\ \forall\, 3 \leq l \leq k}}(c^k) \nonumber\\ &\hspace{60mm} \times \varepsilon^{2(d-1)} \vert v \cdot \omega_1 \vert \cdot \vert v' \cdot \omega_2 \vert \dd c_k \dots \dd c_3 \dd \omega_2 \dd t_2 \dd \omega_1 \dd t_1.
\end{align}
We have then:
\begin{align}
\varphi_{\varepsilon,1}^{(2)}(k) = \varphi_{\varepsilon,2}^{(2)}(k) + \varphi_{\varepsilon,c_2}^{(2),\text{patho.}}(k).
\end{align}

\paragraph{Eliminating the recollisions.} We will apply the same cut-off procedure for all the other scatterers $c_j$, with $j \geq 3$. This procedure presents an additional step with respect to the case of $c_2$. Let us present in detail the case of the third scatterer $c_3$, the case when $j \geq 4$ being exactly the same.\\
In the case of $c_3$, the hard sphere transport between $t-t_2$ (the first time of collision with the second scatterer) and $t-t_3$ (the first time of collision with the third scatterer) may not be only given by the free transport: the tagged particle may collide several time between the two first scatterers before reaching the third one. In this case, the parametrization of the position $c_3$ of this third scatterer becomes intricate, and it is a priori not clear how to proceed with the change of variables.\\
We will now choose the constants $\delta_1$ and $\delta_3$ (introduced respectively in \eqref{EQUATDefinPathoDelt1}, \eqref{EQUATDefinPathoDelt1_bis_} and \eqref{EQUATDefinPathoDelt3}) and the parameter $\varepsilon$ such that if $(t_2,\omega_2) \in \widetilde{\mathcal{P}}_{c_2}^c$, then no recollision can take place during the time interval $]t-t_2,t-t_3[$. In this case, the position $c_3$ of the third scatterer writes:
\begin{align}
\label{EQUATParam_c_3_}
c_3 = x + (t-t_1)v + (t_1-t_2)v' + (t_2-t_3)v'' - \varepsilon\omega_3,
\end{align}
for certain $t_3 \in\ ]0,t_2[$ and $\omega_3 \in \mathbb{S}^{d-2}$ such that $\omega_3 \cdot v'' \leq 0$. We can always parametrize $c_3$ as in \eqref{EQUATParam_c_3_}. Nevertheless, at this step it might be that this piecewise affine parametrization does not correspond to the trajectory of the tagged particle, because of the possible recollisions. By construction the distance between $c_1$ and $x(t-t_2)$ satisfies:
\begin{align}
\vert c_2 - c_1 \vert &= \big\vert x + (t-t_1)v + (t_1-t_2)v' - \varepsilon\omega_2 - x - (t-t_1)v + \varepsilon\omega_1 \big\vert \nonumber\\
&\geq \vert (t_1-t_2) v' \vert - 2\varepsilon \geq \varepsilon^{\delta_2} - 2\varepsilon \geq \frac{1}{2} \varepsilon^{\delta_2},
\end{align}
for $0 < \delta_2 < 1$ and $\varepsilon$ smaller than a certain $\varepsilon_1 = \varepsilon_1(\delta_2)$ that depends only on $\delta_2$. On the other hand, if we assume that a recollision takes place, there exists a time $\tau$ such that:
\begin{align}
x(t-\tau) = x + (t-t_1)v + (t_1-t_2)v' + (t_2-\tau)v'' = c_1 + \varepsilon\sigma
\end{align}
for a certain $\sigma \in \mathbb{S}^{d-1}$. We observe then that, since:
\begin{align}
x(t-\tau) - x(t-t_2) = c_1 + \varepsilon\sigma - c_2 - \varepsilon\omega_2,
\end{align}
we have:
\vspace{-6mm}
\begin{align}
\frac{c_1 -c_2}{\vert c_1 - c_2 \vert} \cdot \frac{x(t-\tau) - x(t-t_2)}{\vert x(t-\tau) - x(t-t_2) \vert} &= \frac{\vert c_1-c_2 \vert + \varepsilon( \sigma-\omega_2) \cdot \frac{c_1-c_2}{\vert c_1-c_2 \vert}}{\vert c_1-c_2 + \varepsilon\sigma - \varepsilon\omega_2 \vert} \nonumber\\
&\geq \frac{\vert c_1-c_2 \vert - 2\varepsilon}{\vert c_1-c_2 \vert + 2\varepsilon} \geq 1 - 4 \frac{\varepsilon}{\vert c_1-c_2 \vert} \geq 1 - 8\varepsilon^{1-\delta_2},
\end{align}
assuming that $\varepsilon/\vert c_1-c_2 \vert < 1$, which is the case for $\varepsilon$ small enough. As a consequence, since $x(t-\tau) - x(t-t_2) = (t_2-\tau)v'' $, if the direction of $x(t-\tau) - x(t-t_2)$ is different enough from the direction of $c_1-c_2$, in the sense that if:
\begin{align}
1 - \Big\vert \frac{v''}{\vert v'' \vert} \cdot \frac{c_2-c_1}{\vert c_2 - c_1 \vert} \Big\vert \leq \varepsilon^{\delta_3} \hspace{5mm} \text{and} \hspace{5mm} 8 \varepsilon^{1-\delta_2} \leq \varepsilon^{\delta_3},
\end{align}
then no recollision can take place. As a consequence, choosing:
\begin{align}
\delta_3 < 1 - \delta_2
\end{align}
so that $8 \varepsilon^{1-\delta_2} \leq \varepsilon^{\delta_3}$ for any $\varepsilon$ small enough (that is, smaller than a certain $\varepsilon_2 = \varepsilon_2(\delta_2,\delta_3)$ that depends only on $\delta_2$ and $\delta_3$), we see that if $(t_2,\omega_2) \in \big( \widetilde{\mathcal{P}}_{c_2}^1 \big)^c \cap \big( \widetilde{\mathcal{P}}_{c_2}^4 \big)^c$, then no recollision can take place on the time interval $[t-t_2,t-t_3]$.\\
\newline
We can therefore proceed to the change of variables, and then introduce the same pathological sets $\widetilde{\mathcal{P}}_{c_3}^i$ for the parameters $(t_3,\omega_3)$, $1 \leq i \leq 5$, as introduced for the second obstacle. Proceeding recursively, we can decompose after $k$ steps the term $\varphi_1(k)$ as follows:
\begin{align}
\varphi_{\varepsilon}^{(2)}(k) = \varphi_{\varepsilon,k}^{(2)}(k) + \sum_{j=1}^k \varphi_{\varepsilon,c_j}^{(2),\text{patho.}}(k),
\end{align}
\noindent
with:
\begin{align}
\varphi_{\varepsilon,k}^{(2)}(k) &= \mu_\varepsilon^k \int_{t_1 = 0}^t \int_{\mathbb{S}^{d-1}_{\omega_1}} \mathds{1}_{\widetilde{\mathcal{P}}_{c_1}^c} \int_{t_2=0}^{t_1} \int_{\mathbb{S}^{d-1}_{\omega_2}} \mathds{1}_{\widetilde{\mathcal{P}}_{c_2}^c} \dots \int_{t_k=0}^{t_{k-1}} \int_{\mathbb{S}_{\omega_k}^{d-1}} \mathds{1}_{\widetilde{\mathcal{P}}_{c_k}^c} e^{- \mu_\varepsilon \vert \overline{B(x,\varepsilon)} \, \cup \, \mathcal{T}^t(c^k) \vert} \prod_{l=1}^k \vert v^{(l-1)} \cdot \omega_l \vert \nonumber\\
&\hspace{20mm} \times \varphi\big( x + \sum_{m=1}^k \big( t_{m-1}-t_m \big) v^{(m-1)} + t_k v^{(k)},v^{(k)} \big) \varepsilon^{k(d-1)} \dd \omega_k \dd t_k \dd \omega_2 \dd t_2 \dd \omega_1 \dd t_1,
\end{align}
\begin{align}
\varphi_{\varepsilon,c_j}^{(2),\text{patho.}}(k) &= \mu_\varepsilon^k \int_{t_1 = 0}^t \int_{\mathbb{S}^{d-1}_{\omega_1}} \mathds{1}_{\widetilde{\mathcal{P}}_{c_1}^c} \dots \int_{t_{j-1} = 0}^{t_{j-2}} \int_{\mathbb{S}^{d-1}_{\omega_{j-1}}} \mathds{1}_{\widetilde{\mathcal{P}}_{c_{j-1}}^c} \int_{t_j=0}^{t_{j-1}} \int_{\mathbb{S}^{d-1}_{\omega_j}} \mathds{1}_{\widetilde{\mathcal{P}}_{c_j}} \int_{c_{j+1}} \dots \int_{c_k} e^{- \mu_\varepsilon \vert \overline{B(x,\varepsilon)} \, \cup \, \mathcal{T}^t(c^k) \vert} \nonumber\\
&\hspace{12mm} \times \varphi\big( T_{c,\varepsilon}^t(x,v) \big) \mathds{1}_{X_0}(c^k) \mathds{1}_{\mathcal{P}_\text{patho.}^c}(c^k) \mathds{1}_{\{ c_i \in B_\varepsilon \forall\, j+1 \leq i \leq k \}} \mathds{1}_{\substack{\text{the obstacles}\\ c_l,\, j+1 \leq l \leq k,\\ \text{are internal}}}(c^k) \mathds{1}_{\substack{\text{the $m$-th collided} \\ \text{obstacle is }c_m \\ \forall\, j+1 \leq m \leq k}}(c^k) \nonumber\\
&\hspace{45mm} \times \prod_{l=1}^j \vert v^{(l-1)} \cdot \omega_l \vert  \varepsilon^{j(d-1)}\dd c_k \dots \dd c_{j+1} \dd \omega_j \dd t_j \dots \dd \omega_1 \dd t_1
\end{align}
and finally
\begin{align}
\widetilde{\mathcal{P}}_{c_j} = \widetilde{\mathcal{P}}_{c_j}^1 \cup \widetilde{\mathcal{P}}_{c_j}^2 \cup \widetilde{\mathcal{P}}_{c_j}^3 \cup \widetilde{\mathcal{P}}_{c_j}^4 \cup \widetilde{\mathcal{P}}_{c_j}^5, \hspace{5mm} \widetilde{\mathcal{P}}_{c_j}^c = \mathcal{A}_{c_j} \backslash \widetilde{\mathcal{P}}_{c_j},
\end{align}
with
\begin{align}
\label{EQUATEnsmbPatho_P_cj__1__}
\widetilde{\mathcal{P}}_{c_j}^1 = \Big\{ (t_j,\omega_j) \in \mathcal{A}_{c_j} \ /\ t_{j-1}-t_j < \max\Big( \frac{\varepsilon^{\delta_2}}{r^{j-1} \vert v \vert },\varepsilon \Big) \Big\},
\end{align}
\begin{align}
\label{EQUATEnsmbPatho_P_cj__2__}
\widetilde{\mathcal{P}}_{c_j}^2 &= \Big\{ (t_j,\omega_j) \in \mathcal{A}_{c_j} \ /\ x+ \sum_{m=1}^j \big(t_{m-1}-t_m\big) v^{(m-1)} - \varepsilon \omega_j \in \overline{B(x,\varepsilon)} \Big\},
\end{align}
\begin{align}
\label{EQUATEnsmbPatho_P_cj__3__}
\widetilde{\mathcal{P}}_{c_j}^3 &= \Big\{ (t_j,\omega_j) \in \mathcal{A}_{c_j} \ /\ \nonumber\\
&\hspace{15mm} x+ \sum_{m=1}^j \big(t_{m-1}-t_m\big) v^{(m-1)} - \varepsilon \omega_j \in [x(t-t_{n-1}),x(t-t_n)] + \overline{B(0,\varepsilon)} \hspace{2mm} \text{for some } n < j \Big\},
\end{align}
\begin{align}
\label{EQUATEnsmbPatho_P_cj__4__}
\widetilde{\mathcal{P}}_{c_j}^4 = \Big\{ (t_j,\omega_j) \in \mathcal{A}_{c_j} \ /\ 1 - \Big\vert \frac{v^{(j)}}{\vert v^{(j)} \vert} \cdot \frac{c_j-c_n}{\vert c_j-c_n \vert} \Big\vert \leq \varepsilon^{\delta_3} \hspace{2mm} \text{for some } n < j \Big\},
\end{align}
\begin{align}
\label{EQUATEnsmbPatho_P_cj__5__}
\widetilde{\mathcal{P}}_{c_j}^5 = \Big\{ (t_j,\omega_j) \in \mathcal{A}_{c_j} \ /\ 1 - \Big\vert \frac{v^{(j)}}{\vert v^{(j)} \vert} \cdot \frac{v^{(n)}}{\vert v^{(n)} \vert} \Big\vert \leq \varepsilon^{\delta_1} \hspace{2mm} \text{for some } n < j  \Big\}.
\end{align}

\paragraph{Estimating the remainder terms.} We turn now to the estimates of the pathological terms $\varphi_{\varepsilon,c_j}^{(2),\text{patho.}}(k)$. To do so, we will decompose:
\begin{align}
\varphi_{\varepsilon,c_j}^{(2),\text{patho.}}(k) = \sum_{l=1}^5 \varphi_{\varepsilon,c_j,l}^{(2),\text{patho.}}(k),
\end{align}
where $\varphi_{\varepsilon,c_j,l}^{(2),\text{patho.}}(k)$ is defined by replacing the pathological set $\widetilde{\mathcal{P}}_{c_j}$ in the expression of $\varphi_{\varepsilon,c_j}^{(2),\text{patho.}}(k)$ by the subset $\widetilde{\mathcal{P}}_{c_j}^l$, defined in \eqref{EQUATEnsmbPatho_P_cj__1__}-\eqref{EQUATEnsmbPatho_P_cj__5__}, that is:
\begin{align}
\varphi_{\varepsilon,c_j,l}^{(2),\text{patho.}}(k) &= \mu_\varepsilon^k \int_{t_1 = 0}^t \int_{\mathbb{S}^{d-1}_{\omega_1}} \mathds{1}_{\widetilde{\mathcal{P}}_{c_1}^c} \dots \int_{t_{j-1} = 0}^{t_{j-2}} \int_{\mathbb{S}^{d-1}_{\omega_{j-1}}} \mathds{1}_{\widetilde{\mathcal{P}}_{c_{j-1}}^c} \int_{t_j=0}^{t_{j-1}} \int_{\mathbb{S}^{d-1}_{\omega_j}} \mathds{1}_{\widetilde{\mathcal{P}}_{c_j}}^l \int_{c_{j+1}} \hspace{-2mm}\dots \int_{c_k} e^{- \mu_\varepsilon \vert \overline{B(x,\varepsilon)} \, \cup \, \mathcal{T}^t(c^k) \vert} \nonumber\\
&\hspace{10mm} \times \varphi\big( T_{c,\varepsilon}^t(x,v) \big) \mathds{1}_{X_0}(c^k) \mathds{1}_{\mathcal{P}_\text{patho.}^c}(c^k) \mathds{1}_{\{ c_i \in B_\varepsilon \forall\, j+1 \leq i \leq k \}} \mathds{1}_{\substack{\text{the obstacles}\\ c_l,\, j+1 \leq l \leq k,\\ \text{are internal}}}(c^k) \mathds{1}_{\substack{\text{the $m$-th collided} \\ \text{obstacle is }c_m \\ \forall\, j+1 \leq m \leq k}}(c^k) \nonumber\\
&\hspace{50mm} \times \prod_{l=1}^j \vert v^{(l-1)} \cdot \omega_l \vert  \varepsilon^{j(d-1)}\dd c_k \dots \dd c_{j+1} \dd \omega_j \dd t_j \dots \dd \omega_1 \dd t_1
\end{align}
for any $1 \leq l \leq 5$. We will estimate the terms $\varphi_{\varepsilon,c_j,l}^{(2),\text{patho.}}(k)$ separately. As for $\varphi_{\varepsilon,c_j,1}^{(2),\text{patho.}}(k)$ (that is, for $l=1$), we will use the following inequality on the subset $S_{j,[a,b]}$ of the $k$-simplex, with $j \leq k$ and $0 \leq a < b \leq t_{j-1}$, defined as:
\begin{align}
S_{j,[a,b]} = \int_0^t\dots \int_{t_{j-1}=0}^{t_{j-2}}\int_{t_j \in [a,b]} \dd t_j \dots \dd t_1.
\end{align}
We have:
\begin{align}
S_{j,[a,b]} &= \int_0^t\dots\int_{t_{j-1}=0}^{t_{j-2}} \big[ b - a \big] \dd t_{j-1} \dots t_1 = \frac{t^{j-1}}{(j-1)!} \big[ b - a \big].
\end{align}

\noindent
In addition, by definition for $j+1 \leq l \leq k$ the obstacle $c_l$ belongs to the dynamical tube obtained for $t \in [t-t_{l-1},t-t_l]$. By construction the configuration of scatterers $c^k$ is chosen such that it belongs to $\mathcal{P}_\text{patho.}^c$, and the obstacles are ordered. In particular, the hard sphere flow $T_{c,\varepsilon}^s$ is well-defined for all time $s \in [0,t]$, and the trajectory presents only a finite number of collisions with the scatterers. So, there exists a sequence of times $0 < t_k < \dots < t_{j+1} < t_j$ such that the time $t-t_l$ corresponds to the first time of collision between the tagged particle and the obstacle $c_l$.\\
Besides, assuming that the position of all the obstacles are known, except the last obstacle $c_k$, we deduce that $c_k$ belongs to a portion of the dynamical tube of length equal to $t_{k-1} \vert v^{(k-1)} \vert$ (where $t-t_{k-1}$ is the time of the first collision with the penultimate obstacle $c_{k-1}$), corresponding to the portion of the dynamical tube obtained in the case when only the first $k-1$ obstacles exist and are internal, for the time interval $[t-t_{k-1},t]$. On this time interval, only a finite number of collisions take place. As a consequence, by recursion, and relying on Lemma \ref{LEMMEEstimMesurTube_Dynam}, we can estimate the measure of the following portion of the dynamical tube $\mathcal{T}^t(c^k)$ as:
\begin{align}
\Big\vert \big\{& y \in \mathbb{R}^d \ /\ \exists\, s\in [t-t_{k-1},t], z \in \partial B(0,\varepsilon) \ \text{such that}\ y = \big(T_{c,\varepsilon}^s(x,v)\big)_x + z \big\} \Big\vert \nonumber\\
&\leq \Big\vert \big[ x(t-t_{k-1}),x(t-t_{k-1}) + t_{k-1} \vert v^{(k-1)} \vert \big] + \overline{B(0,\varepsilon)} \Big\vert,
\end{align}
with
\begin{align}
x(t-t_{k-1}) = \big(T_{c,\varepsilon}^{t-t_{k-1}}(x,v)\big)_x \hspace{5mm} \text{and} \hspace{5mm} v^{(k-1)} = \big(T_{c,\varepsilon}^{t-t_{k-1}}(x,v)\big)_v,
\end{align}
where the subset of which the measure is taken as an upper bound corresponds to the rectified trajectory $T_{c,\varepsilon}^s(x,v)$ plus the closed ball $\overline{B(0,\varepsilon)}$, in the case when there is no collision on the time interval $]t-t_{l-1},t-t_l[$. Since in addition the last obstacle $c_k$ cannot belong to the ball
\begin{align}
\overline{B\big(\big(T_{c,\varepsilon}^{t-t_{k-1}}(x,v)\big)_x,\varepsilon\big)},
\end{align}
we deduce that $c_k$ belongs to the following subset $C_k \subset \mathbb{R}^d$:
\begin{align}
C_k = \Big\{& y \in \mathbb{R}^d \ /\ \exists\, s\in [t-t_{l-1},t-t_l], z \in \partial B(0,\varepsilon) \ \text{such that}\ y = \big(T_{c,\varepsilon}^s(x,v)\big)_x + z \Big\} \, \backslash \, \overline{B\big(\big(T_{c,\varepsilon}^{t-t_{k-1}}(x,v)\big)_x,\varepsilon\big)}
\end{align}
and the volume of $C_k$ can be estimated as follows:
\begin{align}
\vert C_k \vert \leq C(d-1) \varepsilon^{d-1}t_{k-1} \vert v^{(k-1)} \vert, 
\end{align}
where $C(d-1)$ is the constant that appears in the formula of the volume of the $(d-1)$-dimensional ball: $\vert B_{\mathbb{R}^{d-1}}(0,\varepsilon) \vert = C(d-1) \varepsilon^{d-1}$.
In the end, we can estimate the following integral as:
\begin{align}
&\int_{c_{j+1}}\dots\int_{c_k} \mathds{1}_{X_0}(c^k) \mathds{1}_{\mathcal{P}_\text{patho.}^c}(c^k) \mathds{1}_{\{ c_i \in B_\varepsilon \forall\, j+1 \leq i \leq k \}} \mathds{1}_{\substack{\text{the obstacles}\\ c_l,\, j+1 \leq l \leq k,\\ \text{are internal}}}(c^k) \mathds{1}_{\substack{\text{the $m$-th collided} \\ \text{obstacle is }c_m \\ \forall\, j+1 \leq m \leq k}}(c^k) \dd c_k \dots \dd c_{j+1} \nonumber\\
&\hspace{4mm}\leq \int_{c_{j+1}}\dots\int_{c_{k-1}} \mathds{1}_{X_0}(c^{k-1}) \mathds{1}_{\mathcal{P}_\text{patho.}^c}(c^{k-1}) \mathds{1}_{\{ c_i \in B_\varepsilon \forall\, j+1 \leq i \leq k-1 \}} \mathds{1}_{\substack{\text{the obstacles}\\ c_l,\, j+1 \leq l \leq k-1,\\ \text{are internal}}}(c^{k-1}) \mathds{1}_{\substack{\text{the $m$-th collided} \\ \text{obstacle is }c_m \\ \forall\, j+1 \leq m \leq k-1}}(c^{k-1}) \nonumber\\
&\hspace{75mm} \times C(d-1) \varepsilon^{d-1} t_{k-1}(c^{k-1}) \vert v^{(k-1)} \vert \dd c_{k-1} \dots \dd c_{j+1},
\end{align}
relying on the result of Lemma \ref{LEMMEEstimMesurTube_Dynam}. Then, for any integer $n$, we decompose the integral on $c_{k-1}$ as:
\begin{align}
\int_{c_{k-1}} t_{k-1}(c^{k-1}) \dd c_{k-1} = \int_{c_{k-1}} \sum_{l=1}^n \mathds{1}_{t_{k-1} \in [\frac{(l-1)}{n}t_{k-2},\frac{l}{n}t_{k-2}]} t_{k-1}(c^{k-1}) \dd c_{k-1},
\end{align}
which provides the upper bound:
\begin{align}
\int_{c_{k-1}} t_{k-1}(c^{k-1}) \dd c_{k-1} &\leq \sum_{l=1}^n \frac{l}{n} t_{k-2}(c^{k-2}) C(d-1) \varepsilon^{d-1} \frac{t_{k-2}(c^{k-2})}{n} \vert v^{(k-2)} \vert \nonumber\\
&\leq C(d-1) \varepsilon^{d-1} \vert v \vert t_{k-2}^2(c^{k-2}) \frac{n(n+1)}{2n^2} \cdotp
\end{align}
In the limit $n \rightarrow +\infty$, we find in particular:
\begin{align}
\int_{c_{k-1}} t_{k-1}(c^{k-1}) \dd c_{k-1} \leq  C(d-1) \varepsilon^{d-1} \vert v \vert \frac{t_{k-2}^2(c^{k-2})}{2}\cdotp
\end{align}
For the next step, in the general case, the same decomposition of the integral provides a series of the form:
\begin{align}
\frac{t_{k-m}^m}{n^m} \sum_{l=1}^n l^{m-1} = t_{k-m}^m \frac{1}{n} \sum_{l=1}^n \Big( \frac{l}{n} \Big)^{m-1},
\end{align}
which is a Riemann sum, and converges as $n$ goes to infinity towards $\frac{ t_{k-m}^m}{m}$. As a consequence, we find for the pathological term $\varphi_{\varepsilon,c_j,l}^{(1),\text{patho.}}(k)$ obtained in the case $l=1$:
\begin{align}
\label{EQUATEstimErreuPhi_1}
\big\vert \varphi_{\varepsilon,c_j,1}^{(2),\text{patho.}}(k) \big\vert &\leq \mu_\varepsilon^k \vertii{\varphi}_\infty \int_{t_1 = 0}^t \int_{\mathbb{S}^{d-1}_{\omega_1}}  \dots \int_{t_{j-1} = 0}^{t_{j-2}} \int_{\mathbb{S}^{d-1}_{\omega_{j-1}}} \int_{t_j=0}^{t_{j-1}} \int_{\mathbb{S}^{d-1}_{\omega_j}} \mathds{1}_{\widetilde{\mathcal{P}}_{c_j}^1} \int_{c_{j+1}} \dots \int_{c_k} e^{- \mu_\varepsilon \vert \overline{B(x,\varepsilon)} \, \cup \, \mathcal{T}^t(c^k) \vert} \nonumber\\
&\hspace{25mm} \times \mathds{1}_{X_0}(c^k) \mathds{1}_{\mathcal{P}_\text{patho.}^c}(c^k) \mathds{1}_{\{ c_i \in B_\varepsilon \forall\, j+1 \leq i \leq k \}} \mathds{1}_{\substack{\text{the obstacles}\\ c_l,\, j+1 \leq l \leq k,\\ \text{are internal}}}(c^k) \mathds{1}_{\substack{\text{the $m$-th collided} \\ \text{obstacle is }c_m \\ \forall\, j+1 \leq m \leq k}}(c^k) \nonumber\\
&\hspace{35mm} \times \prod_{l=1}^j \vert v^{(l-1)} \cdot \omega_l \vert  \varepsilon^{j(d-1)}\dd c_k \dots \dd c_{j+1} \dd \omega_j \dd t_j \dots \dd \omega_1 \dd t_1 \nonumber\\
&\leq \mu_\varepsilon^k \vertii{\varphi}_\infty \vert \mathbb{S}^{d-1} \vert^j \vert v \vert^j \varepsilon^{j(d-1)} \frac{\big[ C(d-1) \varepsilon^{d-1}t \vert v \vert \big]^{k-j}}{(k-j)!} \frac{t^{j-1}}{(j-1)!} \Big(\frac{\varepsilon^{\delta_2}}{r^{j-1} \vert v \vert} + \varepsilon \Big) \nonumber\\
&\leq \frac{ \big( \vert \mathbb{S}^{d-1} \vert t \vert v \vert r^{-1} \big)^{j-1}}{(j-1)!} \frac{\big[ C(d-1) t \vert v \vert \big]^{k-j}}{(k-j)!} \vert \mathbb{S}^{d-1} \vert \vertii{\varphi}_\infty \big( \varepsilon^{\delta_2} + \vert v \vert \varepsilon \big),
\end{align}
using in particular that in the Boltzmann-Grad limit we have $\mu_\varepsilon^k \varepsilon^{j(d-1)}\varepsilon^{(k-j)(d-1)} = 1$.\\
In the case of $\varphi_{\varepsilon,c_j,2}^{(2),\text{patho.}}(k)$ ($l = 2$) we use a similar argument. For $\omega_j$ fixed, the set of $t_j$ such that $(t_j,\omega_j) \in \widetilde{\mathcal{P}}_{c_j}^2$ is contained in a certain interval , with $I_\text{path.}$ that depends on the $x$, $t_m$ and $v^{(m-1)}$ for $1 \leq m \leq j-1$, and such that:
\begin{align}
\big\vert I_\text{path.} \big\vert \leq \frac{2}{\vert v^{(j)} \vert}\varepsilon \leq \frac{2}{r^j \vert v \vert} \varepsilon.
\end{align}
We find therefore:
\begin{align}
\label{EQUATEstimErreuPhi_2}
\big\vert \varphi_{\varepsilon,c_j,2}^{(2),\text{patho.}}(k) \big\vert &\leq \frac{ \big( \vert \mathbb{S}^{d-1} \vert t \vert v \vert r^{-1} \big)^{j-1}}{(j-1)!} \frac{\big[ C(d-1) t \vert v \vert \big]^{k-j}}{(k-j)!}\frac{\vert \mathbb{S}^{d-1} \vert \vertii{\varphi_0}_\infty}{r} \varepsilon.
\end{align}
As for $\varphi_{\varepsilon,c_j,3}^{(2),\text{patho.}}(k)$ ($l = 3$), we use the fact that $c^k \in \widetilde{\mathcal{P}}_{c_{j-1}}^c$, so that in particular $c^k \in \widetilde{\mathcal{P}}_{c_{j-1}}^5$, and so the direction of $v^{(j-1)}$ is far enough from the direction of $v^{(n)}$, that is, the direction of $x(t-t_{n-1})-x(t-t_n)$, for any $n < j-1$. Therefore, for any index $n < j-1$, only a small subset of time parameters $t_j$ are such that $T_{c,\varepsilon}^s(x,v)$ belongs to the cylinder $[x(t-t_{n-1}),x(t-t_n)] + \overline{B(0,\varepsilon)}$. More precisely, for $\theta$ the angle between $v^{(j-1)}$ and $v^{(n)}$, we have that $\vert \cos\theta \vert$ is bounded from above by $1 - \varepsilon^{\delta_1}$, so that the maximal length of a trajectory contained in the cylinder is bounded from above by:
\begin{align}
2 \frac{\varepsilon}{\vert v^{(j-1)} \vert \vert \sin\theta \vert} &\leq 2 \frac{\varepsilon}{r^{j-1} \vert v \vert \sqrt{1 - \big(1 - \varepsilon^{\delta_1}\big)^2}} \leq 2 \frac{\varepsilon}{r^{j-1} \vert v \vert\sqrt{\varepsilon^{\delta_1}}} \leq \frac{2}{r^{j-1} \vert v \vert} \varepsilon^{1-\delta_1/2},
\end{align}
for any $\varepsilon$ smaller than a certain $\varepsilon_3 = \varepsilon_3(\delta_1)$ that depends only on the positive constant $\delta_1$. We find then:
\begin{align}
\label{EQUATEstimErreuPhi_3}
\big\vert \varphi_{\varepsilon,c_j,3}^{(2),\text{patho.}}(k) \big\vert &\leq \mu_\varepsilon^k \vertii{\varphi}_\infty \vert \mathbb{S}^{d-1} \vert^j \vert v \vert^j \varepsilon^{j(d-1)} \frac{\big[ C(d-1) \varepsilon^{d-1}t \vert v \vert \big]^{k-j}}{(k-j)!} \frac{t^{j-1}}{(j-1)!} \frac{2(j-2)\varepsilon^{1-\delta_1/2}}{r^{j-1} \vert v \vert} \nonumber\\
&\leq 2\frac{ \big( \vert \mathbb{S}^{d-1} \vert t \vert v \vert r^{-1} \big)^{j-2}}{(j-2)!} \frac{\big[ C(d-1) t \vert v \vert \big]^{k-j}}{(k-j)!} t \vert v \vert r^{-1} \vert \mathbb{S}^{d-1} \vert^2 \vertii{\varphi}_\infty \varepsilon^{1-\delta_1/2}.
\end{align}
Finally, we will rely on Lemma \ref{LEMMEColinearitScatt} to estimate the size of the pathological terms $\varphi_{\varepsilon,c_j,4}^{(2),\text{patho.}}(k)$ and $\varphi_{\varepsilon,c_j,5}^{(2),\text{patho.}}(k)$. Each of the sets $\widetilde{\mathcal{P}}_{c_j}^4$ and $\widetilde{\mathcal{P}}_{c_j}^5$ corresponds to a set of $j-1$ pathological directions, that are obtained by taking the angular parameter $\omega_j$ in $j-1$ pathological subsets, of respective measure estimated by $C(d,r)\varepsilon^{\delta_3/2}$ and $C(d,r)\varepsilon^{\delta_1/2}$. In the end, we find in the case $l = 4$:
\begin{align}
\label{EQUATEstimErreuPhi_4}
\big\vert \varphi_{\varepsilon,c_j,4}^{(2),\text{patho.}}(k) \big\vert &\leq \mu_\varepsilon^k \vertii{\varphi}_\infty \vert v \vert^j \varepsilon^{j(d-1)} \vert \mathbb{S}^{d-1} \vert^{j-1}   \frac{\big[ C(d-1) \varepsilon^{d-1}t \vert v \vert \big]^{k-j}}{(k-j)!} \frac{t^j}{j!} (j-1) C(d,r) \varepsilon^{\delta_3/2}\nonumber\\
&\leq C(d,r) \frac{ \big( \vert \mathbb{S}^{d-1} \vert t \vert v \vert \big)^{j-1}}{(j-1)!} \frac{\big[ C(d-1) t \vert v \vert \big]^{k-j}}{(k-j)!} t \vert v \vert \vertii{\varphi}_\infty \varepsilon^{\delta_3/2},
\end{align}
and similarly, for $\varphi_{\varepsilon,c_j,5}^{(2),\text{patho.}}(k)$ ($l = 5$):
\begin{align}
\label{EQUATEstimErreuPhi_5}
\big\vert \varphi_{\varepsilon,c_j,5}^{(2),\text{patho.}}(k) \big\vert &\leq C(d,r) \frac{ \big( \vert \mathbb{S}^{d-1} \vert t \vert v \vert \big)^{j-1}}{(j-1)!} \frac{\big[ C(d-1) t \vert v \vert \big]^{k-j}}{(k-j)!} t \vert v \vert \vertii{\varphi}_\infty \varepsilon^{\delta_1/2}.
\end{align}
Gathering the estimates, we denote by $R_2$ the collection of the remainder terms:
\begin{align}
R_2 = \sum_{k=1}^{+\infty} \sum_{j=1}^k \sum_{l=1}^5 \varphi_{\varepsilon,c_j,l}^{(2),\text{patho.}}(k),
\end{align}
so that:
\begin{align}
\varphi_\varepsilon^{(1)} = \sum_{k=1}^{+\infty} \varphi_{\varepsilon,k}^{(2)}(k) + \sum_{k=1}^{+\infty} \sum_{j=1}^k \sum_{l=1}^5 \varphi_{\varepsilon,c_j,l}^{(2),\text{patho.}}(k) = \sum_{k=1}^{+\infty} \varphi_{\varepsilon,k}^{(2)}(k) + R_2.
\end{align}
Denoting by:
\begin{align}
C_1 = \vert \mathbb{S}^{d-1} \vert t \vert v \vert r^{-1} \hspace{5mm} \text{and} \hspace{5mm} C_2 = C(d-1) t \vert v \vert
\end{align}
and relying on \eqref{EQUATEstimErreuPhi_1}-\eqref{EQUATEstimErreuPhi_5}, we find:
\begin{align}
&\big\vert R_2 \big\vert \leq \vertii{\varphi}_\infty \sum_{k=1}^{+\infty} \sum_{j=1}^k \frac{C_1^{j-1}}{(j-1)!} \frac{C_2^{k-j}}{(k-j)!} \nonumber\\
&\hspace{10mm} \times \Big[ \vert \mathbb{S}^{d-1} \vert \big( \varepsilon^{\delta_2} + \vert v \vert \varepsilon \big) + \frac{\vert \mathbb{S}^{d-1} \vert}{r} \varepsilon + 2 t \vert v \vert \frac{\vert \mathbb{S}^{d-1} \vert^2}{r} \varepsilon^{1-\delta_1/2} + C(d,r) t \vert v \vert \varepsilon^{\delta_3/2} + C(d,r) t \vert v \vert \varepsilon^{\delta_1/2} \Big].
\end{align}
Inverting the sums:
\begin{align}
\sum_{k=1}^{+\infty} \sum_{j=1}^k \frac{C_1^{j-1}}{(j-1)!} \frac{C_2^{k-j}}{(k-j)!} = \sum_{j=1}^{+\infty} \sum_{k=j}^{+\infty} \frac{C_1^{j-1}}{(j-1)!} \frac{C_2^{k-j}}{(k-j)!} = e^{C_1+C_2}.
\end{align}
In the end, choosing:
\begin{align}
\delta_1 = \frac{1}{2}, \hspace{5mm} \delta_2 = \frac{1}{4} \hspace{5mm} \text{and} \hspace{5mm} \delta_3 = \frac{1}{2},
\end{align}
we have obtained that there exist two positive constants $\widetilde{C}_1 = \widetilde{C}_1(d,r)$ and $\widetilde{C}_2 = \widetilde{C}_1=2(d,r)$ that depend only on the dimension $d$ and the restitution coefficient $r$ such that:
\begin{align}
\big\vert R_2 \big\vert \leq \vertii{\varphi}_\infty \widetilde{C}_1 \big(1 + \max(1,t) \vert v \vert\big) e^{\widetilde{C}_2 t \vert v \vert} \varepsilon^{1/4}.
\end{align}
The proof of Proposition \ref{PROPOReecrPhiEpElimiRecol} is complete.
\end{proof}

\noindent
Proposition \ref{PROPOReecrPhiEpElimiRecol} completed the important step which consists in comparing $\varphi_\varepsilon$ with the solution $\widetilde{\varphi}$ of the adjoint equation \eqref{EQUATLineaBoltzSphDuFormeAdjoi}. To do so, we established that for most of the configurations $c$ of the scatterers, the flow $T_{c,\varepsilon}^t(x,v)$ of the tagged particle corresponds to $\big( x + \sum_{m=1}^k \big( t_{m-1}-t_m \big) v^{(m-1)} + t_k v^{(k)},v^{(k)} \big)$, the position and velocity at which the intial datum $\varphi_0$ is evaluated in the series representation \eqref{EQUATRepreSerieEquatAdjoi} of the solution $\widetilde{\varphi}$.

\begin{propo}[Final comparison between $\varphi_\varepsilon$ and $\varphi$]
\label{PROPOCompaFinal_Phi_PhiEp}
Let $\varepsilon > 0$. Let $c$ be a $d$-dimensional Poisson process of intensity $\mu_\varepsilon > 0$ such that the Boltzmann-Grad scaling \eqref{EQUATBoltzmann-Grad_Limit} holds. Then, there exists a universal constant $\varepsilon_0 > 0$ and two positive constants $\widetilde{C}_3 = \widetilde{C}_3(d,r)$ and $\widetilde{C}_4 = \widetilde{C}_4(d,r)$ that depend only on the dimension $d$ and the restitution coefficient $r$ such that, for any $0 < \varepsilon \leq \varepsilon_0$, and any function $\varphi \in \mathcal{C}_0(\mathbb{R}^d \times \mathbb{R}^d)$, the quantity $\varphi_\varepsilon$, defined as \eqref{EQUATDefinPhiEp_Avec_FlotDefin}, and the series $\psi$ defined in \eqref{EQUATRepreSerieEquatAdjoi} with initial datum $\varphi$ satisfy:
\begin{align}
\varphi_\varepsilon - \psi = R_1 + R_2 + R_3,
\end{align}
where $R_1$ and $R_2$ satisfy respectively \eqref{EQUATDefinReste_R_1_} and \eqref{EQUATDefinReste_R_2_PropoElimiRecol}, and $R_3$ is such that:
\begin{align}
\label{EQUATDefinReste_R_3_}
\big\vert R_3 \big\vert \leq \vertii{\varphi}_\infty \widetilde{C}_3 \max\big(1,t \vert v \vert^2, t^2 \vert v \vert^2 \big) e^{\widetilde{C}_4 t \vert v \vert} \varepsilon^{1/4}.
\end{align}
\end{propo}

\begin{proof}
We observe that the series representation \eqref{EQUATRepreSerieEquatAdjoi} of $\widetilde{\varphi}$ is well defined assuming only that $\varphi$ is continuous and vanishing at infinity. To compare $\sum_{k=1}^{+\infty} \varphi_{\varepsilon,k}^{(2)}(k)$ with the series representation, we proceed in two steps: firstly, we determine lower and upper bounds for the measure of the dynamical tube, and secondly, we will estimate the error term coming from the truncations $\mathds{1}_{\widetilde{\mathcal{P}}_{c_j}}$ in the domain of integration.\\
As first step, we consider the collisionless tube, whose measure is given by:
\begin{align}
\mu_\varepsilon \big\vert [x,x+tv] + \overline{B(0,\varepsilon)} \big\vert = \mu_\varepsilon \big[ \vert \overline{B(0,\varepsilon)} \vert + C(d-1) \varepsilon^{d-1} t \vert v \vert \big],
\end{align}
so that
\begin{align}
\Big\vert \mu_\varepsilon \big\vert [x,x+tv] + \overline{B(0,\varepsilon)} \big\vert - C(d-1) t \vert v \vert \Big\vert \leq \mu_\varepsilon \vert \overline{B(0,\varepsilon)} \vert = C(d) \varepsilon.
\end{align}
Besides, considering the dynamical tube $\mathcal{T}^t(c^k)$ (that is, in the case when collisions can occur) and applying Lemma \ref{LEMMEEstimMesurTube_Dynam}, we find:
\begin{align}
\mu_\varepsilon \big\vert \mathcal{T}^t(c^k) \big\vert &\leq \mu_\varepsilon \vert \overline{B(0,\varepsilon)} \vert + \mu_\varepsilon C(d-1) \varepsilon^{d-1} \Big[ \sum_{m=1}^k (t_{m-1}-t_m) \vert v^{(m-1)} \vert + t_k \vert v^{(k)} \vert \Big] \nonumber\\
&\leq C(d) \varepsilon + C(d-1) \Big[ \sum_{m=1}^k (t_{m-1}-t_m) \vert v^{(m-1)} \vert + t_k \vert v^{(k)} \vert \Big].
\end{align}
To obtain a lower bound on the previous quantity, we have $t_m-t_{m+1} \geq \varepsilon$ for any $1 \leq m \leq k-1$ thanks to the cut-off $\widetilde{\mathcal{P}}_{c_j}^1$. Therefore:
\begin{align}
\mu_\varepsilon \big\vert \mathcal{T}^t(c^k) \big\vert &\geq \mu_\varepsilon C(d-1) \varepsilon^{d-1} \Big[ (t-t_1) \vert v \vert + (t_1-t_2-2\varepsilon) \vert v' \vert + \dots + (t_{k-1}-t_k-2\varepsilon)\vert v^{(k-1)} \vert + t_k \vert v^{(k)} \vert \Big]
\end{align}
which implies:
\begin{align}
\Big\vert \big( \mu_\varepsilon \big\vert \overline{B(x,\varepsilon)} \cup \mathcal{T}^t(c^k) \big\vert \big) - \big( \sum_{j=1}^k C_d (t_{j-1}-t_j) \vert v^{(j-1)} \vert + C_d t_k \vert v^{(k)} \vert \big)  \Big\vert \leq \max\big( 2C(d)\varepsilon,2C(d-1)k\vert v \vert\varepsilon \big),
\end{align}
observing that the constants $C_d$ and $C(d-1)$ match. As a consequence, comparing $\varphi_{\varepsilon,k}^{(2)}(k)$ with $\varphi_{\varepsilon,k}^{(3)}(k)$ defined as (the exponential term below the integrals is replaced by the corresponding term of $\psi$):
\begin{align}
\varphi_{\varepsilon,k}^{(3)}(k) &= \int_{t_1 = 0}^t \int_{\mathbb{S}^{d-1}_{\omega_1}} \mathds{1}_{\widetilde{\mathcal{P}}_{c_1}^c} \int_{t_2=0}^{t_1} \int_{\mathbb{S}^{d-1}_{\omega_2}} \mathds{1}_{\widetilde{\mathcal{P}}_{c_2}^c} \dots \int_{t_k=0}^{t_{k-1}} \int_{\mathbb{S}_{\omega_k}^{d-1}} \mathds{1}_{\widetilde{\mathcal{P}}_{c_k}^c} e^{-\sum_{j=1}^k C_d (t_{j-1}-t_j) \vert v^{(j-1)} \vert - C_d t_k \vert v^{(k)} \vert} \nonumber\\
&\hspace{10mm} \times \prod_{l=1}^k \vert v^{(l-1)} \cdot \omega_l \vert \varphi\big( x + \sum_{m=1}^k \big( t_{m-1}-t_m \big) v^{(m-1)} + t_k v^{(k)},v^{(k)} \big) \dd \omega_k \dd t_k \dd \omega_2 \dd t_2 \dd \omega_1 \dd t_1
\end{align}
we have:
\begin{align}
\big\vert \varphi_{\varepsilon,k}^{(2)}(k) - \varphi_{\varepsilon,k}^{(3)}(k) \big\vert &\leq \int_{t_1 = 0}^t \int_{\mathbb{S}^{d-1}_{\omega_1}} \mathds{1}_{\widetilde{\mathcal{P}}_{c_1}^c} \int_{t_2=0}^{t_1} \int_{\mathbb{S}^{d-1}_{\omega_2}} \mathds{1}_{\widetilde{\mathcal{P}}_{c_2}^c} \dots \int_{t_k=0}^{t_{k-1}} \int_{\mathbb{S}_{\omega_k}^{d-1}} 2 \max \big[ C(d),C(d-1) k \vert v \vert \big] \varepsilon \nonumber\\
&\hspace{5mm} \times \prod_{l=1}^k \vert v^{(l-1)} \cdot \omega_l \vert \varphi\big( x + \sum_{m=1}^k \big( t_{m-1}-t_m \big) v^{(m-1)} + t_k v^{(k)},v^{(k)} \big) \dd \omega_k \dd t_k \dd \omega_2 \dd t_2 \dd \omega_1 \dd t_1 \nonumber\\
&\leq 2 \max \big[ C(d),C(d-1) k \vert v \vert \big] \frac{\vert \mathbb{S}^{d-1} \vert^k t^k \vert v \vert^k}{k!} \vertii{ \varphi }_\infty \varepsilon,
\end{align}
so that
\begin{align}
\Big\vert \sum_{k=1}^{+\infty} \varphi_{\varepsilon,k}^{(2)}(k) - \sum_{k=1}^{+\infty} \varphi_{\varepsilon,k}^{(3)}(k) \Big\vert \leq 2 \max \Big[ C(d) e^{\vert \mathbb{S}^{d-1} \vert t \vert v \vert} \vertii{\varphi}_\infty,C(d-1) \vert \mathbb{S}^{d-1} \vert t \vert v \vert^2 e^{\vert \mathbb{S}^{d-1} \vert t \vert v \vert} \Big] \varepsilon.
\end{align}
As for the second step, we compare finally $\sum_{k=1}^{+\infty} \varphi_{\varepsilon,k}^{(3)}(k)$ with the series representation \eqref{EQUATRepreSerieEquatAdjoi} of $\psi$. We have:
\begin{align*}
\big\vert \psi - \sum_{k=1}^{+\infty} \varphi_{\varepsilon,k}^{(3)}(k) \big\vert &= \Big\vert \sum_{k=1}^{+\infty} \sum_{j=1}^k \int_{t_1 = 0}^t \int_{\mathbb{S}^{d-1}_{\omega_1}} \mathds{1}_{\widetilde{\mathcal{P}}_{c_1}^c} \int_{t_2=0}^{t_1} \int_{\mathbb{S}^{d-1}_{\omega_2}} \mathds{1}_{\widetilde{\mathcal{P}}_{c_2}} \dots \int_{t_{j-1}=0}^{t_{j-2}} \int_{\mathbb{S}^{d-1}_{\omega_{j-1}}} \mathds{1}_{\widetilde{\mathcal{P}}_{c_{j-1}}^c} \int_{t_j=0}^{t_{j-1}} \int_{\mathbb{S}^{d-1}_{\omega_j}} \mathds{1}_{\widetilde{\mathcal{P}}_{c_j}} \nonumber\\
&\hspace{10mm} \times\int_{t_{j+1}=0}^{t_j} \int_{\mathbb{S}_{\omega_{j+1}}^{d-1}} \dots \int_{t_k=0}^{t_{k-1}} \int_{\mathbb{S}_{\omega_k}^{d-1}} e^{-\sum_{j=1}^k C_d (t_{j-1}-t_j) \vert v^{(j-1)} \vert - C_d t_k \vert v^{(k)} \vert}  \prod_{l=1}^k \vert v^{(l-1)} \cdot \omega_l \vert \nonumber\\
&\hspace{25mm} \times \varphi\big( x + \sum_{m=1}^k \big( t_{m-1}-t_m \big) v^{(m-1)} + t_k v^{(k)},v^{(k)} \big) \dd \omega_k \dd t_k \dd \omega_2 \dd t_2 \dd \omega_1 \dd t_1 \Big\vert \nonumber
\end{align*}
so that
\begin{align}
\big\vert \psi - \sum_{k=1}^{+\infty} \varphi_{\varepsilon,k}^{(3)}(k) \big\vert &\leq \vertii{\varphi}_\infty \sum_{k=1}^{+\infty} \sum_{j=1}^k \vert \mathbb{S}^{d-1} \vert^{k-1} \frac{t^{j-1}}{(j-1)!}\frac{t^{k-j}}{(k-j)!} \vert v \vert^k \Bigg[ \frac{\varepsilon^{\delta_2}}{r \vert v \vert} \vert \mathbb{S}^{d-1} \vert + \frac{2 \varepsilon}{r^j \vert v \vert} \vert \mathbb{S}^{d-1} \vert \nonumber\\
&\hspace{17.5mm} + \frac{2(j-2) \varepsilon^{1-\delta_1/2}}{r^{j-1} \vert v \vert} \vert \mathbb{S}^{d-1} \vert + t(j-1) C(d,r) \varepsilon^{\delta_3/2} + t(j-1)C(d,r)\varepsilon^{\delta_1/2} \Bigg].
\end{align}
Therefore, choosing as before $\delta_1 = 1/2$, $\delta_2 = 1/4$ and $\delta_3 = 1/2$, there exists a constant $\widetilde{C} = \widetilde{C}(d,r)$ that depends on the dimension $d$ and the restitution coefficient $r$ such that:
\begin{align}
\big\vert \psi - \sum_{k=1}^{+\infty} \varphi_{\varepsilon,k}^{(3)}(k) \big\vert \leq \widetilde{C} \max\big(1,t^2\vert v \vert^2\big) e^{\big(1 + \frac{1}{r}\big) \vert \mathbb{S}^{d-1} \vert t \vert v \vert} \varepsilon^{1/4}.
\end{align}
The proof of Proposition \ref{PROPOCompaFinal_Phi_PhiEp} is now complete.
\end{proof}

\noindent
With the result of Proposition \ref{PROPOCompaFinal_Phi_PhiEp}, we are now in position to prove Theorem \ref{THEORDerivationBoltzmann_InelaLinea}.

\begin{proof}[Proof of Theorem \ref{THEORDerivationBoltzmann_InelaLinea}]
Let $f$ be the unique weak solution to the inelastic linear Boltzmann equation \eqref{EQUATLineaBoltzSphDuFormeForte} in the sense of Definition \ref{DEFINSolutFaibleEquatBoltzLineaInela}, with initial datum $f_0$. We consider a function $\varphi$ of $\mathcal{C}_0(\mathbb{R}^d \times \mathbb{R}^d)$, and a positive real number $\varepsilon > 0$. For any $t_0 > 0$ fixed and for any $\varepsilon > 0$ small enough, we will prove that we have for any $t \in [0,t_0]$:
\begin{align}
\Big\vert \int_{\mathbb{R}^d_x} \hspace{-1mm} \int_{\mathbb{R}^d_v} \varphi(x,v) f_\varepsilon
(t,\dd x, \dd v) - \int_{\mathbb{R}^d_x} \hspace{-1mm} \int_{\mathbb{R}^d_v} \varphi(x,v) f(t,\dd x,\dd v) \Big\vert \leq \varepsilon.
\end{align}
To do so, we prove instead the weak$-*$ convergence of $f_\varepsilon$ towards the weak solution $g$ introduced in Proposition \ref{PROPOExistWeak_Solut}, with initial datum $f_0$. By definition, we have:
\begin{align}
\int_{\mathbb{R}^d_x} \hspace{-1mm} \int_{\mathbb{R}^d_v} \varphi(x,v) f_\varepsilon
(t,\dd x, \dd v) = \int_{\mathbb{R}^d_x} \hspace{-1mm} \int_{\mathbb{R}^d_v} \varphi_\varepsilon(t,x,v) f_0(\dd x, \dd v)
\end{align}
where $\varphi_\varepsilon$ is defined as in \eqref{EQUATDefinPhiEp_Avec_FlotDefin}, and
\begin{align}
\int_{\mathbb{R}^d_x} \hspace{-1mm} \int_{\mathbb{R}^d_v} \varphi(x,v) g(t,\dd x,\dd v) = \int_{\mathbb{R}^d_x} \hspace{-1mm} \int_{\mathbb{R}^d_v} \psi(t,x,v) f_0(\dd x,\dd v),
\end{align}
where $\psi$ is defined by the series \eqref{EQUATRepreSerieEquatAdjoi}, with initial datum $\varphi$. We consider now $\varepsilon > 0$ smaller than the universal constant $\varepsilon_0$ given by Propositions \ref{PROPOReecrPhiEpElimiRecol} and \ref{PROPOCompaFinal_Phi_PhiEp}, so that gathering the results of Propositions \ref{PROPOReecrPhiEpIntegObsta}, \ref{PROPOReecrPhiEpElimiRecol} and \ref{PROPOCompaFinal_Phi_PhiEp}, we have:
\begin{align}
\Big\vert \int_{\mathbb{R}^d_x} \hspace{-1mm} \int_{\mathbb{R}^d_v} \varphi(x,v) f_\varepsilon(t,\dd x, \dd v) &- \int_{\mathbb{R}^d_x} \hspace{-1mm} \int_{\mathbb{R}^d_v} \varphi(x,v) f(t,\dd x,\dd v) \Big\vert \leq \Big\vert \int_{\mathbb{R}^d_x} \hspace{-1mm} \int_{\mathbb{R}^d_v} \big[ R_1 + R_2 + R_3 \big](t,x,v) f_0(\dd x,\dd v) \Big\vert,
\end{align}
where $R_1$, $R_2$ and $R_3$ satisfy respectively \eqref{EQUATDefinReste_R_1_}, \eqref{EQUATDefinReste_R_2_PropoElimiRecol} and \eqref{EQUATDefinReste_R_3_}.\\
Without loss of generality, we can assume that:
\begin{align}
\max \Big[ \int_{\mathbb{R}^d_x} \hspace{-1mm} \int_{\mathbb{R}^d_v} f_0(\dd x,\dd v), \int_{\mathbb{R}^d_x} \hspace{-1mm} \int_{\mathbb{R}^d_v} e^{\vert v \vert^p} f_0(\dd x,\dd v) \Big] = 1.
\end{align}
Since the intensity $\mu_\varepsilon$ of the Poisson process satisfies the Boltzmann-Grad scaling \eqref{EQUATBoltzmann-Grad_Limit}, we have:
\begin{align}
\label{EQUATTheorFinalEstim_R_1_}
\int_{\mathbb{R}^d_x} \hspace{-1mm} \int_{\mathbb{R}^d_v} \vert R_1 \vert f_0(\dd x,\dd v) \leq \vertii{\varphi}_\infty C(d) \varepsilon.
\end{align}
\noindent
Since by \eqref{EQUATDefinReste_R_2_PropoElimiRecol} and \eqref{EQUATDefinReste_R_3_} we have:
\begin{align}
\vert R_2 \vert \leq \vertii{ \varphi }_\infty \widetilde{C}_1 \big(1 + \max(1,t) \vert v \vert \big) e^{\widetilde{C}_2 t \vert v \vert} \varepsilon^{1/4} \hspace{2mm} \text{and} \hspace{2mm} \vert R_3 \vert \leq \vertii{ \varphi }_\infty \widetilde{C}_3 \max\big( 1,t\vert v \vert^2,t^2 \vert v \vert^2 \big) e^{\widetilde{C}_4 t \vert v \vert} \varepsilon^{1/4},
\end{align}
and since $t_0 > 0$ is fixed, there exists a constant $M = M(d,r,p,t_0,\vertii{\varphi}_\infty)$ such that
\begin{align}
\sup_{\substack{t \in [0,t_0] \\ x\in\mathbb{R}^d\hspace{-1mm},\ v \in \mathbb{R}^d}} \big[\vert R_2 \vert + \vert R_3 \vert\big] e^{- \vert v \vert^p} \leq M \varepsilon^{1/4}.
\end{align}
We find therefore:
\begin{align}
\label{EQUATTheorFinalEstim_R_2_v_grd}
\Big\vert \int_{\mathbb{R}^d_x} \hspace{-1mm} \int_{\mathbb{R}^d_v} R_2(t,x,v) f_0(\dd x,\dd v) \Big\vert \leq M \varepsilon^{1/4} \Big\vert \int_{\mathbb{R}^d_x} \hspace{-1mm} \int_{\mathbb{R}^d_v} e^{\vert v \vert^p} f_0(\dd x,\dd v) \Big\vert \leq M \varepsilon^{1/4},
\end{align}
and similarly:
\begin{align}
\label{EQUATTheorFinalEstim_R_3_v_grd}
\Big\vert \int_{\mathbb{R}^d_x} \hspace{-1mm} \int_{\mathbb{R}^d_v} R_3(t,x,v) f_0(\dd x,\dd v) \Big\vert \leq M \varepsilon^{1/4}.
\end{align}
Gathering \eqref{EQUATTheorFinalEstim_R_1_}, \eqref{EQUATTheorFinalEstim_R_2_v_grd} and \eqref{EQUATTheorFinalEstim_R_3_v_grd}, we obtain that there exists a constant $C_\text{final} = C_\text{final}(d,r,p,t_0,\vertii{\varphi}_\infty)$ such that for any $\varepsilon$ smaller than the universal constant $\min(\varepsilon_0,1) > 0$, we have:
\begin{align}
\Big\vert \int_{\mathbb{R}^d_x} \hspace{-1mm} \int_{\mathbb{R}^d_v} \big[ \varphi_\varepsilon(t,x,v) - \psi(t,x,v) \big] f_0(\dd x, \dd v) \Big\vert \leq C_\text{final} \varepsilon^{1/4},
\end{align}
which concludes the proof of the weak$-*$ convergence of $f_\varepsilon(t,\cdot,\cdot)$ towards the measure $g$. Since $g$ is a weak solution to \eqref{EQUATLineaBoltzSphDuFormeForte} in the sense of Definition \ref{DEFINSolutFaibleEquatBoltzLineaInela} according to Proposition \ref{PROPOExistWeak_Solut}, and since such a weak solution is unique according to Proposition \ref{PROPOUniquWeak_Solut}, the proof of Theorem \ref{THEORDerivationBoltzmann_InelaLinea} is complete.
\end{proof}

\section{Well-posedness of the dynamics of the particle system}
\label{SECTIWell_PosedDynam}

In this section, we will address the question of the well-posedness of the dynamics of the particle system. More precisely, we will prove that the forward dynamics of a tagged particle colliding inelastically with inelastic, fixed scatterers distributed according to a Poisson process is globally well-posed, except for a set of distributions of scatterers realized with a zero probability.\\
This question is fundamental: there is no hope to establish a rigorous derivation of a kinetic equation from a particle system without proving before that the dynamics of the particle system is indeed well-posed. We also emphasize that, to the best of our knowledge, such a result has not been established  in the case of inelastic scatterers.\\
It is usually argued that the elastic case can be addressed using the result of Burago-Ferleger-Kononenko \cite{BuFK998}, which establishes that there exists a bound on the maximal number of collisions that a system of $N$ elastic hard spheres can experience, globally in time. Remarkably, the result depends on the number of hard spheres $N$, but is uniform on the initial configurations of the particle system, and holds also when the hard spheres evolve in a domain with boundary, covering in particular the case of one single elastic hard sphere evolving among fixed scatterers.\\
In this section, we present a self-contained result, which applies of course to the inelastic case, but also to the elastic case $r=1$. It can also be easily extended to any type of tagged particle dynamics such that the particle evolves according to the free flow between two collisions with scatterers, and such that the norm of the velocity of the tagged particle does not increase at any collision. The proof relies on a direct adaptation of the original proof due to Alexander (\cite{Alex975}, \cite{Alex976}, revisited in \cite{GSRT013}), concerning the global well-posedness of the system of $N$ elastic hard spheres.

\begin{remar} The proof of Alexander does not rule out the possibility that trajectories mght present infinitely many collisions. To this regard, the results of Burago-Ferleger-Kononenko \cite{BuFK998} and Alexander \cite{Alex975}-\cite{Alex976} are complementary.
\end{remar}

\noindent
From Definition \ref{DEFINFlot_TempsCrois}, it is possible to identify the potential difficulties that may arise when defining the dynamics. The first one is the fact that is not clear how to define the particle dynamics when a collision involving the tagged particle and multiple scatterers occurs at a specific time. Additionally, it is also unclear whether the dynamics can be constructed globally over any arbitrary time interval $[0,t]$, as an infinite number of collisions, i.e. a collapse, may occur strictly before time $t$. We prove the following result. 

\begin{propo}[Global well-posedness of the forward flow]
\label{PROPODefinGlobaDynam}
Let $r \in \ ]0,1[$, $\varepsilon > 0$ and $\mu > 0$ be three positive real numbers. Let $x,v \in \mathbb{R}^d$ be two vectors. We consider a Poisson point process $C$ of intensity $\mu$ in $\mathbb{R}^d$. Then, there exists a subset $\mathcal{P}_{\text{patho.}}$ of scatterers, contained in:
\begin{align}
X_0 = \big\{ c \in C\ /\ c \cap \overline{B(0,\varepsilon)} = \emptyset \big\}
\end{align}
such that
\begin{align}
\mathbb{P}_\mu\big(\mathcal{P}_{\text{patho.}}\big) = 0
\end{align}
and such that if $c \in \mathcal{P}_{\text{patho.}}^c$, then the forward inelastic hard sphere flow, introduced in Definition \ref{DEFINFlot_TempsCrois}, of the tagged particle among the distribution of scatterers $c$ of radius $\varepsilon$ is well-posed on $[0,t]$ for any $t > 0$. In other words, the mapping $s \in [0,t] \mapsto T^s_{c,\varepsilon}(x,v) \in \mathbb{R}^d$ introduced in \ref{DEFINFlot_TempsCrois} is globally well-defined for almost every distribution of scatterers $c$. In addition, the mapping $s \in [0,t] \mapsto T^s_{c,\varepsilon}(x,v)\in \mathbb{R}^d$ is piecewise affine, right continuous and with a limit from the left at all point, and it satisfies \eqref{EQUATDefinFlot_Cond1}-\eqref{EQUATDefinFlot_Cond2}.
\end{propo}

\begin{proof}
Fixing a small time interval $I$, the idea is to determine a set $\mathcal{P}$ of distributions of scatterers for which the consecutive collisions of the tagged particles cannot take place in the same time interval $I$, and to prove that the probability of choosing a distribution in the complement of the set $\mathcal{P}$ is small. As a consequence, the dynamics of a tagged particle evolving among scatterers distributed according to an element of $\mathcal{P}$ is well-posed on such a time interval $I$. Finally, decomposing any arbitrary time interval $[0,t]$ as $[0,t] = \bigcup_{k} I_k$, the objective is to show that $\sum_k \mathcal{P}_k^c = 0$, where $\mathcal{P}_k^c$ is the probability to choose a distribution of scatterers for which the dynamics is not well-defined on $I_k$.\\
Without loss of generality, we assume that the initial velocity $v$ of the tagged particle is such that $\vert v \vert = 1$, and that the initial position of the tagged particle is $0$. Let $t > 0$ be a positive real number. Let $\delta > 0$. We introduce the following sets:
\begin{gather}
X_0 = \Big\{ c\ /\ c \cap \overline{B(0,\varepsilon)} = \emptyset \Big\},\\
X_1 = X_1(t) = \Big\{ c\ /\ c \cap \left( [0,tv] + \overline{B(0,\varepsilon)} \right) \neq \emptyset \Big\},
\end{gather}
where $[0,tv] + \overline{B(0,\varepsilon)}$ denotes:
\begin{align}
[0,tv] + \overline{B(0,\varepsilon)} = \big\{x \in \mathbb{R}^d\ /\ \exists s\in [0,t], y \in \overline{B(0,\varepsilon)}\ /\ x = sv + y \big\}.
\end{align}
In addition, if the distribution $c$ of scatterers belongs to $X_1$, we define the point $x_1$ as:
\begin{align}
x_1 = t_1 v \hspace{5mm} \text{where} \hspace{5mm} t_1 = \min\{s \geq 0\ /\ sv \in c + \overline{B(0,\varepsilon)}\},
\end{align}
and the set:
\begin{align}
Y_1 = Y_1(t,\delta) = \{ c \in X_1 \ /\ \# \{ c \cap \big( x_1 + \overline{B(0,\varepsilon+\delta \vert v \vert}\big)\} = 1 \}.
\end{align}
Let us observe first that if $c \notin X_0$, then the tagged particle, initially at the position $x=0$, lies outside any scatterer. If in addition $c \notin X_1$, then there is no scatterer that intersects the segment $[0,tv]$. In this case, we define the dynamics of the tagged particle on the time interval $[0,t]$ as:
\begin{align}
\left(x(s),v(s)\right) = \left( sv,v \right) \hspace{3mm} \forall s \in [0,t].
\end{align}
If on the contrary $c \in X_0 \cap X_1$, at least one scatterer is intersecting the segment $[0,tv]$. In this case, we have to distinguish between two situations.\\
If $c \in X_0 \cap X_1 \cap Y_1$, in particular a single scatterer intersects the point $x_1$ with probability $1$ (because the center of such a scatterer has to be at distance exactly $\varepsilon$ from $x_1$), so that we can define in a unique way $v'$, using the the reflection law \eqref{EQUATDefinFlot_TempsCroissLoi_Crois}, and choosing the angular parameter $\omega$ as $(c_{i_1}-x_1)/\varepsilon$, where $i_1$ is the index of the scatterer that intersects $x_1$. In addition, we know that no scatterer will intersect the segment $[x_1,x_1+\delta v']$, and since $\vert v' \vert \leq \vert v \vert$, we define the dynamics of the tagged particle as:
\begin{gather}
\left( x(s),v(s) \right) = \left( sv,v \right) \hspace{3mm} \forall s \in [0,t_1[,\\
\left( x(s),v(s) \right) = \left( t_1v + (s-t_1)v',v' \right) = \left(x_1 + (s-t_1)v',v'\right) \hspace{3mm} \forall s \in [t_1,t_1+\delta].
\end{gather}
so that in particular we know that the tagged particle will undergo a single collision on the time interval $[0,\delta]$.\\
If now $c \in X_0 \cap X_1 \cap Y_1^c$, then in particular:
\begin{align}
\# \{ c \cap \big( x_1 + \overline{B(0,\varepsilon+\delta)}\big)\} \geq 2.
\end{align}
In addition, by definition of $x_1$, we have $\# \{ c \cap B(x_1,\varepsilon) \} = 0$, so that if $c \in X_0 \cap X_ \cap Y_1^c$ we have:
\begin{align}
\# \{ c \cap \big( \overline{B(x_1,\varepsilon+\delta)} \backslash B(x_1,\varepsilon) \big)\} \geq 2.
\end{align}
Since the measure of $\overline{B(x_1,\varepsilon+\delta)} \backslash B(x_1,\varepsilon)$ is $C(d,\varepsilon) \delta$, where $C(d,\varepsilon)$ is a number that depends only on the dimension $d$ and the size $\varepsilon$ of the scatterers, the probability that the distribution of scatterers $c$ belongs to $X_0 \cap X_1 \cap Y_1^c$ is given by:
\begin{align}
\mathbb{P}_\mu( c \in X_0 \cap X_1 \cap Y_1^c ) &\leq \sum_{k \geq 2} e^{- \mu \left\vert \overline{B(x_1,\varepsilon+\delta)} \backslash B(x_1,\varepsilon) \right\vert} \mu^k \frac{\left\vert \overline{B(x_1,\varepsilon+\delta)} \backslash B(x_1,\varepsilon) \right\vert^k}{k!} \leq C(d,\varepsilon) \mu^2 \delta^2. 
\end{align}
We proceed recursively: we assume now that we define the dynamics of the particle globally on $[0,t]$ if $c \in \left(X_0 \cap X_1^c\right) \cup \left(X_0 \cap X_1 \cap Y_1 \cap X_2^c\right) \cup \dots \cup \left(X_0 \cap X_1 \cap Y_1 \cap \dots \cap X_{k-2} \cap Y_{k-2} \cap X_{k-1}^c\right)$, and on the time interval $[0,t_k+\delta]$ if $c \in X_0 \cap X_1 \cap Y_1 \cap \dots \cap X_{k-1} \cap Y_{k-1}$, as:
\begin{gather}
\left( x(s),v(s) \right) = \big( sv,v \big) \hspace{3mm} \forall\, s \in [0,t_1[,\label{EQUATDefinTraje__1__}\\
\left( x(s),v(s) \right) = \big( t_1v + (s-t_1)v',v' \big) = \big( x_1+(s-t_1)v',v' \big) \hspace{3mm} \forall s \in [t_1,t_2[,\\
\dots \nonumber\\
\left( x(s),v(s) \right) = \Big( \sum_{j=1}^{k-1} t_j v^{(j-1)} + (s-t_{k-1}) v^{(k-1)},v^{(k-1)} \Big)\nonumber\\
\hspace{48.5mm}= \big( x_{k-1} + (s-t_{k-1})v^{(k-1)},v^{(k-1)} \big) \hspace{3mm} \forall\, s \in [t_{k-1},t_{k-1}+\delta]. \label{EQUATDefinTraje_k-1_}
\end{gather}
Introducing then the set:
\begin{align}
X_k = \big\{ c \in X_0 \cap X_1 \cap Y_1 \cap \dots \cap X_{k-1} \cap Y_{k-1}\ /\ c \cap \Big( \big[x_{k-1},x_{k-1}+t\frac{v^{(k-1)}}{\vert v^{(k-1)} \vert}\big] + \overline{B(0,\varepsilon)} \Big) \neq 0 \big\},
\end{align}
we observe that if $c \in X_0 \cap X_1 \cap Y_1 \cap \dots \cap X_k^c$, we can define globally on $[0,t]$ the dynamics of the tagged particle, by completing the definition \eqref{EQUATDefinTraje__1__}-\eqref{EQUATDefinTraje_k-1_} as:
\begin{align}
\left( x(t),v(t) \right) = \big( x_{k-1} + (s-t_{k-1})v^{(k-1)},v^{(k-1)} \big) \hspace{3mm} \forall\, s \in [t_{k-1},t].
\end{align}
If now $c \in X_0 \cap X_1 \cap Y_1 \cap \dots \cap X_k$, we introduce the position $x_k$ and the set $Y_k$, defined as:
\begin{align}
x_k = x_{k-1} + (t_k-t_{k-1}) v^{(k-1)} \hspace{3mm} \text{where} \hspace{3mm} t_k = \min \{ s \geq t_{k-1}\ /\ x_{k-1} + (s-t_{k-1})v^{(k-1)} \in c + \overline{B(0,\varepsilon)} \}
\end{align}
and
\begin{align}
Y_k = \{ c \in X_1 \cap Y_1 \cap \dots \cap X_k \ /\ \# \{ c \cap \big( x_k + \overline{B(0,\varepsilon+\delta}\big)\} = 1 \}.
\end{align}
As before, we have:
\begin{align}
\mathbb{P}_\mu( c \in X_0 \cap X_1 \cap Y_1 \cap \dots \cap X_k \cap Y_k^c ) \leq C(d,\varepsilon) \mu^2 \delta^2,
\end{align}
and if $c \in X_0 \cap X_1 \cap Y_1 \cap \dots \cap X_k \cap Y_k$, $x_k$ intersects a single scatterer $i_k$ with probability $1$, so that we can complete the definition of the dynamics of the tagged particle as:
\begin{gather}
\big(x(s),v(s)\big) = \big( x_{k-1} + (s-t_{k-1})v^{(k-1)},v^{(k-1)} \big) \hspace{3mm} \forall\, s \in [t_{k-1},t_k[\\
\big(x(s),v(s)\big) = \big( x_{k-1} + (t_k-t_{k-1})v^{(k-1)} + (s-t_k)v^{(k)},v^{(k)} \big) \hspace{3mm} \forall\, s \in [t_k,t_k+\delta].
\end{gather}
In summary, we defined recursively the dynamics of the tagged particle, either on the whole time interval $[0,t]$, or on the time interval $[0,t_k+\delta]$ for any $k \geq 1$ if $c \in (X_0 \cap X_1^c) \cup (X_0 \cap X_1 \cap Y_1 \cap X_2^c) \cup \dots \cup (X_0 \cap X_1 \cap Y_1 \cap \dots X_k \cap Y_k)$.\\
On the one hand, since by construction we have $t_{k+1} - t_k \geq \delta \hspace{1mm} \forall\, k$, we need only $t/\delta$ iterations to define the dynamics on the whole time interval $[0,t]$.\\
On the other hand, the dynamics is not defined if $c \in X_0$, or if $c \in (X_0 \cap X_1 \cap Y_1^c) \cup (X_0 \cap X_1 \cap Y_1 \dots \cap X_k \cap Y_k^c)$. But since we have, for $k_0$ the smallest integer larger or equal to $t/\delta$:
\begin{align}
\label{EQUATMesurEnsmbPatho}
\mathbb{P}_\mu \big( (X_0 \cap X_1 \cap Y_1^c) \cup \dots \cup (X_0 \cap X_1 \cap Y_1 \dots \cap X_{k_0} \cap Y_{k_0}^c) \big)
&\leq \sum_{k=1}^{k_0} \mathbb{P} \big( X_0 \cap X_1 \cap Y_1 \dots \cap X_k \cap Y_k^c \big) \nonumber\\
&\leq C(d,\varepsilon) \mu^2 k_0 \delta^2.
\end{align}
Now, for all $\delta > 0$ such that $t/\delta$ is an integer, we consider the intersection:
\begin{align}
\label{EQUATDefinEnsmbPatho}
\mathcal{P}_{\text{patho.}}(t) = \bigcap_{\substack{\delta > 0\\ t/\delta \in \mathbb{N}^*}} \Big[ (X_0 \cap X_1 \cap Y_1(\delta)^c) \cup \dots \cup (X_0 \cap X_1 \cap Y_1(\delta) \dots \cap X_{k_0}(\delta) \cap Y_{k_0}(\delta)^c) \big) \Big].
\end{align}
By \eqref{EQUATMesurEnsmbPatho}, we have
\begin{align}
\mathbb{P}_\mu\big(\mathcal{P}_{\text{patho.}}(t)\big) = 0,
\end{align}
and also, if $c \in X_0 \cap \mathcal{P}(t)^c$, that the dynamics of the tagged particle is defined on the whole time interval $[0,t]$. In addition, if $c \in X_0 \cap \mathcal{P}(t)^c$, there exists $\delta_0 > 0$ such that all the collisions of the trajectory of the tagged particles are separated by a time interval larger than $\delta_0$.\\
Finally, we can repeat the argument for a countable sequence of times $(t_n)_{n\in \mathbb{N}^*}$ which tends to infinity as $n \rightarrow +\infty$, to obtain that the dynamics of the tagged particle is defined on the whole time interval $[0,+\infty[$, with probability $1$. The proof of Proposition \ref{PROPODefinGlobaDynam} is complete.
\end{proof}

\begin{remar}
As in the original proof of Alexander for the global well-posedness of the elastic hard sphere system, the key argument is that, except for a pathological set of measure $\delta^2$, we can define the dynamics further on a time interval of length $\delta$. In the original proof, the pathological set corresponds to the configurations in the phase space such that at least two pairs of particles are at a distance smaller than $\delta$. In the present case, we rely on the property of a Poisson process: having two scatterers or more in an annulus of radius $\delta$ (or, in general dimension $d$, in the difference of two concentric balls with radii that differ by $\delta$) has a probability smaller than $\delta^2$.\\
We observe also that the proof can be reused without any change in the case when the interaction with the scatterers is more general, provided only that the norm of the velocity of the tagged particle is non-increasing along the collisions.
\end{remar}

\section{Conclusion}

In the present paper we presented a rigorous derivation of the inelastic linear Boltzmann equation \eqref{EQUATInelaLineaBoltzFinal}, from the inelastic Lorentz gas, in the Boltzmann-Grad limit (Theorem \ref{THEORDerivationBoltzmann_InelaLinea}). To the best of our knowledge, this result constitutes the first rigorous derivation of an inelastic version of the Boltzmann equation.\\
To achieve such a derivation, we proved the absence of inelastic collapse in the forward in time dynamics of the inelastic Lorentz gas (Proposition \ref{PROPODefinGlobaDynam}), and the analog of Alexander's theorem in our setting, ensuring that the dynamics of the particle system is well-posed.\\
In addition, we relied on the series representation \eqref{EQUATRepreSerieEquatAdjoi} of the solutions of the adjoint equation \eqref{EQUATLineaBoltzSphDuFormeAdjoi}. Assuming that an exponential moment of the initial datum is finite, we showed the convergence of such a series. This result allowed us to proof the existence of weak solutions to the inelastic linear Boltzmann equation (Proposition \ref{PROPOExistWeak_Solut}). Under the same assumption, we also proved the convergence of the series \eqref{EQUATRepreSerieSolution__BoltzLineaInela}, which provided the existence of explicit strong solutions to the linear inelastic Boltzmann equation (Theorem \ref{PROPOConveRepreSerieSolutBoltzLineaInela}).\\
We further remark that we obtained explicit estimates  on the measures of the pathological sets preventing the Markovian behaviour of the limit process, i.e.~the configurations of scatterers leading to recollisions or interferences, hence allowing us to provide a quantitative derivation of the linear inelastic Boltzmann equation.
\newline 
Nevertheless,  the rigorous derivation we established in this paper holds only in terms of weak convergence of the distribution function \eqref{EQUATQuantGallavotti} of the microscopic tagged particles towards the associated solution of \eqref{EQUATInelaLineaBoltzFinal}. In future work, our aim is to perform the derivation in the stronger $L^p$ topology, in the same spirit of the original proof obtained in the elastic case by Gallavotti \cite{Gallavotti}. Another question that it is worth investigating is the analysis of the long-time behaviour of solutions to the linear inelastic Boltzmann equation \eqref{EQUATInelaLineaBoltzFinal}. In particular, the characterization of the decay of the temperature (Haff's law) remains to be proved.

\begin{appendices}

\section{Proof of the series representation of the solutions to the adjoint equation \eqref{EQUATLineaBoltzSphDuFormeAdjoi}}
\label{SECTIProofSerieRepreSolutAdjoi}

\begin{proof}[Proof of Proposition \ref{PROPORepreSerieEquatAdjoi}]
For any integer $n \geq 0$, we introduce the function $\varphi_n$ defined as:
\begin{align}
\widetilde{\varphi}_n(t,x,v) &= e^{-C_d \vert v \vert t}\varphi(x+tv,v) + \sum_{k=1}^n \int_{t_1=0}^t \int_{\mathbb{S}^{d-1}_{\omega_1}} \hspace{-3mm}\dots \int_{t_k=0}^{t_{k-1}}\int_{\mathbb{S}^{d-1}_{\omega_k}} e^{\big[ \sum_{j=1}^k C_d \vert v^{(j-1)} \vert (t_j-t_{j-1}) - C_d \vert v^{(k)} \vert t_k \big]}  \nonumber\\
&\hspace{15mm} \times \prod_{l=1}^k \big\vert v^{(l-1)}\cdot\omega_l \big\vert \varphi(x + \sum_{m=1}^k \big(t_{m-1}-t_m\big) v^{(m-1)} + t_k v^{(k)},v^{(k)}) \dd \omega_k \dd t_k \dots \dd \omega_1 \dd t_1.
\end{align}
The function $\widetilde{\varphi}_n$ is $\mathcal{C}^1$, and we have:
\begin{align}
\partial_t \widetilde{\varphi}_n(t,x,v) &= v \cdot \nabla_x \varphi(x+tv,v) e^{-C_d \vert v \vert t} - C_d \vert v \vert \varphi(x+tv,v) e^{-C_d \vert v \vert t} + I_1 + I_2 + I_3
\end{align}
with
\begin{align}
I_1 &= \sum_{k=1}^n \int_{\mathbb{S}^{d-1}_{\omega_1}} \hspace{-3mm}\dots \int_{t_k=0}^{t_{k-1}}\int_{\mathbb{S}^{d-1}_{\omega_k}} e^{\big[ \sum_{j=2}^k C_d \vert v^{(j-1)} \vert (t_j-t_{j-1}) - C_d \vert v^{(k)} \vert t_k \big]} \prod_{l=1}^k \big\vert v^{(l-1)}\cdot\omega_l \big\vert \nonumber\\
&\hspace{45mm} \times \varphi(x + \sum_{m=2}^k \big(t_{m-1}-t_m\big) v^{(m-1)} + t_k v^{(k)},v^{(k)}) \dd \omega_k \dd t_k \dots \dd t_2 \dd \omega_1,
\end{align}
\begin{align}
I_2 &= \sum_{k=1}^n \int_{t_1=0}^t \int_{\mathbb{S}^{d-1}_{\omega_1}} \hspace{-3mm}\dots \int_{t_k=0}^{t_{k-1}}\int_{\mathbb{S}^{d-1}_{\omega_k}} \big( - C_d \vert v \vert \big) e^{\big[ \sum_{j=1}^k C_d \vert v^{(j-1)} \vert (t_j-t_{j-1}) - C_d \vert v^{(k)} \vert t_k \big]} \prod_{l=1}^k \big\vert v^{(l-1)}\cdot\omega_l \big\vert \nonumber\\
&\hspace{45mm} \times \varphi(x + \sum_{m=1}^k \big(t_{m-1}-t_m\big) v^{(m-1)} + t_k v^{(k)},v^{(k)}) \dd \omega_k \dd t_k \dots \dd \omega_1 \dd t_1
\end{align}
and
\begin{align}
I_3 &= \sum_{k=1}^n \int_{t_1=0}^t \int_{\mathbb{S}^{d-1}_{\omega_1}} \hspace{-3mm}\dots \int_{t_k=0}^{t_{k-1}}\int_{\mathbb{S}^{d-1}_{\omega_k}} e^{\big[ \sum_{j=1}^k C_d \vert v^{(j-1)} \vert (t_j-t_{j-1}) - C_d \vert v^{(k)} \vert t_k \big]} \prod_{l=1}^k \big\vert v^{(l-1)}\cdot\omega_l \big\vert \nonumber\\
&\hspace{40mm} \times v \cdot \nabla_x \varphi(x + \sum_{m=1}^k \big(t_{m-1}-t_m\big) v^{(m-1)} + t_k v^{(k)},v^{(k)}) \dd \omega_k \dd t_k \dots \dd \omega_1 \dd t_1.
\end{align}
Observe that the effect of the time derivative on the sum produces two terms. On the one hand for the first term, the integral in $t_1$ disappears, and the integrand is evaluated in the variable $t_1$ at $t$. Therefore, the first term in the sum inside the exponential vanishes, as well as the first term in $\varphi$ below the integral. On the other hand, the second term is obtained by differentiating with respect to $t$ below the integrals.\\
Identifying the different terms in the previous computation, we find that $\varphi_n$ solves the equation:
\begin{align}
\partial_t \widetilde{\varphi}_n = v \cdot \nabla_x \widetilde{\varphi}_n - C_d \vert v \vert \widetilde{\varphi}_n + \int_{\mathbb{S}^{d-1}_\omega} \hspace{-3mm} \vert v \cdot \omega \vert \widetilde{\varphi}_{n-1}(t,x,v') \dd \omega.
\end{align}
In addition, the sequences $\big( \widetilde{\varphi}_n \big)_{n \geq 0}$, $\big( \partial_t \widetilde{\varphi}_n \big)_{n \geq 0}$ and $\big( \nabla_x \widetilde{\varphi}_n \big)_{n \geq 0}$ being Cauchy sequences locally uniformly in $(t,x,v)$, they converge locally uniformly towards respective limits $\widetilde{\psi}$, $\widetilde{g}$ and $\widetilde{h}$, that satisfy:
\begin{align}
\widetilde{g}(t,x,v) = v \cdot \widetilde{h}(t,x,v) - C_d \vert v \vert \widetilde{\psi}(t,x,v) +  \int_{\mathbb{S}^{d-1}_\omega} \hspace{-3mm} \vert v \cdot \omega \vert \widetilde{\psi}(t,x,v') \dd \omega,
\end{align}
and where $\widetilde{\psi}$ is given by the expression \eqref{EQUATRepreSerieEquatAdjoi}. Since in addition $\widetilde{\varphi}_n$ and its derivatives are converging locally uniformly, we have $\partial_t \widetilde{\psi} = \widetilde{g}$ and $\nabla_x \widetilde{\psi} = \widetilde{h}$, so that $\widetilde{\psi}$ solves the adjoint equation \eqref{EQUATLineaBoltzSphDuFormeAdjoi} of the linear inelastic Boltzmann equation \eqref{EQUATLineaBoltzSphDuFormeForte}, concluding the proof of the first part of Proposition \ref{PROPORepreSerieEquatAdjoi}.\\
Conversely, for a solution $\widetilde{\varphi}$ of the adjoint equation \eqref{EQUATLineaBoltzSphDuFormeAdjoi} with initial datum $\varphi$, integrating $\widetilde{\varphi}$ along the characteristics it is possible to show by recursion that, for all $n \geq 0$:
\begin{align}
\label{EQUATFormuRecurSolutFormeAdjoi}
\widetilde{\varphi}(t,x,v) &= e^{-C_d \vert v \vert t}\varphi(x+tv,v) + \sum_{k=1}^{n} \int_{t_1=0}^t \int_{\mathbb{S}^{d-1}_{\omega_1}} \hspace{-3mm}\dots \int_{t_k=0}^{t_{k-1}}\int_{\mathbb{S}^{d-1}_{\omega_k}} e^{\big[ \sum_{j=1}^k C_d \vert v^{(j-1)} \vert (t_j-t_{j-1}) - C_d \vert v^{(k)} \vert t_k \big]} \prod_{l=1}^k \big\vert v^{(l-1)}\cdot\omega_l \big\vert \nonumber\\
&\hspace{45mm} \times \varphi(x + \sum_{m=1}^k \big(t_{m-1}-t_m\big) v^{(m-1)} + t_k v^{(k)},v^{(k)}) \dd \omega_k \dd t_k \dots \dd \omega_1 \dd t_1 \nonumber\\
&+ \int_{t_1=0}^t \int_{\mathbb{S}^{d-1}_{\omega_1}} \hspace{-3mm}\dots \int_{t_n=0}^{t_{n-1}}\int_{\mathbb{S}^{d-1}_{\omega_n}} \int_{t_{n+1}=0}^{t_n} \int_{\mathbb{S}^{d-1}_{\omega_{n+1}}} e^{\big[ \sum_{j=1}^{n+1} C_d \vert v^{(j-1)} \vert (t_j-t_{j-1}) \big]} \prod_{l=1}^{n+1} \big\vert v^{(l-1)}\cdot\omega_l \big\vert \nonumber\\
&\hspace{35mm} \times \widetilde{\varphi}(t_{n+1},x + \sum_{m=1}^{n+1} \big(t_{m-1}-t_m\big) v^{(m-1)},v^{(n+1)}) \dd \omega_{n+1} \dd t_{n+1} \dd \omega_n \dd t_n \dots \dd \omega_1 \dd t_1.
\end{align}
The difference between $\widetilde{\varphi}$ and $\widetilde{\varphi}_n$ is exactly equal to the last term in the previous formula \eqref{EQUATFormuRecurSolutFormeAdjoi}. We call $R$ such a term. Since by assumption $\widetilde{\varphi}$ is a $\mathcal{C}^1$ function, it is locally bounded, and we have:
\begin{align}
\big\vert R \big\vert &\leq \vertii{ \mathds{1}_{[0,t]\times B(x,t \vert v \vert)\times B(0,\vert v \vert)} \widetilde{\varphi}}_\infty \Big\vert \int_{t_1=0}^t \int_{\mathbb{S}^{d-1}_{\omega_1}} \hspace{-3mm}\dots \int_{t_n=0}^{t_{n-1}}\int_{\mathbb{S}^{d-1}_{\omega_n}} \int_{t_{n+1}=0}^{t_n} \int_{\mathbb{S}^{d-1}_{\omega_{n+1}}}  \vert v \vert^{n+1} \dd \omega_{n+1} \dd t_{n+1} \dd \omega_n \dd t_n \dots \dd \omega_1 \dd t_1 \Big\vert \nonumber\\
&\leq \vertii{ \mathds{1}_{[0,t]\times B(x,t \vert v \vert)\times B(0,\vert v \vert)} \widetilde{\varphi}}_\infty \frac{t^{n+1}  \vert v \vert^{n+1} \vert \mathbb{S} \vert^{n+1}}{(n+1)!},
\end{align}
we deduce that $\widetilde{\varphi} = \widetilde{\psi}$, concluding the proof of the second part of Proposition \ref{PROPORepreSerieEquatAdjoi}.
\end{proof}

\section{Proofs of the geometrical lemmas \ref{LEMMEColinearitScatt} and \ref{LEMMEEstimMesurTube_Dynam}}

\subsection{Estimate coming from the colinearity condition after scattering}

\begin{proof}[Proof of Lemma \ref{LEMMEColinearitScatt}]
The proof of Lemma \ref{LEMMEColinearitScatt} is divided into the following steps:
\begin{itemize}
\item we write the scalar product $\frac{v'}{\vert v' \vert} \cdot p$ in coordinates, in terms of the three vectors $v$, $\omega$ and $\sigma$ ($\sigma$ being the angular parameter such that the renormalized post-collisional velocity obtained from $v$ with angular parameter $\sigma$ is $p$),
\item we analyze in detail the particular case when $v$, $\omega$ and $\sigma$ are coplanar,
\item we finally prove that the case when $v$, $\omega$ and $\sigma$ are not coplanar can be treated as a consequence of the result obtained in the previous particular case.
\end{itemize}

\noindent
The scalar product $\frac{v'}{\vert v' \vert} \cdot p$, that we will denote as $f$ in the rest of this proof, can be rewritten as:
\begin{align}
\label{EQUATPreuvLemmeColinExprePrdSc__1__}
f = \frac{v'}{\vert v' \vert} \cdot p &= \left[ \frac{v - (1+r) \big( v \cdot \omega \big) \omega}{ \big\vert v - (1+r) \big( v \cdot \omega \big) \omega \big\vert} \right] \cdot \left[ \frac{ v - (1+r) \big( v\cdot\sigma \big) \sigma}{ \big\vert v - (1+r) \big( v\cdot\sigma \big) \sigma \big\vert } \right] \nonumber\\
&= \left[ \frac{\frac{v}{\vert v \vert} - (1+r) \big( \frac{v}{\vert v \vert} \cdot \omega \big) \omega}{ \big\vert \frac{v}{\vert v \vert} - (1+r) \big( \frac{v}{\vert v \vert} \cdot \omega \big) \omega \big\vert } \right] \cdot \left[ \frac{ \frac{v}{\vert v \vert} - (1+r) \big( \frac{v}{\vert v \vert}\cdot\sigma \big) \sigma}{ \big\vert \frac{v}{\vert v \vert} - (1+r) \big( \frac{v}{\vert v \vert}\cdot\sigma \big) \sigma \big\vert } \right] .
\end{align}
The norms in the denominator are given by the following expressions:
\begin{align}
\Big\vert \frac{v}{\vert v \vert} - (1+r) \big( \frac{v}{\vert v \vert} \cdot \omega \big) \omega \Big\vert &= \sqrt{ 1 - 2 (1+r) \Big( \frac{v}{\vert v \vert} \cdot \omega \Big)^2 + (1+r)^2 \Big( \frac{v}{\vert v \vert} \cdot \omega \Big)^2} \nonumber\\
&= \sqrt{ 1 - (1-r^2) \Big( \frac{v}{\vert v \vert} \cdot \omega \Big)^2},
\end{align}
and similarly:
\begin{align}
\Big\vert \frac{v}{\vert v \vert} - (1+r) \big( \frac{v}{\vert v \vert} \cdot \omega \big) \omega \Big\vert = \sqrt{ 1 - (1-r^2) \Big( \frac{v}{\vert v \vert} \cdot \sigma \Big)^2} .
\end{align}
Expanding the scalar product in the numerator of \eqref{EQUATPreuvLemmeColinExprePrdSc__1__}, we find:
\begin{align}
f = \frac{1 - (1+r) \big( \frac{v}{\vert v \vert}\cdot\sigma \big)^2 - (1+r) \big( \frac{v}{\vert v \vert}\cdot\omega \big)^2 + (1+r)^2 \big(v \cdot \sigma\big)\big(v \cdot \omega\big) \big(\sigma\cdot\omega\big)}{ \sqrt{ \Big[ 1 - (1-r^2) \big( \frac{v}{\vert v \vert} \cdot \sigma \big)^2 \Big] \cdot \Big[ 1 - (1-r^2) \big( \frac{v}{\vert v \vert} \cdot \omega \big)^2 \Big]} } \cdotp
\end{align}
From this point, it will be convenient to decompose the unit vector $\omega$ as its component along the plane spanned by $v$ and $\sigma$ and its orthogonal, that is we write:
\begin{align}
\omega = a \frac{v}{\vert v \vert} + b \sigma + u,
\end{align}
with $a,b \in \mathbb{R}$, and $u \in \mathbb{R}^d$ such that $u \cdot v = u \cdot \sigma = 0$.\\
In what follows, we will denote the following scalar products as:
\begin{align}
\frac{v}{\vert v \vert}\cdot\sigma = \cos\theta_0 \hspace{3mm} \text{and} \hspace{3mm} \frac{v}{\vert v \vert} \cdot \omega = \cos\theta,
\end{align}
where $\theta_0$ is the orientated angle from $v$ to $\sigma$, and $\theta$ is the orientated angle from $v$ to $\omega$ (after having chosen an orientation on the planes respectively spanned by $(v,\sigma)$, and $(v,\omega)$). $\omega$ being unitary, we have:
\begin{align}
1 = \vert \omega \vert^2 = a^2 + 2ab \cos\theta + b^2 + \vert u \vert^2.
\end{align}
With these notations, we have:
\begin{align}
\frac{v}{\vert v \vert}\cdot\omega = \cos\theta = a + b \cos\theta_0 \hspace{3mm} \text{and} \hspace{3mm} \sigma\cdot\omega = a \cos\theta_0 + b.
\end{align}
Finally, considering the scalar product $f$ at $v$, $\sigma$ and $\vert u \vert$ fixed, it is possible to express $\sigma\cdot\omega$ as a function of $\cos\theta$ by writing:
\begin{align}
\sigma \cdot \omega = a \cos\theta_0 + b = \cos\theta_0 \big(a + b \cos\theta_0\big) + b \sin^2\theta_0 = \cos\theta_0 \cos\theta + \sin\theta_0 \, \text{sgn}(b) \sqrt{1 - \vert u \vert^2 - \cos^2\theta},
\end{align}
where we used the identity:
\begin{align}
1 = \cos^2\theta + b^2\sin^2\theta_0 + \vert u \vert^2.
\end{align}
Observe that the sign of $\sin\theta_0$ has been taken positive by assumption, which can always been done by choosing appropriately the orientation in the plane spanned by $v$ and $\sigma$. Nevertheless, the sign of $b$ has to be taken into account.\\
In the end, the scalar product $f$ can be rewritten as:
\begin{align}
\label{EQUATPreuvLemmeColinExprePrdScGener}
f = \frac{1 - (1+r) \cos^2\theta_0 - (1+r) \cos^2\theta + (1+r)^2 \cos\theta_0 \cos\theta \big( \cos\theta_0 \cos\theta + \sin\theta_0 \, \text{sgn}(b) \sqrt{1 - \vert u \vert^2 - \cos^2\theta} \big)}{ \sqrt{ \Big[ 1 - (1-r^2) \cos^2\theta_0 \Big] \cdot \Big[ 1 - (1-r^2) \cos^2\theta \Big]} } \cdotp
\end{align}
Obtaining the general expression \eqref{EQUATPreuvLemmeColinExprePrdScGener} of the scalar product $f$ in coordinates concludes the first step of the proof.\\
\newline
Turning to the second step, we now assume that the three vectors $v$, $\sigma$ and $\omega$ are coplanar, that is, we assume that $u = 0$. We now study $f$ under this assumption.\\
In such a case, the sign of $\sin\theta$ is equal to the sign of $b$, and we have:
\begin{align}
\cos\theta_0 \cos\theta + \sin\theta_0 \, \text{sgn}(b) \sqrt{1 - \vert u \vert^2 - \cos^2\theta} = \cos\theta_0 \cos\theta + \sin\theta_0 \, \text{sgn}\big( \sin\theta \big) \sqrt{1 - \cos^2 \theta} = \cos(\theta-\theta_0).
\end{align}
Therefore, the scalar product $f$ can be rewritten as:
\begin{align}
f = \frac{1 - (1+r) \cos^2\theta_0 - (1+r) \cos^2\theta + (1+r)^2 \cos\theta_0 \cos\theta \cos(\theta-\theta_0)}{ \sqrt{ \Big[ 1 - (1-r^2) \cos^2\theta_0 \Big] \cdot \Big[ 1 - (1-r^2) \cos^2\theta \Big]} } \cdotp
\end{align}
We consider $f$ as a function of $\theta$, with $\theta_0$ fixed. Its derivative with respect to $\theta$ writes:
\begin{align}
\partial_\theta f &= \frac{\Big[ 2(1+r)\cos\theta\sin\theta - (1+r)^2 \cos\theta_0 \sin\theta \cos(\theta-\theta_0) - (1+r)^2 \cos\theta_0 \cos\theta \sin(\theta-\theta_0) \Big]\big[ 1 - (1-r^2)\cos^2\theta \big]}{\Big[ 1 - (1-r^2) \cos^2\theta_0 \Big]^{1/2} \cdot \Big[ 1 - (1-r^2) \cos^2\theta \Big]^{3/2}} \nonumber\\
& \hspace{5mm} - \frac{ \Big[ 1 - (1+r) \cos^2\theta_0 - (1+r) \cos^2\theta + (1+r)^2 \cos\theta_0 \cos\theta \cos(\theta-\theta_0)\Big] \big[ (1-r^2) \cos\theta\sin\theta \big] }{\Big[ 1 - (1-r^2) \cos^2\theta_0 \Big]^{1/2} \cdot \Big[ 1 - (1-r^2) \cos^2\theta \Big]^{3/2}} \cdotp
\end{align}
Gathering the terms:
\begin{align}
2(1+r) \cos\theta \sin\theta - (1-r^2)\cos\theta\sin\theta = (1+r)^2 \cos\theta\sin\theta,
\end{align}
we find:
\begin{align}
\partial_\theta f &= \frac{(1+r)^2\cos\theta\sin\theta -2(1+r)(1-r^2)\cos^3\theta\sin\theta}{ \Big[ 1 - (1-r^2) \cos^2\theta_0 \Big]^{1/2} \cdot \Big[ 1 - (1-r^2) \cos^2\theta \Big]^{3/2} } \nonumber\\
&\hspace{5mm} - \frac{\Big[ (1+r)^2 \cos\theta_0 \sin\theta \cos(\theta-\theta_0) + (1+r)^2 \cos\theta_0 \cos\theta\sin(\theta-\theta_0) \Big] \big[1 - (1-r^2)\cos^2\theta \big]}{ \Big[ 1 - (1-r^2) \cos^2\theta_0 \Big]^{1/2} \cdot \Big[ 1 - (1-r^2) \cos^2\theta \Big]^{3/2} } \nonumber\\
&\hspace{5mm} + \frac{ \Big[ (1+r)\cos^2\theta_0 + (1+r) \cos^2\theta - (1+r)^2 \cos\theta_0 \cos\theta \cos(\theta-\theta_0) \Big] \big[ (1-r^2) \cos\theta \sin\theta \big] }{ \Big[ 1 - (1-r^2) \cos^2\theta_0 \Big]^{1/2} \cdot \Big[ 1 - (1-r^2) \cos^2\theta \Big]^{3/2} } \cdotp
\end{align}
It is now possible to factor by $(1+r)^2$ in all the terms of the derivative. Denoting now by $g$ the quantity:
\begin{align}
g = \frac{\Big[ 1 - (1-r^2) \cos^2\theta_0 \Big]^{1/2} \cdot \Big[ 1 - (1-r^2) \cos^2\theta \Big]^{3/2} \partial_\theta f }{(1+r)^2},
\end{align}
and observing that the terms $(1+r)^2(1-r^2)\cos\theta_0\cos^2\theta\sin\theta\cos(\theta-\theta_0)$ in the second and third lines cancel each other, we obtain:
\begin{align}
g &= \cos\theta\sin\theta - 2 (1-r) \cos^3\theta\sin\theta - \big[ \cos\theta_0 \sin\theta \cos(\theta-\theta_0) + \cos\theta_0 \cos\theta \sin(\theta-\theta_0) \big] \nonumber\\
&\hspace{5mm} + (1-r^2) \cos\theta_0 \cos^3\theta \sin(\theta-\theta_0) + (1-r) \big[ \cos^2\theta_0 + \cos^2\theta \big] \cos\theta\sin\theta \nonumber\\
&= \sin(\theta-\theta_0) \Big[ - \cos\theta_0\cos\theta + (1-r^2) \cos\theta_0\cos^3\theta \Big] \nonumber\\
&\hspace{5mm} + \sin\theta \Big[ \cos\theta - 2(1-r)\cos^3\theta - \cos\theta_0\cos(\theta-\theta_0) + (1-r) \big[\cos^2\theta_0 + \cos^2\theta \big]\cos\theta \Big] \nonumber\\
&= \sin(\theta-\theta_0) \Big[ (1-r^2) \cos\theta_0\cos^3\theta - \cos\theta_0\cos\theta \Big] \nonumber\\
&\hspace{5mm} + \sin\theta \Big[ \cos\theta + (1-r)\big[ \cos^2\theta_0 - \cos^2\theta \big] \cos\theta - \cos\theta_0\cos(\theta-\theta_0) \Big].
\end{align}
In the case when $\theta = \theta_0$, $\omega$ and $\sigma$ are colinear, so $v'$ and $p$ are colinear in this case. We have in particular $f =1$. Since $\frac{v'}{\vert v' \vert} \cdot p$ is the scalar product of two unit vectors, $\theta = \theta_0$ corresponds therefore to a maximum of the function $f$. As a consequence, to obtain a simpler expression of the derivative of $f$, we will factor the expression of $g$ by $\sin(\theta-\theta_0)$, that is a factor of already some terms.\\
To do so, we use the following identity:
\begin{align}
\cos^2 \theta_0 - \cos^2\theta = \sin^2\theta - \sin^2\theta_0 = \sin(\theta+\theta_0)\sin(\theta-\theta_0).
\end{align}
Therefore we have:
\begin{align}
\cos\theta &+ (1-r) \big[ \cos^2\theta_0-\cos^2\theta \big] \cos\theta - \cos\theta_0\cos(\theta-\theta_0) \nonumber\\
&= \cos\theta + (1-r) \sin(\theta+\theta_0)\sin(\theta-\theta_0)\cos\theta - \cos^2\theta_0 \cos\theta - \cos\theta_0\sin\theta_0 \sin\theta \nonumber\\
&= (1-r) \sin(\theta+\theta_0)\sin(\theta-\theta_0)\cos\theta + \sin^2\theta_0 \cos\theta - \cos\theta_0\sin\theta_0\sin\theta \nonumber\\
&= \Big[ (1-r) \sin(\theta+\theta_0) \cos\theta - \sin\theta_0 \Big] \sin(\theta-\theta_0).
\end{align}
As a consequence, $g$ can be rewritten as:
\begin{align}
g &= \sin(\theta-\theta_0) \Big[ (1-r^2)\cos\theta_0\cos^3\theta - \cos\theta_0\cos\theta + (1-r) \cos\theta\sin\theta\sin(\theta+\theta_0) - \sin\theta_0\sin\theta \Big] \nonumber\\
&= \sin(\theta-\theta_0) \Big[ \Big( (1-r^2)\cos^2\theta - 1 + (1-r)\sin^2\theta \Big) \cos\theta_0\cos\theta + \Big( (1-r)\cos^2\theta - 1 \Big) \sin\theta_0\sin\theta \Big] \nonumber\\
&= \sin(\theta-\theta_0) \Big[ \Big( -r^2\cos^2\theta -r\sin^2\theta \Big) \cos\theta_0\cos\theta + \Big( -r\cos^2\theta - \sin^2\theta \Big) \sin\theta_0\sin\theta \Big].
\end{align}
In the end, we find the following factorization for $g$:
\begin{align}
g = \sin(\theta-\theta_0) \Big( -r\cos^2\theta - \sin^2\theta \Big) \Big( r\cos\theta_0\cos\theta + \sin\theta_0\sin\theta \Big),
\end{align}
which allows to deduce that $g$ is zero if and only if $\theta = \theta_0 + k\pi$ or $\theta_1 + k\pi$ ($k \in \mathbb{Z}$), where $\theta_1$ is defined as:
\begin{align}
\label{EQUATPreuvLemmeColinExpreSnCosThet1}
\sin\theta_1 = -\frac{r\cos\theta_0}{\sqrt{r^2\cos^2\theta_0 + \sin^2\theta_0}} \hspace{3mm} \text{and} \hspace{3mm} \cos\theta_1 = \frac{\sin\theta_0}{\sqrt{r^2\cos^2\theta_0 + \sin^2\theta_0}} \cdotp
\end{align}
We observe now that the function $f$ is $\pi$-periodic in $\theta$. Since there is exactly one of each of the angles of the form $\theta_0+k\pi$ and $\theta_1+k\pi$ in the interval $[\frac{\pi}{2},\frac{3\pi}{2}[$ (this interval being chosen because we consider only pre-collisional configurations, that is such that $v \cdot \omega \leq 0$), decomposing the interval $[\frac{\pi}{2},\frac{3\pi}{2}[$ into three intervals delimited respectively by $\frac{\pi}{2}$, $\theta_0 + k_1\pi$, $\theta_1 + k_2\pi$ ($k_1,k_2 \in \mathbb{Z}$ being chosen such that $\theta_0 + k_1\pi,\theta_1 + k_2\pi \in [\frac{\pi}{2},\frac{3\pi}{2}[$) we deduce that the function $f$ is monotone on each of these three intervals.\\
\newline
We observed already that $f(\theta_0) = 1$, because:
\begin{align}
f(\theta_0) = \frac{1 - 2(1+r)\cos^2\theta_0 + (1+r)^2 \cos^2\theta_0}{1 - (1-r^2) \cos^2\theta_0} = 1.
\end{align}
In the case when $\theta = \theta_1$ we have:
\begin{align}
f(\theta_1) &= \frac{1 - (1+r)\cos^2\theta_0 - (1+r)\cos^2\theta_1 + (1+r)^2 \cos\theta_0\cos\theta_1 \Big(\cos\theta_1\cos\theta_0 + \sin\theta_1\sin\theta_0\Big)}{\sqrt{ \Big[ 1 - (1-r^2)\cos^2\theta_0 \Big] \cdot \Big[ 1 - (1-r^2)\cos^2\theta_1 \Big] }} \nonumber\\
&\hspace{-10mm}= \frac{1 - (1+r)\cos^2\theta_0 - (1+r) \displaystyle{\frac{\sin^2\theta_0}{r^2\cos^2\theta_0 + \sin^2\theta_0}} + (1+r)^2 \displaystyle{\frac{\cos^2\theta_0 \sin^2\theta_0 - r \cos^2\theta_0 \sin^2\theta_0 }{r^2\cos^2\theta_0 + \sin^2\theta_0}} }{\sqrt{ \Big[ 1 - (1-r^2)\cos^2\theta_0 \Big] \cdot \Big[ \displaystyle{ \frac{r^2\cos^2\theta_0 + \sin^2\theta_0 - (1-r^2)\sin^2\theta_0}{r^2\cos^2\theta_0 + \sin^2\theta_0}} } \Big] } \nonumber\\
&\hspace{-10mm}= \frac{r^2 \cos^2\theta_0 + \sin^2\theta_0 - (1+r)r^2\cos^4\theta_0 - (1+r)\cos^2\theta_0\sin^2\theta_0 - (1+r)\sin^2\theta_0 + (1+r)^2(1-r)\cos^2\theta_0\sin^2\theta_0}{r \Big(r^2 \cos^2\theta_0 + \sin^2\theta_0 \Big)} \nonumber\\
&\hspace{-10mm}= \frac{r\cos^2 \theta_0 - (1+r)r\cos^4\theta_0 - (1+r)r \cos^2\theta_0\sin^2\theta_0 - \sin^2\theta_0}{r^2\cos^2\theta_0 + \sin^2\theta_0} = \frac{r \cos^2\theta_0 - (1+r)r\cos^2\theta_0 - \sin^2\theta_0}{r^2\cos^2\theta_0 + \sin^2\theta_0} = -1.
\end{align}
As a consequence, the global extrema of $f$ are reached exactly at $\theta = \theta_0$, and $\theta = \theta_1$, and there are no other local extrema on the whole interval $[\frac{\pi}{2},\frac{3\pi}{2}[$. In between two consecutive global extrema, which are necessarily distinct (the extrema $1$ is followed by $-1$ and vice versa), $f$ is monotone.\\
We will now assume that:
\begin{align}
\big\vert f(\theta) - 1 \big\vert \leq \delta,
\end{align}
for $\delta$ small enough, and we will characterize such angles $\theta$. More precisely, we choose a positive real number $\mu_1 > 0$ such that:
\begin{align}
\big\vert r \cos \theta_0 \cos\theta + \sin\theta_0 \sin\theta \big\vert \geq \frac{r}{2} \hspace{3mm} \forall\, \theta \in [\theta_0-\mu_1,\theta_0+\mu_1].
\end{align}
We also choose a positive real number $\mu_2 > 0$ such that:
\begin{align}
\mu_2 \leq \min(\mu_1,\frac{\pi}{2}) \hspace{8mm} \text{and} \hspace{8mm} \frac{x}{2} \leq \sin(x) \hspace{3mm} \forall\, x \in [0,\mu_2],
\end{align}
and we finally define:
\begin{align}
\delta_0 = C(r) \int_0^{\mu_2} \frac{x}{2} \dd x \hspace{5mm} \text{with} \hspace{5mm} C(r) = \frac{(1+r)^2}{2r^2}.
\end{align}
We consider now a positive number $\delta$ smaller than $\delta_0$. By continuity, there exist $\theta_\delta^\pm$ such that:
\begin{align}
\label{EQUATPreuvLemmeColinDefinThetaDelta}
1 - f(\theta_\delta^\pm) = \delta \hspace{3mm} \text{and} \hspace{3mm} \theta_\delta^- < \theta_0 < \theta_\delta^+.
\end{align}
We denote by $\theta_\delta^\pm$ the closest solutions to $\theta_0$ of \eqref{EQUATPreuvLemmeColinDefinThetaDelta}. We now estimate the distance between $\theta_0$ and $\theta_\delta^\pm$. We will treat in detail only the case of $\theta_\delta^-$, since the case of $\theta_\delta^+$ can be obtained with exactly the same arguments. We consider then the angle $\theta_m = \theta_0-\mu_2$, so that we have:
\begin{align}
1 - f(\theta_m) = \int_{\theta_m}^{\theta_0} \partial_\theta f(\theta) \dd \theta.
\end{align}
Turning to the derivative $f'$, we observe that on the interval $[\theta_m,\theta_0]$ we have:
\begin{align}
\big\vert \partial_\theta f \big\vert &= \frac{(1+r)^2 \big\vert \sin(\theta-\theta_0) \big(-r\cos^2\theta-\sin^2\theta\big) \big(r\cos\theta_0\cos\theta+\sin\theta_0\sin\theta\big) \big\vert}{\Big[1 - (1-r^2)\cos^2\theta_0\Big]^{1/2} \cdot \Big[1 - (1-r^2)\cos^2\theta\Big]^{3/2}} \nonumber\\
&\geq \frac{(1+r)^2}{r^4} \cdot r \cdot \frac{r}{2} \sin(\theta_0-\theta)
\end{align}
because $0 \leq \theta_0 - \theta_m \leq \mu_1$. In particular, the only zero of the derivative on this interval is at $\theta = \theta_0$ because $\mu_2 \leq \frac{\pi}{2}$, and therefore $\partial_\theta f$ is non-negative on the whole interval $[\theta_m,\theta_0]$, so that:
\begin{align}
1 - f(\theta_m) \geq C(r) \int_{\theta_m}^{\theta_0} \sin(\theta_0-\theta) \dd \theta.
\end{align}
Since we have by definition $\theta_0-\theta_m = \mu_2$, we have in particular:
\begin{align}
\int_{\theta_m}^{\theta_0} \sin(\theta_0-\theta) \dd \theta = \int_0^{\mu_2} \sin(\theta) \dd \theta \geq \int_0^{\mu_2} \frac{\theta}{2} \dd \theta = \frac{\delta_0}{C(r)} \cdotp
\end{align}
Therefore, by continuity, there exists an angle $\theta_\delta^- \in [\theta_m,\theta_0]$ such that:
\begin{align}
1 - f(\theta_\delta^-) = \int_{\theta_\delta^-}^{\theta_0} \partial_\theta f(\theta) \dd \theta = \delta. 
\end{align}
In particular, we deduce that:
\begin{align}
\delta = \int_{\theta_\delta^-}^{\theta_0} \partial_\theta f(\theta) \dd \theta \geq C(r) \int_0^{\theta_0-\theta_\delta^-} \frac{\theta}{2} \dd \theta = C(r)\frac{\big(\theta_0 - \theta_\delta^-\big)^2}{4} \cdotp
\end{align}
As a consequence, $f$ being monotone between the consecutive global extrema for which $f = \pm 1$, these extrema being reach at points of the form $\theta_0 + k\pi$ when $f(\theta) = 1$ and of the form $\theta_1 + k\pi$ when $f(\theta) = -1$, we deduce that:
\begin{align}
f(\theta) \geq 1 - \delta \hspace{5mm} \text{if and only if} \hspace{5mm} \theta \in [\theta_\delta^- + k\pi,\theta_\delta^+ + k\pi] \ \text{for a certain } k \in \mathbb{Z},
\end{align}
and we have:
\begin{align}
\label{EQUATPreuvLemmeColinEstimThetaDelta}
\theta_0 - \frac{2 \delta^{1/2}}{\sqrt{C(r)}} \leq \theta_\delta^- \leq \theta_0 \leq \theta_\delta^+ \leq \theta_0 + \frac{2 \delta^{1/2}}{\sqrt{C(r)}} \cdotp
\end{align}
The same approach applied to $\theta=\theta_1$ leads to the existence of a similar family of intervals in which lies $\theta$ if $f(\theta) \leq -1+\delta$, of which the size is estimated exactly in the same way.\\
Estimate \eqref{EQUATPreuvLemmeColinEstimThetaDelta} concludes the study of the scalar product $f$ in the case when $\vert u \vert = 0$.\\
\newline
As for the third and last step of the proof, we consider the general case $\vert u \vert \neq 0$. In such a case, the vectors $v$, $\sigma$ and $\omega$ are not coplanar. In particular, the plane spanned by $v$ and $\omega$ is not the same plane as the one spanned by $v$ and $\sigma$, and the orientations of these two planes are a priori independent, so that the orientation of the plane spanned by $v$ and $\omega$ remains to be chosen. In order to make a consistent choice with the case when $\vert u \vert = 0$, we define the orientation of the plane spanned by $v$ and $\omega$ so that the sign of $b$ corresponds to the sign of $\sin\theta$ (we recall that the angle $\theta$ is the angle orientated from $v$ to $\omega$).\\
With such a choice of orientation, according to the expression \eqref{EQUATPreuvLemmeColinExprePrdScGener} the scalar product $f$ writes:
\begin{align}
\label{EQUATPreuvLemmeColinExprePrdScGener__2__}
f = \frac{1 - (1+r)\cos^2\theta_0 - (1+r) \cos^2\theta + (1+r)^2 \cos\theta_0 \cos\theta \Big( \cos\theta_0 \cos\theta + \sin\theta_0 \,\text{sgn}(\sin\theta) \sqrt{1 - \vert u \vert^2 - \cos^2\theta} \Big) }{ \sqrt{ \Big[ 1 - (1-r^2)\cos^2\theta_0 \Big] \cdot \Big[ 1 - (1-r^2) \cos^2 \theta \Big]} } \cdotp
\end{align}
Introducing the following function $g: x \in [0,1] \mapsto \displaystyle{\frac{1 - (1+r)x}{\sqrt{1 - (1-r^2)x}}}$,
and computing its two first derivatives, we find:
\begin{align}
g'(x) &= \frac{-2(1+r) \big(1 - (1-r^2)x \big) - \big(1 - (1+r)x\big) \big( -(1-r^2) \big)}{2 \big( 1 - (1-r^2)x \big)^{3/2}} \nonumber\\
&= \frac{(1+r) \Big( -2 + 2(1-r^2)x + (1-r) - (1-r^2)x \Big)}{2 \big( 1 - (1-r^2)x \big)^{3/2}} \nonumber\\
&= \frac{(1+r)^2 \big( (1-r)x - 1 \big)}{2 \big( 1 - (1-r^2)x \big)^{3/2}} \leq 0
\end{align}
and
\begin{align}
g''(x) = (1+r)^2(1-r) \frac{\big( (1-r^2) x - 1 - 3r \big)}{4\big( (r^2-1)x+1 \big)^{5/2}},
\end{align}
so we deduce then that:
\begin{align}
-1 = g(1) \leq g(x) \leq g(0) = 1 \ \forall x \in [0,1],
\end{align}
and that $g$ is decreasing and concave.\\
Collecting the results we established above, we can now conclude the third and last step of the proof of Lemma \ref{LEMMEColinearitScatt}. We recall that we will conduct the proof under the assumption that $\theta_0 \in [\frac{\pi}{2},\pi[$, without loss of generality.\\
First, we observe that the derivative of the scalar product with respect to the real variable $\vert u \vert^2$ is:
\begin{align}
\partial_{\vert u \vert^2} f = - \frac{(1+r)^2 \cos\theta_0\sin\theta_0 \, \text{sgn}(\sin\theta) \cos\theta}{ 2\sqrt{ \Big[ 1 - (1-r^2)\cos^2\theta_0 \Big] \cdot \Big[ 1 - (1-r^2) \cos^2 \theta \Big] \cdot \big[ 1 - \vert u \vert^2 - \cos^2\theta \big]} } \cdotp
\end{align}
In particular, on the interval $[\frac{\pi}{2},\frac{3\pi}{2}[$, this derivative has the sign of $\cos\theta_0\sin\theta_0$ on $[\frac{\pi}{2},\pi]$, and has the opposite sign on $[\pi,\frac{3\pi}{2}[$.\\
We observe also that in the case when $\cos\theta_0\sin\theta_0 = 0$, we have $f(\vert u \vert^2,\theta) = f(0,\theta)$, and this case follows directly from the second step of the proof.\\
Assuming then $\cos\theta_0\sin\theta_0 \neq 0$, and:
\begin{align}
\label{EQUATPreuvLemmeColinInegaFinal}
f(\vert u \vert^2,\theta) \geq 1 - \delta,
\end{align}
we consider three different cases.\\
First, let us assume that \eqref{EQUATPreuvLemmeColinInegaFinal} holds for a certain $\theta \in [\frac{\pi}{2},\pi[$. Using the sign of the derivative $\partial_{\vert u \vert^2} f$, which is negative on $[\frac{\pi}{2},\pi[$, we have:
\begin{align}
1 - \delta \leq f(\vert u \vert^2,\theta) \leq f(0,\theta),
\end{align}
and using again the second step of the proof, we deduce that $\theta$ belongs to an interval of measure smaller than $C(r) \sqrt{\delta}$.\\
We consider now a second case, when \eqref{EQUATPreuvLemmeColinInegaFinal} holds for a certain $\theta \in [\partial\pi,\frac{3\pi}{2}] \subset \, ]\pi,\frac{3\pi}{2}[$, and we assume also that:
\begin{align}
\label{EQUATPreuvLemmeColinCas_2Thet0}
\vert \cos\theta_0 \vert > \frac{\sqrt{2}}{1+r} \sqrt{\delta} \hspace{5mm} \text{and} \hspace{5mm} \vert \sin \theta_0 \vert > \frac{2 r}{1+r} \sqrt{\delta}
\end{align}
for any $0 < \delta$ small enough, that is, smaller than $\delta_0 = 1$. 
Relying on the sign of $\partial_{\vert u \vert^2} f$, which is this time non-negative because $\theta \in [\pi,\frac{3\pi}{2}]$, we deduce that $f(\vert u \vert^2,\theta)$ is bounded from above by the corresponding expression of $f$ that we can define choosing $\vert u \vert^2$ maximal, which provides:
\begin{align}
\label{EQUATPreuvLemmeColinBorne_de_D|u|^2}
f(\vert u \vert^2,\theta) &\leq \frac{1 - (1+r) \cos^2\theta_0 - (1+r) \cos^2\theta + (1+r)^2 \cos^2\theta_0 \cos^2\theta}{\sqrt{ \Big[ 1 - (1-r^2)\cos^2\theta_0 \Big] \cdot \Big[ 1 - (1-r^2) \cos^2\theta \Big]}} \nonumber\\
&\leq \frac{1 - (1+r) \cos^2\theta_0}{\sqrt{ 1 - (1-r^2)\cos^2\theta_0 }} \cdot \frac{1 - (1+r)\cos^2 \theta}{\sqrt{ 1 - (1-r^2)\cos^2\theta }}\cdot
\end{align}
Therefore, since the function $g$ takes values in $[-1,1]$, we have:
\begin{align}
\label{EQUATPreuvLemmeColinContrFinal__f_1}
f(\vert u \vert^2,\theta) = \big\vert f(\vert u \vert^2,\theta) \big\vert \leq \left\vert \frac{1 - (1+r) \cos^2\theta_0}{\sqrt{ 1 - (1-r^2)\cos^2\theta_0 }} \right\vert.
\end{align}
In particular, if on the one hand $\cos^2\theta_0$ is such that the right hand side in \eqref{EQUATPreuvLemmeColinContrFinal__f_1} is non-negative, using the convexity of the function $g$, we have:
\begin{align}
\left\vert \frac{1 - (1+r) \cos^2\theta_0}{\sqrt{ 1 - (1-r^2)\cos^2\theta_0 }} \right\vert = \frac{1 - (1+r) \cos^2\theta_0}{\sqrt{ 1 - (1-r^2)\cos^2\theta_0 }} &\leq 1 - \frac{(1+r)^2}{2} \cos^2\theta_0 \nonumber\\
&< 1 - \delta
\end{align}
using in the last inequality the assumption \eqref{EQUATPreuvLemmeColinCas_2Thet0} on the cosine of $\theta_0$. On the other hand, if $\cos^2\theta_0$ is such that $g(\cos^2\theta_0)$ is negative, we have:
\begin{align}
\left\vert \frac{1 - (1+r) \cos^2\theta_0}{\sqrt{ 1 - (1-r^2)\cos^2\theta_0 }} \right\vert &= - \frac{1 - (1+r) \cos^2\theta_0}{\sqrt{ 1 - (1-r^2)\cos^2\theta_0 }} \nonumber\\
&\leq 1 + \frac{(1+r)^2}{4r^2} \big( \cos^2\theta_0 - 1 \big) \nonumber\\
&\leq 1 - \frac{(1+r)^2}{4r^2} \sin^2\theta_0 \nonumber\\
&< 1 - \delta,
\end{align}
using in the first inequality that we have:
\begin{align}
-g(x) \leq 1 + \frac{(1+r)^2}{2r}(x-1)
\end{align}
for any $x \in [0,1]$ such that $g(x)$ is negative (because the line $y = -1 + (1-x)\frac{(1+r)^2}{2r}$ intersects the first axis at $x_0 = \frac{1+r^2}{(1+r)^2}$, and we have $g(x_0) = \frac{1-r}{\sqrt{1+r^2}} > 0$), and using in the last inequality the assumption \eqref{EQUATPreuvLemmeColinCas_2Thet0} on $\sin\theta_0$.
As a consequence, we see that, regardless the sign of $g(\cos^2\theta_0)$, when \eqref{EQUATPreuvLemmeColinCas_2Thet0} holds, there exists no angle $\theta \in [\pi,\frac{3\pi}{2}]$ such that \eqref{EQUATPreuvLemmeColinInegaFinal} can hold.\\
Turning now the last of the three cases we consider, we assume that \eqref{EQUATPreuvLemmeColinInegaFinal} holds, together with the fact that:
\begin{align}
\label{EQUATPreuvLemmeColinCas_3Thet0}
\vert \cos\theta_0 \vert \leq  \frac{\sqrt{2}}{1+r} \sqrt{\delta} \hspace{5mm} \text{or} \hspace{5mm} \vert \sin \theta_0 \vert \leq \frac{2 r}{1+r} \sqrt{\delta}
\end{align}
for any $\delta \leq 1$. In this case, we see that if we assume that the angle $\theta$ is such that:
\begin{align}
\big\vert \cos\theta \big\vert > \frac{\sqrt{2}}{1+r}\sqrt{\delta} \hspace{5mm} \text{and} \hspace{5mm} \big\vert \sin\theta \big\vert > \frac{2r}{1+r}\sqrt{\delta},
\end{align}
that is, up to exclude two intervals in $\theta$ of measure at most $C(r)\sqrt{\delta}$, we deduce using the same arguments as in the previous case that:
\begin{align}
f(\vert u \vert^2,\delta) \leq \left\vert \frac{1 - (1+r)\cos^2\theta}{\sqrt{1 - (1-r^2) \cos^2\theta}} \right\vert < 1 - \delta.
\end{align}
As a consequence, in the case when \eqref{EQUATPreuvLemmeColinCas_3Thet0} holds, \eqref{EQUATPreuvLemmeColinInegaFinal} cannot hold, except if $\theta$ belongs to one of two intervals of length at most $C(r)\sqrt{\delta}$.\\
\newline
The proof of Lemma \ref{LEMMEColinearitScatt} can now be concluded, since except if the angle $\theta$ defined as $\cos(\theta) = \frac{v}{\vert v \vert} \cdot \omega$ belongs to a finite family of intervals of length at most $C(r)\sqrt{\delta}$, we have $f(\vert u \vert^2,\theta) < 1-\delta$. This condition defines a subset of the unit sphere $\mathbb{S}^{d-1}$ of measure $C(d) \sqrt{\delta}$. The proof of Lemma \ref{LEMMEColinearitScatt} is now complete.
\end{proof}

\noindent
It is also possible to describe in more details the behaviour of the function $f$ defined by the expression \eqref{EQUATPreuvLemmeColinExprePrdScGener__2__}. More precisely, we have the following result.

\begin{propo}[Behaviour of the scalar product \eqref{EQUATPreuvLemmeColinExprePrdSc__1__} in the general case $\vert u \vert \neq 0$]
\label{PROPOPlus_de_Detailssur_f}
Without loss of generality, let us assume that $\theta_0 \in [\frac{\pi}{2},\pi[$. If we assume in addition that $\cos\theta_0\sin\theta_0 \neq 0$, then the scalar product \eqref{EQUATPreuvLemmeColinExprePrdSc__1__}, also describe by the expression \eqref{EQUATPreuvLemmeColinExprePrdScGener__2__}, has a global minimum on the interval $[\partial\pi,\frac{3\pi}{2}] \subset [\pi,\frac{3\pi}{2}]$, with $\partial\pi$ defined as $\cos^2(\partial\pi) = 1 - \vert u \vert^2$, and the function is decreasing between $\partial\pi$ and the abscissa of its global minimum, and increasing between this abscissa and $\frac{3\pi}{2}$.\\
In addition we have:
\begin{align}
f(\partial\pi) = \frac{1 - (1+r)\cos^2\theta_0}{\sqrt{1 - (1-r^2)\cos^2\theta_0}} \cdot \frac{-r + (1+r) \vert u \vert^2}{\sqrt{ 1 - (1-r^2)\big( 1 - \vert u \vert^2 \big) }} \hspace{5mm} \text{and} \hspace{5mm} f(\frac{3\pi}{2}) = \frac{1 - (1+r)\cos^2\theta_0}{\sqrt{1 - (1-r^2)\cos^2\theta_0}}\cdotp
\end{align}
\end{propo}

\begin{proof}
We start with considering the derivative of $f$ with respect to $\theta$ in the general case when $\vert u \vert \neq 0$. Without generality, we will focus on the particular case when $\theta_0 \in [\frac{\pi}{2},\pi[$, and when $\theta \in ]\pi,\frac{3\pi}{2}[$. In the present case, we have $\text{sgn}(\sin\theta) = -1$. All the other cases can be studied in the same way. A direct but tedious computation enables to obtain that the derivative of the expression \eqref{EQUATPreuvLemmeColinExprePrdScGener__2__} of $f$ writes:
\begin{align}
\partial_\theta f = \frac{(1+r)^2\sin\theta \Big( 1 + (r^2-1)\cos^2\theta_0 \Big)}{\sqrt{1 - \vert u \vert^2 - \cos^2\theta} \Big( \big[ 1-(1-r^2)\cos^2\theta_0 \big] \cdot \big[ 1-(1-r^2)\cos^2\theta \big]\Big)^{3/2}} F(\vert u \vert,\theta_0,\theta)
\end{align}
with
\begin{align}
F(\vert u \vert,\theta_0,\theta) &= \sqrt{1 - \vert u \vert^2 - \cos^2\theta} \Big[ \big( 1 - (1+r)\cos^2\theta_0 \big)\cos\theta + \big( -1+r + (1-r^2)\cos^2\theta_0 \big)\cos^3\theta\Big] \nonumber\\
&\hspace{5mm}+ (\vert u \vert^2 - 1) \cos\theta_0\sin\theta_0 + 2 \cos\theta_0\sin\theta_0 \cos^2\theta + (r^2-1) \cos\theta_0\sin\theta_0 \cos^4\theta.
\end{align}
The zeros of the derivative $\partial_\theta f$ correspond exactly to the zero of $F$. Although obtaining a simple expression of such zeros seems impossible, we observe that if $\theta$ is a zero of $F$, then $X = \cos^2\theta$ is a zero of the quartic polynomial $P$ defined as:
\begin{align}
\label{EQUATPreuvLemmeColinPolynDeriv}
P(X) = C^2 \big[ 1 - \vert u \vert^2 - X \big] X \big( 1 + (r-1) X \big)^2 - \cos^2\theta_0\sin^2\theta_0 \Big( 1 - \vert u \vert^2 - 2X + (1-r^2) X^2 \Big)^2,
\end{align}
with $C = 1 - (1+r)\cos^2\theta_0$. We will now estimate the maximal number of zeros of $P$ in the interval $[0,1]$.\\
We observe first that the polynomial $P$ satisfies the following properties:
\begin{align}
P(X) \leq 0 \hspace{5mm} \forall\, X \leq 0 \hspace{5mm} \text{and} \hspace{5mm} \forall\, X \geq 1.
\end{align}
In addition, except if $\vert u \vert = 0$, which corresponds to a set of angles $\omega$ of zero Lebesgue measure, we have $P(X) < 0\ \forall\, X \geq 1$. Similarly, we have always $P(X) < 0\ \forall\, X < 0$.\\
If we assume now that there are three distinct zeros or more to the polynomial $P$ in $[0,1]$, because of the sign of $P$ outside $[0,1]$ we deduce that there are at least four distinct zeros, or if the number of zeros is exactly three, then at least one of them has to be a root of multiplicity at least $2$. As a consequence, using Rolle's theorem, we deduce that there exist at least three zeros (counted with their multiplicity) of the derivative $P'$ of the polynomial $P$ in the interval $[0,1]$.\\
Then, since:
\begin{align}
P'(X) &= - C^2 X \big( 1 + (r-1)X \big)^2 + C^2 \big[ 1 - \vert u \vert^2 - X \big] \big( 1 - (r-1)X \big)^2 \nonumber\\
&\hspace{5mm}+ 2 C^2(r-1) \big[ 1 - \vert u \vert^2 - X \big] X \big( 1 + (r-1)X \big) \nonumber\\
&\hspace{5mm}- 2\cos^2\theta_0\sin^2\theta_0 \big( -2 + 2(1-r^2)X \big) \Big( 1- \vert u \vert^2 - 2X + (1-r^2) X^2 \Big),
\end{align}
we observe that:
\begin{align}
P'\big( \frac{1}{1-r} \big) &= 4 \cos^2\theta_0\sin^2\theta_0 \Big( 1 - \frac{1-r^2}{1-r} \Big) \Big( 1 - \vert u \vert^2 - \frac{2}{1-r} + \frac{1-r^2}{1-r}\Big) \nonumber\\
&= - 4 \cos^2\theta_0 \sin^2\theta_0 r \Big( - \vert u \vert^2 - r \frac{(1+r)}{(1-r)} \Big) \geq 0.
\end{align}
Taking into account the limits of $P$ as $X \rightarrow \pm \infty$, we deduce that there exists another root to the derivative $P'$ outside the interval $[0,1]$.\\
We obtain therefore a contradiction: $P$ being a quartic polynomial, $P'$ is a cubic polynomial and has therefore at most three real roots counted with their multiplicity. We deduce then that $P$ has at most two real roots in the interval $[0,1]$.\\
Back to the expression of $F$, we observe that when $\theta = \partial \pi$ (where $\partial\pi$ denotes the angle such that $1 - \vert u \vert^2 - \cos^2\theta = 0$, which delimits the interval of definition of $F$ in $[\pi,\frac{3\pi}{2}]$) we have:
\begin{align}
F(\vert u \vert,\theta_0,\partial\pi) = \cos\theta_0\sin\theta_0 \big[ \vert u \vert^2 - 1 + 2 + r^2 - 1\big] \leq 0,
\end{align}
keeping in mind that $\theta_0 \in [\frac{\pi}{2},\pi[$, so that $\sin\theta_0 > 0$ and $\cos\theta_0 \leq 0$. We have also:
\begin{align}
F(\vert u \vert,\theta_0,\frac{3\pi}{2}) = \big( \vert u \vert^2 - 1 \big) \cos\theta_0\sin\theta_0 \geq 0.
\end{align}
In the case when $\vert u \vert \neq 1$ and $\cos\theta_0\sin\theta_0 \neq 0$ we have even:
\begin{align}
F(\vert u \vert,\theta_0,\partial\pi) < 0 \hspace{5mm} \text{and} \hspace{5mm} F(\vert u \vert,\theta_0,\frac{3\pi}{2}) > 0.
\end{align}
We prove now that if $F$ has a double root, in the sense that $F$ and its derivative $\partial_\theta F$ vanish together, then so does $P$. Indeed, if $F$ and $\partial_\theta F$ vanish at the same value $\widetilde{\theta}$, there exists a differentiable function $G$ such that:
\begin{align}
F(\theta) = \big( \cos\theta - \cos\widetilde{\theta} \big) G\big( \cos\theta \big) \ \forall\, \theta, \hspace{5mm} \text{and} \hspace{5mm} G\big( \cos\widetilde{\theta} \big) = 0. 
\end{align}
$F$ has the form:
\begin{align}
F ( \theta ) = \cos\theta R\big( \cos^2\theta \big) P_1\big( \cos^2\theta \big) + P_2\big( \cos^2\theta \big)
\end{align}
where $P_1$ and $P_2$ are two polynomials, and $R$ is a function such that its square is a polynomial. By definition, the polynomial $P$ introduced in \eqref{EQUATPreuvLemmeColinPolynDeriv} is equal to:
\begin{align}
P(X) = X R^2(X)P_1^2(X)-P_2^2(X).
\end{align}
Therefore, if $F$ and $\partial_\theta F$ both vanish at $\widetilde{\theta}$, we have:
\begin{align}
P(\cos^2\theta) &= F\big( \cos\theta \big) \Big( \cos\theta R\big( \cos^2\theta \big) P_1\big( \cos^2\theta \big) + P_2\big( \cos^2\theta \big) \Big) \nonumber\\
&= \big( \cos\theta - \cos\widetilde{\theta} \big) G\big( \cos\theta \big) \Big( \cos\theta R\big( \cos^2\theta \big) P_1\big( \cos^2\theta \big) + P_2\big( \cos^2\theta \big) \Big).
\end{align}
This last expression proves that $P$ and its derivative also vanish at $\cos^2\widetilde{\theta}$, so that $P$ would have a double root.\\
\newline
We are now in position to deduce that $F$ has exactly one root in the case when $\cos\theta_0\sin\theta_0 \neq 0$. Let us assume that $F$ has at least two distinct roots. First, if $F$ has three distinct roots or more, we obtain a contradiction because we would deduce that $P$ has three distinct roots or more, which is not possible.\\
If now we assume that $F$ has exactly two distinct roots, necessarily at least one of these roots is a double root, because of the signs of $F$ at the boundary of its interval of definition. We deduce then that $P$ would have two distinct roots, one having a multiplicity at least two. However, in the case when $P$ has at least two distinct roots (hence, exactly two), the respective multiplicities of these roots have to be equal to $1$. Indeed, if the two roots have both a respective multiplicity $k_1,k_2$ larger than one, the derivative of $P$ would have two roots with respective multiplicity $(k_1-1)$ and $(k_2-1)$. Applying Rolle's theorem, $P'$ vanishes also at another point, between the two distinct roots of $P$. Using once again the fact that $P'$ vanishes at a certain point above $\frac{1}{1-r}$ we obtain a contradiction, because the non-trivial cubic polynomial $P'$ would have at least four roots, counted with their multiplicities. Finally, if one root has multiplicity $1$, while the other has a multiplicity $k > 1$, on the one hand we obtain the same contradiction by counting the roots if $k \geq 3$. On the other hand, if the multiplicity $k$ of the second root is exactly $2$, we obtain a contradiction by sign considerations: the limits of $P$ cannot be both $-\infty$ in such a case.\\
\newline
Therefore, we have proved that $F$ has exactly one root in the interval $[\pi,\frac{3\pi}{2}[$, provided that $\cos\theta_0\sin\theta_0 \neq 0$. Relying on the signs of $F(\vert u \vert,\theta_0,\partial\theta)$ and $F(\vert u \vert,\theta_0,\frac{3\pi}{2})$, we can deduce that there exists an angle $\theta_\text{min}$ such that:
\begin{align}
F(\theta) < 0 \ \forall\, \theta \in [\partial\pi,\theta_\text{min}[ \hspace{5mm} \text{and} \hspace{5mm} F(\theta) > 0 \ \forall\, \theta \in\ ]\theta_\text{min},\frac{3\pi}{2}]. 
\end{align}
In conclusion, we proved that if $\cos\theta_0\sin\theta_0 \neq 0$ and $\vert u \vert \neq 0$, then:
\begin{align}
f(\vert u \vert,\theta_0,\theta) \leq \max\Big(f(\vert u \vert,\theta_0,\partial\pi),f(\vert u \vert,\theta_0,\frac{3\pi}{2})\Big).
\end{align}
We conclude now the proof of the proposition by evaluating the two boundaries values of the scalar product $f$, at $\theta = \partial\pi$ and $\frac{3\pi}{2}$ respectively. On the one hand we have, relying on the identity $\cos^2(\partial\pi) = 1 - \vert u \vert^2$:
\begin{align}
f(\partial \pi) &= \frac{1 - (1+r)\cos^2\theta_0 - (1+r)\big( 1 - \vert u \vert^2 \big) + (1+r)^2 \cos^2\theta_0 \big( 1 - \vert u \vert^2 \big)}{ \sqrt{ \Big[ 1 - (1-r^2)\cos^2\theta_0 \Big] \cdot \Big[ 1 - (1-r^2) \big( 1 - \vert u \vert^2 \big) \Big]} } \nonumber\\
&= \frac{-r +r(1+r)\cos^2\theta_0 + (1+r)\big[ 1 - (1+r)\cos^2\theta_0 \big] \vert u \vert^2}{ \sqrt{ \Big[ 1 - (1-r^2)\cos^2\theta_0 \Big] \cdot \Big[ 1 - (1-r^2) \big( 1 - \vert u \vert^2 \big) \Big]} } \nonumber\\
&= \frac{1 - (1+r) \cos^2\theta_0}{\sqrt{ 1 - (1-r^2)\cos^2\theta_0 }} \cdot \frac{-r + (1+r) \vert u \vert^2}{\sqrt{ 1 - (1-r^2) \big( 1 - \vert u \vert^2 \big) }}\cdotp
\end{align}
On the other hand, we find:
\begin{align}
f(\frac{3\pi}{2}) = \frac{1 - (1+r)\cos^2\theta_0}{\sqrt{ 1 - (1-r^2)\cos^2\theta_0 }} \cdotp
\end{align}
The proof of Proposition \ref{PROPOPlus_de_Detailssur_f} is complete.
\end{proof}

\subsection{Estimate of the measure of the dynamical tube}

\begin{proof}[Proof of Lemma \ref{LEMMEEstimMesurTube_Dynam}]
To prove Lemma \ref{LEMMEEstimMesurTube_Dynam}, we will decompose the set $\mathcal{T}_\varepsilon$ into the three following subsests, that are not disjoint. We define:
\begin{align}
T_1 &= \big( [x_1,x_2] + \overline{B(0,\varepsilon)} \big) \cap \{ y \in \mathbb{R}^d\ /\ (y-x_2)\cdot(x_2-x_1) \leq 0\} \nonumber\\
&= \big\{ y \in \mathbb{R}^d\ /\ \exists\ \lambda \in [0,1],\, z \in \overline{B(0,\varepsilon)} \ \text{such that}\ y = \lambda x_1 + (1-\lambda) x_2 + z,\ \text{and} (y-x_2)\cdot(x_2-x_1) \leq 0 \big\}.
\end{align}
$T_1$ is the ``first part'' of the tube $\mathcal{T}_\varepsilon$. It is constituted of a portion of a cylinder, together with a half sphere centered on $x_1$. Similarly, we introduce:
\begin{align}
T_2 &= \big( [x_2,x_3] + \overline{B(0,\varepsilon)} \big) \cap \{ y \in \mathbb{R}^d\ /\ (y-x_2)\cdot(x_2-x_3) \leq 0\} \nonumber\\
&= \big\{ y \in \mathbb{R}^d\ /\ \exists\ \lambda \in [0,1],\, z \in \overline{B(0,\varepsilon)} \ \text{such that}\ y = \lambda x_2 + (1-\lambda) x_3 + z,\ \text{and} (y-x_2)\cdot(x_2-x_3) \leq 0 \big\},
\end{align}
and we denote finally by $B$ the ball:
\begin{align}
B = \overline{B(x_2,\varepsilon)}.
\end{align}
We have $\mathcal{T}_\varepsilon = T_1 \cup B \cup T_2$, where the three subsets $T_1$, $B$, $T_2$ intersect each other pairwise.\\
Relying the inclusion-exclusion formula, we have:
\begin{align}
\big\vert \mathcal{T}_\varepsilon \big\vert = \vert T_1 \vert + \vert T_2 \vert + \vert B \vert - \vert T_1 \cap B \vert - \vert T_2 \cap B \vert - \vert T_1 \cap T_2 \vert + \vert T_1 \cap T_2 \cap B \vert.
\end{align}
The objective is to prove that $\big\vert \mathcal{T}_\varepsilon \big\vert \leq \vert T_1 \vert + \vert T_2 \vert$, which would conclude the proof of the lemma. Since by the inclusion-exclusion formula we have also:
\begin{align}
\big\vert B \cap (T_1 \cup T_2) \big\vert = \big\vert \big( B \cap T_1 \big) \cup \big( B \cap T_2 \big) \big\vert = \vert B \cap T_1 \vert + \vert B \cap T_2 \vert - \vert B \cap T_1 \cap T_2 \vert,
\end{align}
we deduce:
\vspace{-5mm}
\begin{align}
&\vert B \vert - \vert T_1 \cap B \vert - \vert T_2 \cap B \vert - \vert T_1 \cap T_2 \vert + \vert T_1 \cap T_2 \cap B \vert \nonumber\\
&\hspace{5mm}= \vert B \vert - \big\vert B \cap \big( T_1 \cup T_2 \big) \big\vert - \vert T_1 \cap T_2 \vert.
\end{align}
Finally, decomposing $
\vert B \vert = \big\vert B \cap \big( T_1 \cup T_2 \big) \big\vert + \big\vert B \cap T_1^c \cap T_2^c \big\vert$ and observing by symmetry that $\big\vert B \cap T_1^c \cap T_2^c \big\vert = \big\vert B \cap T_1 \cap T_2 \big\vert$ (in this last step we used the fact $T_1$ and $T_2$ both contain a half ball centered on $x_2$), we conclude:
\begin{align}
\vert B \vert - \vert T_1 \cap B \vert - \vert T_2 \cap B \vert - \vert T_1 \cap T_2 \vert + \vert T_1 \cap T_2 \cap B \vert \leq 0.
\end{align}
concluding the proof of Lemma \ref{LEMMEEstimMesurTube_Dynam}.
\end{proof}

\end{appendices}

\bigskip

\noindent\textbf{Acknowledgements.} 
The authors gratefully
acknowledges the support by the project PRIN 2022 (Research Projects of National Relevance)
- Project code 202277WX43.

\normalsize

\vspace{1.5cm}

\def\adresse{
\begin{description}

\item[T.~Dolmaire:] { 
Dipartimento di Ingegneria e Scienze\\ dell'Informazione e Matematica (DISIM),\\ Università degli Studi dell'Aquila, \\ 67100  L'Aquila, Italy \\  E-mail: \texttt{theophile.dolmaire@univaq.it}} 

\item[A. Nota:] {Gran Sasso Science Institute,\\ Viale Francesco Crispi 7, \\ 67100 L’Aquila, Italy  \\
E-mail: \texttt{alessia.nota@gssi.it}}

\end{description}
}

\adresse

\end{document}